\documentclass[acmsmall,screen,nonacm]{acmart}

\usepackage{ifthen}

\usepackage{refcount}

\usepackage{amsfonts}

\usepackage{mathtools}

\usepackage[only, llbracket, rrbracket, fatsemi, Yleft]{stmaryrd}

\usepackage{shuffle}

\usepackage{amsthm}
\usepackage{thmtools,thm-restate}

\usepackage{tikz}
\usetikzlibrary{arrows, calc, shapes, snakes, automata, fit, positioning, intersections}
\tikzset{>=stealth'} %

\tikzstyle{every picture} = [style=semithick]
\tikzstyle{every node}    = [font=\small]
\tikzstyle{every state}   = [thick, minimum size=1mm, inner sep=2pt]

\tikzstyle{initial}   = [initial   by arrow, initial   text=, initial   distance=4mm]
\tikzstyle{accepting} = [accepting by arrow, accepting text=, accepting distance=4mm]

\usepackage{algorithm}
\usepackage{algpseudocodex}
\algrenewcommand\algorithmicrequire{\textbf{Input:}}
\algrenewcommand\algorithmicensure{\textbf{Output:}}

\usepackage{subcaption}

\usepackage[capitalise, nameinlink, noabbrev]{cleveref}

\AtEndPreamble{%
\theoremstyle{acmplain}
\newtheorem{fact}[theorem]{Fact}
\theoremstyle{acmdefinition}
\newtheorem{remark}[theorem]{Remark}
}

\usepackage[textsize=footnotesize]{todonotes}

\newcommand{\equaldef}{\coloneqq}%

\newcommand{\setN}{\mathbb{N}}
\newcommand{\setZ}{\mathbb{Z}}
\newcommand{\setQ}{\mathbb{Q}}

\renewcommand{\vec}[1]{{\mathbf #1}}
\newcommand{\norm}[1]{{\mathopen{\|}#1\mathclose{\|}}}
\newcommand{\scalar}[2]{#1\cdot#2}
\newcommand{\Min}[2]{\operatorname{Min}_{#2}(#1)}

\newcommand{\con}[1]{\operatorname{Con}(#1)}

\newcommand{\per}[1]{\operatorname{Per}(#1)}

\newcommand{\lin}[1]{\operatorname{Lin}(#1)}
\newcommand{\clin}[1]{\overline{\operatorname{Lin}}(#1)}
\renewcommand{\dim}[1]{\operatorname{dim}(#1)}

\newcommand{\cone}[1]{\con{#1}}
\newcommand{\spanper}[1]{\per{#1}}
\newcommand{\vechull}[1]{\operatorname{Vec}(#1)}

\newcommand{\B}{\mathcal{B}}
\newcommand{\IB}{\mathcal{S}}

\newcommand{\src}[1]{\operatorname{src}(#1)}
\newcommand{\tgt}[1]{\operatorname{tgt}(#1)}
\newcommand{\Runs}[1]{\operatorname{Runs}(#1)}
\newcommand{\IRuns}[1]{\operatorname{IRuns}(#1)}
\newcommand{\reach}[2][]{\operatorname{Reach}_{#1}(#2)}
\newcommand{\Inv}{\vec{I}}
\newcommand{\bvasrel}[1]{\xrightarrow{#1}}
\newcommand{\steprel}[2][]{\xhookrightarrow{#1}_{#2}}

\newcommand{\edge}[1][]{\twoheadrightarrow_{#1}}
\newcommand{\edgestar}[1][]{\stackrel{*}{\twoheadrightarrow}_{#1}}

\newcommand{\post}[2]{\operatorname{Post}^*_{#1}(#2)}
\newcommand{\rank}[1]{\operatorname{rank}(#1)}
\newcommand{\source}[1]{\src{#1}}
\newcommand{\ra}[1]{\operatorname{arity}(#1)}
\newcommand{\act}[1]{\operatorname{act}(#1)}
\newcommand{\ball}[1]{\operatorname{Ball}_{#1}}

\newcommand{\dir}[1]{\operatorname{dir}(#1)}

\newcommand{\V}{\mathcal{V}}

\newcommand{\enc}[1]{\langle #1 \rangle}
\newcommand{\uprun}[2]{#1{\uparrow}#2}
\newcommand{\runcat}{\star}

\newcommand{\pre}[2]{\operatorname{Pre}^*_{#1}(#2)}

\newcommand{\fo}[1]{\textrm{FO}(#1)}

\newcommand{\extract}[3]{\operatorname{extract}_{#1}(#2,#3)}

\newcommand{\transformer}[2]{\stackrel{#1,#2}{\curvearrowright}}

\makeatletter
\AddToHook{cmd/appendix/before}{\def\cref@section@alias{appendix}}
\makeatother

\begin{document}

\title{%
  Solving the Reachability Problem for Branching Vector Addition Systems
  via Semilinear Inductive Invariants
}

\author{Clotilde Bizière}
\orcid{0009-0003-6469-1170}
\affiliation{%
  \institution{LaBRI, University of Bordeaux, CNRS, Bordeaux INP}
  \city{Talence}
  \country{France}
}
\affiliation{%
  \institution{Institute of Informatics, University of Warsaw}
  \city{Warsaw}
  \country{Poland}
}

\author{Jérôme Leroux}
\orcid{0000-0002-7214-9467}
\affiliation{%
  \institution{LaBRI, University of Bordeaux, CNRS, Bordeaux INP}
  \city{Talence}
  \country{France}
}

\author{Grégoire Sutre}
\orcid{0009-0004-3839-0005}
\affiliation{%
  \institution{LaBRI, University of Bordeaux, CNRS, Bordeaux INP}
  \city{Talence}
  \country{France}
}

\begin{abstract}
  In this paper, we solve the reachability problem for branching vector addition systems (BVAS), a long standing open problem. Our approach is based on semilinear inductive invariants. More precisely, we prove that if a configuration of a BVAS is not reachable, then there exists an inductive invariant, given as a semilinear set, that does not contain this configuration. Based on this property, we deduce a very simple (enumerative) algorithm solving the reachability problem for BVAS.
\end{abstract}

\maketitle

\section{Introduction}\label{sec:introduction}
\paragraph{Context.}

Branching vector addition systems (BVAS) are a computational model for the distribution of additive resources through branching structures.
The model can be traced back to the work of Rambow~\cite{DBLP:conf/acl/Rambow94} in computational linguistics. To overcome the limitations of context-free grammars as a model of natural-language syntax, Rambow introduced multiset-valued linear index grammars (MV-LIG), in which nonterminals carry multisets of resources that are distributed among the children of a derivation node. When studying the language-emptiness problem for such grammars, the left-to-right order of derivation trees becomes irrelevant. The resulting abstraction is precisely a BVAS, and language emptiness can equivalently be viewed as a reachability problem.

A decade later, the same model was independently rediscovered in two seemingly unrelated contexts. Verma and Goubault-Larrecq~\cite{DBLP:journals/dmtcs/VermaG05} used BVAS to study a class of equational tree automata arising in the analysis of cryptographic protocols, obtaining decidability results for a lossy variant of the model. Around the same time, de Groote, Guillaume, and Salvati~\cite{GGS04} proposed BVAS reachability as a natural automata-theoretic reformulation of provability in multiplicative exponential linear logic (MELL). In this setting, BVAS executions capture the flow of resources through proof trees.

Since then, BVAS and closely related models have appeared in a variety of areas of theoretical computer science, including logics over data trees~\cite{DBLP:conf/pods/BojanczykDMSS06,jacquemard:hal-00769249,DBLP:conf/fossacs/AbriolaFF17}, timed pushdown systems~\cite{DBLP:conf/lics/ClementeLLM17}, verification of concurrent systems~\cite{DBLP:journals/toplas/BouajjaniE13}, and the semantics of higher-order functional languages~\cite{DBLP:conf/esop/Cotton-BarrattM17}.

Formally, a BVAS
is a finite sequence $\B \coloneqq (\vec{\Delta}_1, \ldots, \vec{\Delta}_r)$ of finite subsets of $\setZ^d$.
Each set $\vec{\Delta}_n$ in this sequence contains actions of arity $n$.
Intuitively,
such an action $\vec{v} \in \vec{\Delta}_n$
may be viewed as the rewriting rule
$\mathtt{x}_1, \ldots, \mathtt{x}_n \rightarrow \vec{v} + \mathtt{x}_1 + \cdots + \mathtt{x}_n$
with formal parameters $\mathtt{x}_1, \ldots, \mathtt{x}_n$.
Configurations of $\B$ are vectors in $\setN^d$, and
executions are finite trees labeled by configurations
in which every internal node is obtained by summing its $n \ge 1$ children and adding an action from $\vec \Delta_n$.
The reachability problem asks,
given a BVAS, a target configuration and a finite set of initial configurations,
whether there exists an execution with root labeled by the target and leaves labeled by initial configurations.
The goal of this paper is to provide a solution to this long standing open problem.

\paragraph{Decidability of the VAS Reachability Problem.}

Vector addition systems (VAS), or equivalently Petri nets, are obtained as the special case where $r=1$. Executions then degenerate into sequences of configurations, and reachability asks whether a given target configuration can be obtained from an initial one by repeatedly applying additive actions from $\vec{\Delta}_1$.

VAS are a fundamental model for concurrency, and their reachability problem is one of the landmark decidability results in the theory of infinite-state systems. First posed in the late 1960s, it was shown decidable in the early 1980s through the works of Mayr~\cite{Mayr84} and Kosaraju~\cite{Kosaraju}. The resulting proof was later refined by several authors, including Lambert~\cite{Lambert}, and became known as the KLM decomposition. Despite decades of subsequent work, the combinatorial principles underlying this decomposition remained only partially understood for a long time. One indication of this difficulty is that the first meaningful complexity upper bounds were obtained only much later~\cite{LS15} in 2015, culminating a few years later in the proof that VAS reachability is Ackermann-complete~\cite{DBLP:conf/lics/LerouxS19,DBLP:conf/focs/Leroux21,CO22}.

A conceptually different approach was introduced by Leroux in the 2010s~\cite{DBLP:conf/popl/Leroux11,Turing-100:Vector_Addition_Systems_Reachability}. He showed that every unreachable configuration can be separated from the initial configurations by a semilinear (i.e., Presburger-definable) inductive invariant. Here, an inductive invariant is simply a set containing the initial configurations and closed under the actions of the VAS. This immediately yields a decision procedure by running in parallel one semi-algorithm enumerating executions and another enumerating semilinear sets until one is found to be an inductive invariant excluding the target.

\paragraph{Challenges of BVAS Reachability.}

Despite these successes for VAS, the decidability of BVAS reachability has remained open for more than thirty years. Existing decidability and complexity results are restricted to small dimensions~\cite{GollerHLT16, DBLP:conf/icalp/FigueiraLLMS17,DBLP:conf/mfcs/BiziereHLS25,DBLP:conf/fossacs/BiziereLS26}, to lossy variants of the model~\cite{DBLP:journals/dmtcs/VermaG05,DBLP:journals/jcss/DemriJLL13,LazicS15,DBLP:conf/concur/MajumdarW13}, and to bounded counters~\cite{DBLP:conf/concur/MazowieckiP19}.
At the same time, neither of the two classical approaches to VAS reachability appears to extend naturally to branching executions.

The KLM decomposition appears particularly difficult to generalize. Its original presentation relies on a delicate combinatorial analysis of linear executions, and its deeper structure 
only emerged through much later work on well-quasi-orderings and ideals~\cite{LS15}. At present, even the right formulation of a KLM-style decomposition for BVAS remains elusive.

The situation is yet more encouraging for the invariant-based approach. 
Leroux proved that VAS reachability sets enjoy strong geometric properties: they are almost semilinear and can therefore be closely approximated by semilinear sets. These properties were recently extended to BVAS~\cite{DBLP:conf/fossacs/BiziereLS26}.
Crucially, these approximations are only available from above: while VAS reachability sets admit semilinear overapproximations whose error has strictly smaller dimension, they do not admit analogous semilinear underapproximations in general.

Leroux circumvents this limitation by growing simultaneously two semilinear sets. The first one, denoted by $\vec S$, initially contains the initial configurations. The second one, denoted by $\vec T$, initially contains an unreachable target configuration that one wishes to exclude from the final invariant. Throughout the construction, no execution leads from a configuration in $\vec S$ to a configuration in $\vec T$.
At each step, $\vec T$ is enlarged using the complement of a semilinear over-approximation of $\post{}{\vec S}$, the set of configurations reachable from $\vec S$. Symmetrically, $\vec S$ is enlarged using the complement of a semilinear over-approximation of the backward reachability set $\pre{}{\vec T}$, i.e., the set of configurations from which one can reach $\vec T$. The crucial observation is that $\pre{}{\vec T}$ is itself a VAS reachability set, namely the reachability set of $\vec T$ in the VAS obtained by replacing each action with its opposite. Consequently, the same geometric machinery can be used to approximate both $\post{}{\vec S}$ and $\pre{}{\vec T}$.

This symmetry disappears in the branching setting, where there is no satisfactory analogue of backward reachability. More precisely, there is no operator $\text{Pre}^*$ satisfying
\[
\vec S \cap \pre{}{\vec T} \neq \emptyset
\quad\Longleftrightarrow\quad
\post{}{\vec S}\cap \vec T \neq \emptyset,
\]
since BVAS executions witnessing $\post{}{\vec S}\cap \vec T\neq \emptyset$ may require several configurations from $\vec S$.

Thus, the first challenge in extending the inductive-invariant approach to BVAS is to construct the invariant using only forward reasoning, i.e., without relying on $\text{Pre}^*$ or on a second set $\vec T$. Instead, one must grow a single set of configurations $\vec S$, while ensuring that it never contains the unreachable target, despite having access only to over-approximations.
This is, however, not the only obstacle.
The tree-shaped structure of BVAS executions introduces additional challenges compared to VAS.
For instance,
given two sets of configurations $\vec{A}_1$ and $\vec{A}_2$ that are each closed under the actions of a VAS,
their union $\vec{A}_1 \cup \vec{A}_2$ is also closed.
This property no longer holds for BVAS as branching executions may mix configurations from $\vec{A}_1$ and from $\vec{A}_2$.

\paragraph{Contributions.}

We solve the reachability problem for BVAS through a novel forward-only invariant construction. Our approach relies on several new results and constructions:
\begin{itemize}
\item
  Our invariant is constructed iteratively. Starting from the empty set $\vec A$, we progressively add configurations until $\vec A$ becomes an inductive invariant.
Throughout the construction, the set $\vec A$ excludes the unreachable target and is maintained as an \emph{attractor}, a new notion related to inductive invariants that we introduce for BVAS to address the mix-approximation problem.
An attractor is a set of configurations such that
every execution with one leaf in the attractor and all other leaves in the attractor or reachable, has its root in the attractor.
Inductive invariants are attractors, and
attractors containing the initial configurations are inductive invariants.
Moreover,
the union of an attractor and the reachability set is an attractor.
\item
  We introduce an abstract graph $\mathcal{G}$ that captures executions outside the current attractor $\vec{A}$. Nodes of $\mathcal{G}$ are pairs $(\rho,\vec{C})$ where $\rho$ is an execution and $\vec{C}$ is a finitely-generated cone. The concretization of such a node is the set of executions larger than $\rho$ (for a natural extension of the known well-quasi-order on VAS executions) that increase the target of $\rho$ by a vector in $\vec{C}$. Edges of $\mathcal{G}$ are defined by a standard $\exists\exists$-abstraction of the BVAS semantics.
  We use the abstract graph $\mathcal{G}$ to extract semilinear approximations of the reachable configurations that avoid the mix-approximation problem.
  To do so,
  we impose an \emph{homogeneity condition} that enforces the cone $\vec{C}$ of a node to only depend on its strongly connected component.
\item
  We show that a part of the reachable configurations of a bottom strongly connected component $\Gamma$ of $\mathcal{G}$ can be captured by enlarging the attractor $\vec{A}$. Our construction is based on a new model, called well-structured vector addition systems (WSVAS), that naturally extends VAS with infinite sets of actions, while still preserving the well-quasi-order and amalgamation property on executions. We use WSVAS to simulate BVAS executions composed of a main branch with, on the side of this branch, configurations that are reachable or in $\vec{A}$.
\item
  Once the attractor $\vec{A}$ has been enlarged, the abstract graph $\mathcal{G}$ is updated by removing $\Gamma$ and by adding new nodes corresponding to executions in the concretization of some nodes in $\Gamma$ that are still not captured by the enlarged attractor. In order to preserve the homogeneity property of $\mathcal{G}$, we establish a geometric decomposition result, that we dub \emph{face stripping theorem}. This theorem provides a way to decompose the difference of two particular semilinear sets into a sequence of semilinear sets that respects the order induced by the reachability relation of the BVAS.
  Moreover,
  the updated abstract graph is smaller than $\mathcal{G}$ for a natural well-founded relation.
\end{itemize}
We iteratively enlarge the attractor and update the abstract graph until the latter becomes empty
(which is bound to happen).
At that point, the attractor is guaranteed to be an inductive invariant that excludes the unreachable target.
This allows us to conclude that the reachability problem for BVAS is decidable.

\paragraph{Related work.}

Our result settles a long-standing open problem simultaneously arising in the theory of counter systems and in the proof theory of linear logic.

Understanding which extensions of VAS preserve decidability of reachability has been a recurring theme in the theory of infinite-state systems. Some extensions remain decidable despite their increased expressive power. This includes VAS equipped with a restricted form of zero tests, whose reachability problem was shown decidable in 2008~\cite{REINHARDT2008239} and whose complexity was recently established to be Ackermann-complete~\cite{VASSnzCGL}, matching that of VAS. Another long-standing open problem was recently settled with the decidability of reachability in pushdown vector addition systems~\cite{GKM25}. By contrast, the decidability of reachability remains open and actively investigated for unordered data nets~\cite{datanets,data-bireach,symVASS}.

The connection between counter systems, in particular VAS, and substructural logics, in particular linear logic, has motivated a long line of research. 
At an intuitive level, both formalisms reason about resources: VAS counters in one setting, and logical assumptions in the other.
This intuition has repeatedly led to reductions relating fragments of linear logic and classes of counter systems.
For example, Lincoln, Mitchell, Scedrov, and Shankar established the undecidability of linear-logic provability through a reduction from Minsky machines~\cite{LL-undecidable}. Conversely, Kanovich derived decidability for the !-Horn fragment of MELL by a reduction to VAS reachability~\cite{Kanovich1995PetriNH}. BVAS provide a particularly striking instance of this connection through the inter-reducibility between BVAS reachability and provability in MELL, established by de Groote, Guillaume, and Salvati~\cite{GGS04}.
Our result entails that MELL provability is decidable.

\section{Preliminaries}\label{sec:prelim}
We start with basic definitions and notations that are used throughout the paper.
Let $\setZ$ denote the set of integers,
$\setN$ denote the set of natural numbers,
$\setQ$ denote the set of rational numbers, and
$\setQ_{\geq 0}$ denote the set of non-negative rational numbers.
We also introduce the sets $\setN_{>0}$ and $\setQ_{>0}$ defined as
$\setN\setminus \{0\}$ and $\setQ_{\geq 0}\setminus\{0\}$, respectively.

\paragraph{Orders}

A \emph{quasi-ordered} set (\emph{qoset} for short) is a pair $(X, \preceq)$ where
$X$ is a set and $\preceq$ is a reflexive and transitive binary relation on $S$.
We let $\Min{X}{\preceq}$ denote the set of minimal elements of a qoset $(X, \preceq)$.
Recall that a \emph{minimal} element of $(X, \preceq)$ is an element $m \in X$ such that
$x \preceq m \Rightarrow m \preceq x$ for all $x \in X$.
A \emph{partially-ordered} set (\emph{poset} for short) is a qoset $(X, \preceq)$ such that $\preceq$ is antisymmetric.
A \emph{well-quasi-ordered} set (\emph{wqo} for short) is a qoset $(X, \preceq)$ such that
every infinite sequence $x_0, x_1, x_2, \ldots$ of elements in $X$ contains an infinite subsequence $x_{i_0} \preceq x_{i_1} \preceq x_{i_2} \cdots$ (with $i_0 < i_1 < i_2 \cdots$).
A \emph{well-partially-ordered} set (\emph{wpo} for short) is a poset $(X, \preceq)$ such that $(X, \preceq)$ is a wqo.
We recall that every wqo $(X, \preceq)$ contains a finite set $M \subseteq X$ such that
$\Min{X}{\preceq} = \{x \in X \mid \exists m \in M : m \preceq x \preceq m\}$ and
$X = \bigcup_{m \in M} \{x \in X \mid m \preceq x\}$.
In particular,
every wpo $(X, \preceq)$ has finitely many minimal elements.
A map $f$ from a qoset $(X, \preceq_X)$ to a qoset $(Y, \preceq_Y)$ is \emph{order-preserving}
if for every $a, b \in X$,
it holds that $a \preceq_X b \Rightarrow f(a) \preceq_Y f(b)$.

\paragraph{Vectors}

We consider a fixed \emph{dimension} $d \in \setN_{>0}$ for the remainder of the paper.
Vectors are typeset in bold face.
Given $\vec{x} \in \setQ^d$,
we let $(\vec{x}(1), \ldots, \vec{x}(d))$ denote the vector of rational numbers defining $\vec{x}$.
The \emph{dot-product} of two vector $\vec{x},\vec{y}\in\setQ^d$ is defined as $\sum_{i=1}^d\vec{x}(i)\vec{y}(i)$.
The partial order $\leq$ on $\setQ^d$ is the component-wise extension of the usual order $\leq$ on $\setQ$
(i.e.,
$\vec{x} \leq \vec{y}$ when $\vec{x}(i) \leq \vec{y}(i)$ for all $i \in \{1, \ldots, d\}$).
The sum and the subtraction of two vectors $\vec{x} + \vec{y}$ and $\vec{x}-\vec{y}$ are defined similarly component-wise.
The sum and the subtraction operators over vectors are extended over sets $\vec{X},\vec{Y}\subseteq\setQ^d$ by $\vec{X}+\vec{Y}\coloneqq\{\vec{x}+\vec{y}\mid \vec{x}\in\vec{X}\wedge\vec{y}\in\vec{Y}\}$ and $\vec{X}-\vec{Y}\coloneqq \{\vec{x}-\vec{y}\mid \vec{x}\in\vec{X}\wedge\vec{y}\in\vec{Y}\}$.
To reduce clutter,
singleton sets $\{\vec{x}\}$ appearing in a sum of sets are shortly written $\vec{x}$
(i.e., without braces).
We let $\setQ_{\geq 0}\vec{X}$ denote the set $\{\lambda\vec{x} \mid \lambda\in \setQ_{\geq 0}\wedge \vec{x}\in\vec{X}\}$.

\subsection{Discrete Sets}

\paragraph{Periodic Sets}
A set $\vec{P}\subseteq \setZ^d$ is called a \emph{periodic set} if $\vec{0}\in \vec{P}$ and $\vec{P}+\vec{P}\subseteq \vec{P}$. The \emph{periodic set spanned} by a set $\vec{X}\subseteq \setZ^d$ is the set $\spanper{\vec{X}}$ of vectors of the form $\vec{x}_1+\cdots+\vec{x}_k$ where $k\in\setN$ and $\vec{x}_1,\ldots,\vec{x}_k\in \vec{X}$. This set is clearly the $\subseteq$-minimal periodic set containing $\vec{X}$. When $\vec{X}$ is finite, the periodic set $\spanper{\vec{X}}$ is said to be \emph{finitely-generated}. We associate with a periodic set $\vec{P}\subseteq\setZ^d$ the quasi-order $\leq_{\vec{P}}$ on $\setZ^d$ defined by $\vec{x}\leq_{\vec{P}}\vec{y}$ if $\vec{y}\in \vec{x}+\vec{P}$. We recall the following folklore result, well-known when $\vec{P}\subseteq \setN^d$.
Its proof can be found in \cref{app:lem:fgper-wqo-on-Zd}.
\begin{restatable}{lemma}{lemFgperWqoOnZd}
  \label{lem:fgper-wqo-on-Zd}
  A periodic set $\vec{P}\subseteq \setZ^d$ is finitely-generated if, and only if, $\leq_{\vec{P}}$ is a wqo on $\vec{P}$.
\end{restatable}

A \emph{linear set} is a set $\vec{L}\subseteq \setZ^d$ of the form $\vec{b}+\vec{P}$ where $\vec{b}\in\setZ^d$ and $\vec{P}\subseteq \setZ^d$ is a finitely-generated periodic set. A \emph{semilinear set} is a finite union of linear sets.
We recall that semilinear sets are effectively closed under union, intersection and complement.
Semilinear sets coincide with the sets definable in the first-order theory $\fo{\setZ,+,\leq}$, known as Presburger arithmetic~\cite{gs66}.

\paragraph{Cylindric Sets and Diagonal Relations}\label{par:cylindric}\label{par:diagonal}
Let $\vec{Q}\subseteq \setZ^d$ be a finitely-generated periodic set. A set $\vec{S}\subseteq \setZ^d$ is said to be \emph{$\vec{Q}$-cylindric} if $\vec{S}+\vec{Q}\subseteq \vec{S}$. A set $\vec{S}\subseteq \setZ^d$ is said to be \emph{finitary $\vec{Q}$-cylindric} if there exists a finite set $\vec{B}\subseteq \setZ^d$ such that $\vec{S}=\vec{B}+\vec{Q}$.
Observe that $\leq_{\vec{Q}}$ is a wqo on a $\vec{Q}$-cylindric set $\vec{S}$ if, and only if, $\vec{S}$ is finitary $\vec{Q}$-cylindric.
This entails that $\vec{Q}$-cylindric sets included in some finitary $\vec{Q}$-cylindric set are finitary $\vec{Q}$-cylindric.
A binary relation $R$ over $\setZ^d$ is said to be \emph{$\vec{Q}$-diagonal} if $(\vec{x},\vec{y})+(\vec{q},\vec{q})\in R$ for every $(\vec{x},\vec{y})\in R$ and every $\vec{q}\in \vec{Q}$.
Note that a binary relation is $\vec{Q}$-diagonal if, and only if, it is $\vec{Q'}$-cylindric for the finitely-generated periodic set $\vec{Q}'\coloneqq\{(\vec{q},\vec{q}) \mid \vec{q}\in \vec{Q}\}$.

\paragraph{Groups}
A set $\vec{G}\subseteq \setZ^d$ is called a \emph{group} if $\vec{0}\in \vec{G}$ and $\vec{G}+\vec{G}\subseteq \vec{G}$, and $-\vec{G} \subseteq \vec{G}$. The \emph{group spanned} by a set $\vec{X}\subseteq \setZ^d$ is the set of vectors of the form $z_1\vec{x}_1+\cdots+z_k\vec{x}_k$ where $k\in\setN$, $\vec{x}_1,\ldots,\vec{x}_k\in \vec{X}$, and $z_1,\ldots,z_k\in\setZ$. Let us recall that every group $\vec{G} \subseteq \setZ^d$ is finitely-generated. Notice that the group spanned by a periodic set $\vec{P}$ is equal to $\vec{P}-\vec{P}$.

\subsection{Dense Sets}
\paragraph{Cones}
A set $\vec{C}\subseteq \setQ^d$ is called a \emph{cone} if $\vec{0}\in \vec{C}$, $\vec{C}+\vec{C}\subseteq \vec{C}$, and $\setQ_{\geq 0}\vec{C}\subseteq \vec{C}$. The \emph{cone spanned} by a set $\vec{X}\subseteq \setQ^d$ is the set $\cone{\vec{X}}$ of vectors of the form $\lambda_1\vec{x}_1+\cdots+\lambda_k\vec{x}_k$ where $k\in\setN$, $\vec{x}_1,\ldots,\vec{x}_k\in \vec{X}$, and $\lambda_1,\ldots,\lambda_k\in\setQ_{\geq 0}$. This set is clearly the $\subseteq$-minimal cone containing $\vec{X}$. Notice that $\cone{\vec{P}} = \setQ_{\geq 0}\vec{P}$ for every periodic set $\vec{P}\subseteq \setZ^d$.

\paragraph{Vector Spaces}
A set $\vec{V}\subseteq\setQ^d$ is called a \emph{vector space} if $\vec{0}\in\vec{V}$, $\vec{V}+\vec{V}\subseteq \vec{V}$, and $\setQ\vec{V}\subseteq \vec{V}$. The \emph{vector space spanned} by a set $\vec{X}\subseteq \setQ^d$ is the set $\vechull{\vec{X}}$ of vectors of the form $\lambda_1\vec{x}_1+\cdots+\lambda_k\vec{x}_k$ where $k\in\setN$, $\vec{x}_1,\ldots,\vec{x}_k\in \vec{X}$, and $\lambda_1,\ldots,\lambda_k\in\setQ$. Let us recall that any vector space $\vec{V}\subseteq\setQ^d$ is spanned by a finite set $\vec{X}\subseteq \vec{V}$. The minimal cardinality of such a set $\vec{X}$ is called the \emph{dimension} of $\vec{V}$ and it is denoted as $\dim{\vec{V}}$. Let us recall that given two vector spaces $\vec{V}\subsetneq \vec{W}$ we have $\dim{\vec{V}}<\dim{\vec{W}}$. It follows that $\dim{\vec{V}}\in\{0,\ldots,d\}$ for every vector space $\vec{V}\subseteq \setQ^d$. Moreover, if $\vec{V}$ is the vector space spanned by a set $\vec{X}\subseteq \setQ^d$, then $\vec{V}$ is also spanned by a set of $\dim{\vec{V}}$ vectors in $\vec{X}$. Notice that $\vechull{\vec{C}} = \vec{C}-\vec{C}$ for every cone $\vec{C}\subseteq \setQ^d$.

\paragraph{Topological Closure}
We denote by $\ball{\vec{v},\varepsilon}$ where $\vec{v}\in\setQ^d$ and $\varepsilon\in\setQ_{>0}$ the \emph{open ball} centered on $\vec{v}$ of radius $\varepsilon$ and defined as the set of $\vec{x}\in\setQ^d$ such that $\norm{\vec{v}-\vec{x}}<\varepsilon$. A \emph{limit} of a set $\vec{X}\subseteq \setQ^d$ is a vector $\vec{v}\in\setQ^d$ such that $\ball{\vec{v},\varepsilon}\cap\vec{X}\not=\emptyset$ for every $\varepsilon\in\setQ_{>0}$. We denote by $\overline{\vec{X}}$ the set of limits of $\vec{X}$, and called the \emph{topological closure} of $\vec{X}$. When $\overline{\vec{X}}=\vec{X}$, we say that $\vec{X}$ is \emph{topologically closed}. Recall that $\overline{\vec{X}}$ is topologically closed for every set $\vec{X}\subseteq \setQ^d$. Notice that if $\vec{C}\subseteq \setQ^d$ is a cone, then $\overline{\vec{C}}$ is a cone as well. Recall also that finitely-generated cones (including vector spaces) are topologically closed. 

\subsection{Linearizations}
We recall the definition of full\footnote{Our definition of \emph{full} periodic set slightly generalizes the one given in~\cite[Definition 2.6]{DBLP:conf/concur/GuttenbergRE23} since the cone $\vec{C}$ is not required to be finitely-generated in our case.} periodic sets from~\cite{DBLP:conf/concur/GuttenbergRE23}.
A periodic set $\vec{P}\subseteq\setZ^d$ is said to be \emph{full} if $\vec{P}=\vec{G}\cap\vec{C}$ for a group $\vec{G}\subseteq\setZ^d$ and a cone $\vec{C}\subseteq\setQ^d$. Given a periodic set $\vec{P}\subseteq \setZ^d$, its \emph{linearization}~\cite[Section 8]{Turing-100:Vector_Addition_Systems_Reachability} is the full periodic set $\lin{\vec{P}}\coloneqq (\vec{P}-\vec{P})\cap \setQ_{\geq 0}\vec{P}$. Clearly, a periodic set $\vec{P}$ is full if, and only if, $\lin{\vec{P}}=\vec{P}$. The \emph{closed linearization}~\cite[Section 5]{DBLP:conf/popl/Leroux11} of a periodic set $\vec{P}\subseteq \setZ^d$ is the full periodic set $\clin{\vec{{P}}}\coloneqq(\vec{P}-\vec{P})\cap\overline{\setQ_{\geq 0}\vec{P}}$.
As a side remark,
we observe that these two linearizations are upper closure operators (i.e., order-preserving, extensive and idempotent operators) on the complete lattice of periodic subsets of $\setZ^d$ ordered by inclusion.

\begin{remark}\label{rem:full}
  A finitely-generated periodic set $\vec{P}\subseteq \setZ^d$ if full if, and only if, $\clin{\vec{P}}=\vec{P}$, since finitely-generated cones are topologically closed.
\end{remark}

\begin{example}
  The set $\{0\}\cup (2+\setN)$ is a finitely-generated periodic set that is not full while $\{(0,0)\}\cup \setN_{>0}^2$ is a full but not-finitely-generated periodic set. 
\end{example}

\section{Safety Witnesses for Branching VAS and Overview of our Approach}\label{sec:big-picture}
A \emph{branching VAS} (or \emph{BVAS} for short)
is a non-empty finite sequence $\B \coloneqq (\vec{\Delta}_1, \ldots, \vec{\Delta}_r)$ of finite subsets of $\setZ^d$.
The natural number $r \geq 1$ is the maximal arity of $\B$, and
each $\vec{\Delta}_n$ with $n \in \{1,\ldots,r\}$ provides the set of \emph{$n$-ary actions}.
A \emph{configuration} of $\B$ is a vector in $\setN^d$.
The set of runs of $\B$ is defined inductively as follows.\footnote{%
  In the literature,
  BVAS runs are classically defined as unordered rooted trees whose nodes are labeled by configurations~\cite{DBLP:journals/dmtcs/VermaG05}.
  We choose an inductive, term-based definition as it simplifies some definitions and some proofs.
  Our term-based definition is clearly equivalent (for reachability) to the classical tree-based definition.
}
A \emph{run} is a pair $\rho \coloneqq (\vec{c}, (\rho_1, \ldots, \rho_n))$ where
$\vec{c} \in \setN^d$ is called the \emph{target} of $\rho$ and
$\rho_1, \ldots, \rho_n$ is a (possibly empty) finite sequence of runs,
with $n \in \{0,\ldots,r\}$,
whose respective targets $\vec{c}_1, \ldots, \vec{c}_n$ satisfy
$\vec{c} - \sum_{j=1}^{n} \vec{c}_j\in \vec{\Delta}_n$ if $n \geq 1$.
Given such a run $\rho$,
we define
$\tgt{\rho} \coloneqq \vec{c}$,
$\ra{\rho} \coloneqq n$,
$\act{\rho} \coloneqq \vec{c} - \sum_{j=1}^{n} \vec{c}_j$, and
$\rho[j] \coloneqq \rho_j$ for each $j \in \{1, \ldots, n\}$.
The \emph{source} of a run $\rho \coloneqq (\vec{c}, (\rho_1, \ldots, \rho_n))$ is
the non-empty word $\source{\rho} \in (\setN^d)^+$
defined inductively by
$\source{\rho} \coloneqq \vec{c}$ if $n = 0$ and
$\source{\rho} \coloneqq \source{\rho_1} \cdots \source{\rho_n}$ if $n \geq 1$.
We write $\Runs{\B}$ the set of runs of $\B$.
Given a language $L\subseteq (\setN^d)^+$,
the \emph{reachability set from $L$},
written $\reach{L}$,
is the set of targets of runs $\rho$ such that $\source{\rho}\in L$.

\begin{figure}[t]
  \centering
  \begin{tikzpicture}[%
    level distance=8mm,
    level 1/.style={sibling distance=12mm}]

    \node {$(0, 0, 1)$}
    child {
      node {$(0, 1, 0)$}
    }
    child {
      node {$(0, 1, 0)$}
    }
    ;
  \end{tikzpicture}
  \hfill
  \begin{tikzpicture}[%
    level distance=8mm,
    level 1/.style={sibling distance=22mm},
    level 2/.style={sibling distance=12mm}]

    \node {$(0, 1, 1)$}
    child {
      node {$(0, 0, 1)$}
      child {
        node {$(0, 1, 0)$}
      }
      child {
        node {$(0, 1, 0)$}
      }
    }
    child {
      node {$(0, 1, 1)$}
      child {
        node {$(1, 0, 0)$}
      }
    }
    ;
  \end{tikzpicture}
  \hfill
  \begin{tikzpicture}[%
    level distance=8mm,
    level 2/.style={sibling distance=28mm},
    level 3/.style={sibling distance=12mm}]

    \node {$(1, 1, 3)$}
    child {
      node {$(2, 0, 2)$}
      child {
        node {$(0, 1, 0)$}
        child {
          node {$(0, 1, 0)$}
        }
        child {
          node {$(0, 0, 1)$}
        }
      }
      child {
        node {$(2, 1, 1)$}
        child {
          node {$(2, 2, 0)$}
        }
        child {
          node {$(0, 1, 0)$}
        }
      }
    }
    ;
  \end{tikzpicture}
  \caption{%
    The runs $\rho$ (left), $\sigma$ (middle) and $\tau$ (right) of the BVAS from \cref{exa:BVAS}.
  }
  \label{fig:BVAS}
\end{figure}
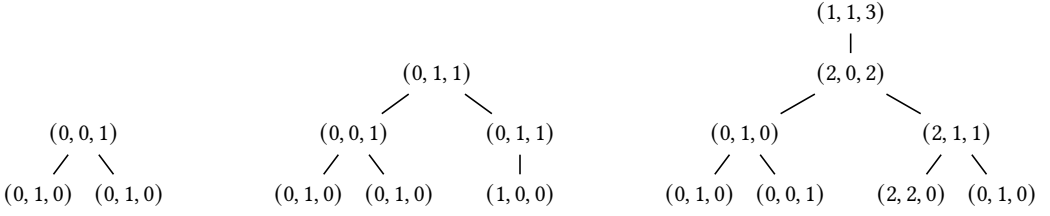

\begin{example}
  \label{exa:BVAS}
  Consider the ($3$-dimensional) BVAS $\B \coloneqq (\vec{\Delta}_1, \vec{\Delta}_2)$ where
  $\vec{\Delta}_1 \coloneqq \{(-1, 1, 1)\}$ and
  $\vec{\Delta}_2 \coloneqq \{(0, -2, 1), (0, 0, -1)\}$.
  Some runs of this BVAS are depicted in \cref{fig:BVAS}.
  The run depicted in the left is
  $\rho \coloneqq ((0, 0, 1), (\lambda, \lambda))$ with
  $\lambda \coloneqq ((0, 1, 0), ())$.
  The run $\lambda$ may be thought of as a (shared) ``leaf'' of the run $\rho$.
  The run depicted in the middle is
  $\sigma \coloneqq ((0, 1, 1), (\rho, ((0, 1, 1), \ldots)))$.
  Let $\tau$ be the run depicted in the right.
  The targets of these runs are
  $\tgt{\rho} = (0, 0, 1)$ and $\tgt{\sigma} = (0, 1, 1)$ and $\tgt{\tau} = (1, 1, 3)$.
  Their sources are
  $\source{\rho} = (0, 1, 0) (0, 1, 0)$,
  $\source{\sigma} = (0, 1, 0) (0, 1, 0) (1, 0, 0)$, and
  $\source{\tau} = (0, 1, 0) (0, 0, 1) (2, 2, 0) (0, 1, 0)$.
  The run $\sigma$ is a witness that
  $(0, 1, 1) \in \reach{\{(1, 0, 0), (0, 1, 0)\}^+}$.
  We will show in \cref{exa:ind-invariants} that
  $(0, 1, 0) \not\in \reach{\{(1, 0, 0)\}^+}$.
  To conclude this example,
  we observe that there is a run with source
  $s \coloneqq (0, 1, 0) (0, 1, 0) (0, 0, 0)$ and target $(0, 0, 0)$,
  but there is no run with source
  $s' \coloneqq (0, 1, 0) (0, 0, 0) (0, 1, 0)$,
  even though the words $s$ and $s'$ are commutatively equivalent.
\end{example}

A \emph{VAS} is a BVAS $(\vec{\Delta}_1, \ldots, \vec{\Delta}_r)$ such that $r = 1$. The following fact extends to BVAS runs the usual monotony property of VAS runs. Its proof is obtained by immediate induction on $\rho$.
\begin{fact}
  \label{fact:monotony}
  For every $\vec{c}_1, \ldots, \vec{c}_k, \vec{x}_1, \ldots, \vec{x}_k \in \setN^d$ and
  every run $\rho$ of $\B$ with $\src{\rho} = \vec{c}_1 \cdots \vec{c}_k$,
  there exists a run $\sigma$ of $\B$ with
  $\src{\sigma} = (\vec{c}_1 + \vec{x}_1) \cdots (\vec{c}_k + \vec{x}_k)$ and
  $\tgt{\sigma} = \tgt{\rho} + \vec{x}_1 + \cdots + \vec{x}_k$.
\end{fact}

An \emph{initialized BVAS} (or \emph{IBVAS} for short)
is a pair $\IB \coloneqq (\vec{\Delta}_0, \B)$ where
$\vec{\Delta}_0 \subseteq \setN^d$ is a finite set of \emph{initial} configurations and
$\B$ is a BVAS.
An \emph{initialized run} (or \emph{irun} for short) of $\IB$ is a run $\rho$
of $\B$ such that $\source{\rho} \in \vec{\Delta}_0^+$.
We let $\IRuns{\IB}$ denote the set of iruns of $\IB$.
The \emph{reachability set} of $\IB$ is the set $\reach{\IB} \coloneqq \reach{\vec{\Delta}_0^+}$.
We also introduce the binary relation $\bvasrel{\IB}$ on configurations defined by
$\vec{x} \bvasrel{\IB} \vec{y}$ if $\vec{y} \in \reach{\vec{\Delta}_0^* \, \vec{x} \, \vec{\Delta}_0^*}$.

\smallskip

The \emph{BVAS reachability problem} takes as input an IBVAS $\IB$ and a configuration $\vec{c}$, and
determines whether $\vec{c}$ is a member of $\reach{\IB}$.
Clearly,
if $\vec{c}\in\reach{\IB}$ then there exists an irun $\rho$ such that $\tgt{\rho}=\vec{c}$.
It follows that the BVAS reachability problem is recursively enumerable.
In order to prove that the problem is decidable,
we show in this paper that if $\vec{c}\not\in\reach{\IB}$ then
there exists a semilinear inductive invariant witnessing that property.
We first recall some definitions.

\smallskip

An \emph{inductive invariant} for an IBVAS $\IB = (\vec{\Delta}_0, (\vec{\Delta}_1, \ldots, \vec{\Delta}_r))$
is a set $\Inv \subseteq \setN^d$ such that
for every $n\in \{0,\ldots,r\}$ and for every $\vec{a}\in\vec{\Delta}_n$,
the following property holds:
\[
  \vec{a}+\underbrace{\Inv+\cdots+\Inv}_{\text{$n$ times}} \quad \subseteq \quad \Inv\cup(\setZ^d\setminus \setN^d)
  \:.
\]
Notice that a set $\Inv\subseteq \setN^d$ is an inductive invariant if, and only if,
$\vec{\Delta}_0 \subseteq \Inv$ and $\reach{\Inv^+}\subseteq \Inv$.
An immediate consequence of \cref{fact:reach-shuffle} below is that
the reachability set itself is an inductive invariant
(in fact,
$\reach{\IB}$ is the $\subseteq$-least inductive invariant).
It follows that
$\vec{c}\not\in \reach{\IB}$
if, and only if,
there exists an inductive invariant $\Inv\subseteq \setN^d$ such that $\vec{c}\not\in\Inv$.
We can effectively decide whether a given semilinear set is an inductive invariant.
So the following theorem, applied with the semilinear set $\vec{\Phi}\coloneqq\setN^d\setminus\{\vec{c}\}$, immediately entails that the BVAS reachability problem is decidable. This theorem is the main result of the paper.

\begin{theorem}\label{thm:main}
  For every IBVAS $\IB$ and for every semilinear set $\vec{\Phi} \subseteq \setN^d$ such that $\reach{\IB} \subseteq \vec{\Phi}$, there exists a semilinear inductive invariant $\Inv$ for $\IB$ such that $\Inv \subseteq \vec{\Phi}$.
\end{theorem}

\begin{remark}
  It is known~\cite{DBLP:conf/fossacs/BiziereLS26} that the reachability set $\reach{\IB}$ is semilinear when $d\leq 5$. It follows that the previous theorem immediately holds with $\Inv\coloneqq\reach{\IB}$ in that case.
\end{remark}

\begin{example}
  \label{exa:ind-invariants}
  Consider the ($3$-dimensional) initialized BVAS $\IB \coloneqq (\vec{\Delta}_0, \B)$ where
  $\B$ is the BVAS of \cref{exa:BVAS} and $\vec{\Delta}_0 \coloneqq \{(1, 0, 0)\}$.
  The set $\vec{I} \coloneqq \setN^3 \setminus \{(0, 0, 0)\}$
  is easily seen to be an inductive invariant for $\IB$.
  We derive that $(0, 0, 0) \not\in \reach{\IB}$.
  But $\vec{I}$ contains $(0, 1, 0)$,
  so it is not precise enough to conclude that $(0, 1, 0) \not\in \reach{\IB}$.
  The set $\setN^3 \setminus \{(0, 1, 0)\}$ is not an inductive invariant,
  since $(0, 0, -1) + \Inv + \Inv$ contains $(0, 1, 0)$ which is not in $\Inv$.
  It is routinely checked that $\vec{J} \coloneqq \setN^3 \setminus \{(0, 0, 0), (0, 1, 0)\}$ is an inductive invariant.
  This entails, in particular, that $(0, 1, 0) \not\in \reach{\IB}$.
\end{example}

Our definition of $\reach{L}$ above is for arbitrary languages $L\subseteq (\setN^d)^+$,
but we actually only care about reachability sets from commutative languages.
Recall that a language $L$ is \emph{commutative} if $uxyv \in L \Rightarrow uyxv \in L$
for every words $u, v, x, y$.
We also recall the definition of the shuffle operator,
which provides an analogue of language concatenation in the commutative setting.
The \emph{shuffle} of two languages $K$ and $L$,
written $K \,{\shuffle}\, L$,
is the set of words $u_1 v_1 \cdots u_k v_k$ such that $u_1 \cdots u_k \in K$ and $v_1 \cdots v_k \in L$.
Notice that $K \,{\shuffle}\, L$ is commutative when $K$ and $L$ are commutative.

\begin{fact}
  \label{fact:reach-shuffle}
  The inclusion $\reach{K \,{\shuffle}\, \reach{L}^+} \subseteq \reach{K \,{\shuffle}\, L^+}$ holds
  for every languages $K, L \subseteq (\setN^d)^+$.
\end{fact}

We prove \cref{thm:main} by iteratively transforming so-called safety witnesses.
To define these safety witnesses,
we need to introduce the notions of attractors and directed iruns.
An \emph{attractor} for an IBVAS $\IB = (\vec{\Delta}_0, \B)$ is a set $\vec{A} \subseteq \setN^d$ such that
$\reach{\vec{A}^+ \,{\shuffle}\, \vec{\Delta}_0^*} \subseteq \vec{A}$.
In particular, both $\emptyset$ and $\setN^d$ are attractors.
The notion of attractor is closely related to the notion of inductive invariant.
Indeed,
every inductive invariant is an attractor, and
every attractor $\vec{A}$ that contains $\vec{\Delta}_0$ is an inductive invariant.
Notice that,
by \cref{fact:reach-shuffle},
a set $\vec{A} \subseteq \setN^d$ is an attractor if, and only if,
$\reach{\vec{A}^+ \,{\shuffle}\, \reach{\IB}^*} \subseteq \vec{A}$.
It follows that $\vec{A} \cup \reach{\IB}$ is an inductive invariant for every attractor $\vec{A}$.
The following lemma shows how to incrementally construct attractors.

\begin{lemma}
  \label{lem:incremental-attractor}
  Let $\vec{A}$ be an attractor for $\IB$ and
  $\Inv$ be an inductive invariant for $\IB$ that contains $\vec{A}$.
  For every set $\vec{X} \subseteq \setN^d$,
  the set
  $\vec{A} \cup \reach{\vec{X}^+ \,{\shuffle}\, \Inv^*}$ is an attractor for $\IB$.
\end{lemma}

\begin{proof}
  Let $\vec{A}$ and $\Inv$ be as in the lemma.
  Note that $\vec{A} \cup \vec{\Delta}_0 \subseteq \Inv$ by assumption.
  For short,
  let us write $\vec{E} \coloneqq \vec{A} \cup \reach{\vec{X}^+ \,{\shuffle}\, \Inv^*}$.
  We need to show that $\vec{E}$ is an attractor,
  i.e.,
  that $\reach{\vec{E}^+ \,{\shuffle}\, \vec{\Delta}_0^*} \subseteq \vec{E}$.
  Observe that $(\vec{Y} \cup \vec{Z})^+ = \vec{Y}^+ \cup (\vec{Y}^* \,{\shuffle}\, \vec{Z}^+)$
  for every sets $\vec{Y}, \vec{Z} \subseteq \setN^d$.
  We get,
  by distributivity of shuffle over union,
  that
  $\reach{\vec{E}^+ \,{\shuffle}\, \vec{\Delta}_0^*}$ is the union of the set
  $\reach{\vec{A}^+ \,{\shuffle}\, \vec{\Delta}_0^*}$ and the set
  $\reach{\vec{A}^* \,{\shuffle}\, \reach{\vec{X}^+ \,{\shuffle}\, \Inv^*}^+ \,{\shuffle}\, \vec{\Delta}_0^*}$.
  The first set is contained in $\vec{A} \subseteq \vec{E}$ since $\vec{A}$ is an attractor.
  According to \cref{fact:reach-shuffle},
  the second set is contained in
  $\reach{\Inv^* \,{\shuffle}\, (\vec{X}^+ \,{\shuffle}\, \Inv^*)^+}$
  since $\vec{A} \cup \vec{\Delta}_0 \subseteq \Inv$.
  The observation that
  $\Inv^* \,{\shuffle}\, (\vec{X}^+ \,{\shuffle}\, \Inv^*)^+ = \vec{X}^+ \,{\shuffle}\, \Inv^*$
  entails that $\reach{\Inv^* \,{\shuffle}\, (\vec{X}^+ \,{\shuffle}\, \Inv^*)^+} \subseteq \vec{E}$.
  We have shown that $\reach{\vec{E}^+ \,{\shuffle}\, \vec{\Delta}_0^*} \subseteq \vec{E}$.
  This concludes the proof of the lemma.
\end{proof}

\begin{example}
  \label{exa:attractors}
  For every $k \in \setN$,
  the set $\vec{A}_k \coloneqq \{(x, y, z) \in \setN^3 \mid x+y+z \geq k\}$ is
  an attractor for the initialized BVAS $\IB$ of \cref{exa:ind-invariants}.
  Indeed,
  the sets
  $\vec{\Delta}_1 + \vec{A}_k$,
  $\vec{\Delta}_2 + \vec{A}_k + \vec{A}_k$ and
  $\vec{\Delta}_2 + \vec{A}_k + \reach{\IB}$
  are all contained in $\vec{A}_k \cup (\setZ^d \setminus \setN^d)$.
  The last inclusion holds because $(0, 0, 0) \not\in \reach{\IB}$,
  see \cref{exa:ind-invariants}.
  It follows, by induction on $\rho$,
  that $\tgt{\rho} \in \vec{A}_k$ for every run $\rho$ with $\src{\rho} \in \vec{A}_k^+ \,{\shuffle}\, \vec{\Delta}_0^*$.
  This means that $\vec{A}_k$ is an attractor.
  While $\vec{A}_0$ and $\vec{A}_1$ are inductive invariants,
  the attractors $\vec{A}_2, \vec{A}_3, \ldots$ are not inductive invariants
  as they do not contain $(1, 0, 0) \in \vec{\Delta}_0$.
\end{example}

The notion of directed iruns relies
on cones (see \cref{sec:prelim}) and
on a wpo on iruns introduced in~\cite{DBLP:conf/fossacs/BiziereLS26} with slightly different notations.
Let us define the partial-order $\sqsubseteq$ on runs of a BVAS $\B$ inductively by $\alpha \sqsubseteq \beta$ if $\alpha = \beta$ or there exists $j \in \{1, \ldots, \ra{\beta}\}$ such that $\alpha \sqsubseteq \beta[j]$. %
We define the partial-order $\trianglelefteq$ on runs of a BVAS $\B$ inductively by $\alpha \trianglelefteq \beta$ if there exists $\beta' \sqsubseteq \beta$ such that
$\tgt{\alpha} \leq \tgt{\beta}, \tgt{\beta'}$,
$\ra{\alpha} = \ra{\beta'}$,
$\act{\alpha} = \act{\beta'}$, and
$\alpha[j] \trianglelefteq \beta'[j]$ for every $j$. %

\begin{example}[Continued from \cref{exa:BVAS}]
  \label{exa:partial-orders-on-runs}
  Consider the runs $\rho$, $\sigma$ and $\tau$ depicted in \cref{fig:BVAS}.
  It holds that $\rho \sqsubseteq \sigma$ and $\rho \not\sqsubseteq \tau$.
  We also have $\rho \trianglelefteq \sigma$ and $\rho \trianglelefteq \tau$.
  The first assertion holds because $\rho \sqsubseteq \sigma$ and $\tgt{\rho} \leq \tgt{\sigma}$.
  Let us explain why the second assertion holds.
  We have $\rho[1] \trianglelefteq \tau[1][1]$ since
  $\rho[1] \sqsubseteq \tau[1][1]$ and $\tgt{\rho[1]} = (0, 1, 0) = \tgt{\tau[1][1]}$.
  Similarly, we have $\rho[2] \trianglelefteq \tau[1][2]$ since
  $\rho[2] \sqsubseteq \tau[1][2]$ and $\tgt{\rho[2]} = (0, 1, 0) \leq (2, 1, 1) = \tgt{\tau[1][2]}$.
  It follows that $\rho \trianglelefteq \tau$ since
  $\tgt{\rho} \leq \tgt{\tau}, \tgt{\tau[1]}$,
  $\ra{\rho} = 2 = \ra{\tau[1]}$,
  $\act{\rho} = (0, -2, 1) = \act{\tau[1]}$, and
  $\rho[j] \trianglelefteq \tau[1][j]$ for every $j$.
  To conclude this example,
  we observe that $\sigma \ntrianglelefteq \tau$ for several reasons.
  One of them is that
  $(1, 0, 0)$ occurs in $\source{\sigma}$ and does not occur in $\source{\tau}$.
\end{example}

Our approach crucially relies on the following lemma,
which was proved in~\cite{DBLP:conf/fossacs/BiziereLS26} with slightly different notations.
For the sake of completeness,
we reprove this lemma in \cref{app:lem:ibvas-wpo-and-amalgamation}.

\begin{restatable}[\cite{DBLP:conf/fossacs/BiziereLS26}]{lemma}{lemIbvasWpoAndAmalgamation}
  \label{lem:ibvas-wpo-and-amalgamation}
  The pair $(\IRuns{\IB}, \trianglelefteq)$ is a wpo and satisfies the amalgamation property,
  i.e.,
  for every iruns $\rho \trianglelefteq \alpha, \beta$,
  there exists an irun $\sigma$ such that
  $\alpha, \beta \trianglelefteq \sigma$ and $\tgt{\rho} + \tgt{\sigma} = \tgt{\alpha} + \tgt{\beta}$.
\end{restatable}

\begin{remark}
  The partial-order $\trianglelefteq$ on runs is a wqo on the set of iruns (see \cref{lem:ibvas-wpo-and-amalgamation}), but it is not a wqo on the set of runs.
  For instance,
  in dimension one,
  the runs $\rho_n\coloneqq (n,())$ where $n\in\setN$ are incomparable for $\trianglelefteq$, i.e., $\rho_m \ntrianglelefteq \rho_n$ when $m \neq n$.
\end{remark}

Consider an IBVAS $\IB$.
For each irun $\rho$ of $\IB$,
we let $\vec{P}_\rho$ denote the periodic subset of $\setN^d$ defined by
$\vec{P}_\rho \coloneqq \{\tgt{\sigma} - \tgt{\rho} \mid \sigma \in \IRuns{\IB}\wedge \rho\trianglelefteq \sigma\}$.
The fact that $\vec{P}_\rho$ is periodic immediately follows from the amalgamation property of $\trianglelefteq$.
A \emph{directed irun} of $\IB$ is a pair $w \coloneqq (\rho_w, \vec{C}_w)$ where
$\rho_w$ is an irun of $\IB$ and $\vec{C}_w \subseteq \setQ_{\geq 0}^d$ is a finitely-generated cone.
We let $\IRuns{w}$ denote the set of iruns $\sigma$ such that $\rho_w \trianglelefteq \sigma$ and $\tgt{\sigma} \in \tgt{\rho_w} + \vec{C}_w$.
We also introduce,
for each directed irun $w \coloneqq (\rho_w, \vec{C}_w)$,
the periodic sets $\vec{P}_w, \vec{Q}_w \subseteq \setN^d$ defined by
$\vec{P}_w \coloneqq \vec{P}_{\rho_w} \cap \vec{C}_w$ and
$\vec{Q}_w \coloneqq \clin{\vec{P}_w}$.
Note that\footnote{It is understood that $\tgt{R}\coloneqq\{\tgt{\rho}\mid \rho\in R\}$ for every set of runs $R$.}
$\tgt{\IRuns{w}} = \tgt{\rho_w} + \vec{P}_w$.
The following lemma will be shown in \cref{subsec:pavas}.

\begin{lemma}\label{lem:Qfinite}
  The periodic set $\vec{Q}_w$ is finitely-generated for every directed irun $w$ of $\IB$.
\end{lemma}

We introduce a binary relation $\edge$ on directed iruns of $\IB$,
defined by $u \edge v$ if $\alpha \sqsubseteq \beta$
for some $\alpha \in \IRuns{u}$ and $\beta \in \IRuns{v}$.
A set $W$ of directed iruns induces a (directed) graph $(W, \edge[W])$,
where $\edge[W]$ is the intersection of $\edge$ and $W \times W$.
Let $\edgestar[W]$ denote the reflexive and transitive closure of $\edge[W]$.
We say that $W$ is \emph{homogeneous} when
$\vec{C}_u = \vec{C}_v$ for every $u, v \in W$ such that
$u \edgestar[W] v \edgestar[W] u$.
In that case,
we also have
$\vec{Q}_u = \vec{Q}_v$ for every $u, v \in W$ such that
$u \edgestar[W] v \edgestar[W] u$.
This property is an immediate consequence of \cref{lem:Qinclusion} below.
Assuming that $W$ is homogeneous,
for each SCC $\Gamma$ of the graph $(W, \edge[W])$,
we let $\vec{C}_\Gamma$ and $\vec{Q}_\Gamma$ denote the common $\vec{C}_\gamma$ and $\vec{Q}_\gamma$
where $\gamma$ ranges over $\Gamma$.

\begin{lemma}\label{lem:Qinclusion}
  It holds that
  $\vec{Q}_u \subseteq \vec{Q}_v$ for every directed iruns $u, v$ of $\IB$ such that
  $u \edge v$ and $\vec{C}_u \subseteq \vec{C}_v$.
\end{lemma}
\begin{proof}
  Let $u$ and $v$ be directed iruns satisfying
  $u \edge v$ and $\vec{C}_u \subseteq \vec{C}_v$.
  By definition of $\edge$,
  there exist $\alpha \in \IRuns{u}$ and $\beta \in \IRuns{v}$ such that
  $\alpha \sqsubseteq \beta$.
  It holds that $\rho_u \trianglelefteq \alpha$ and $\rho_v \trianglelefteq \beta$.
  We first show that $\vec{z} + \vec{P}_{\rho_u} \subseteq \vec{P}_{\rho_v}$ where
  $\vec{z} \equaldef \tgt{\beta} - \tgt{\rho_v}$.
  Let $\vec{x} \in \vec{P}_{\rho_u}$.
  We have
  $\vec{x} = \tgt{\sigma} - \tgt{\rho_u}$ for some irun $\sigma$ with $\rho_u \trianglelefteq \sigma$.
  As $\rho_u \trianglelefteq \alpha, \sigma$,
  we get from the amalgamation property of \cref{lem:ibvas-wpo-and-amalgamation} that
  there exists an irun $\hat{\alpha}$ such that
  $\alpha, \sigma \trianglelefteq \hat{\alpha}$ and $\tgt{\rho_u} + \tgt{\hat{\alpha}} = \tgt{\alpha} + \tgt{\sigma}$.
  Since $\alpha \sqsubseteq \beta$ and $\alpha \trianglelefteq \hat{\alpha}$,
  there exists an irun $\hat{\beta}$ that satisfies
  $\beta \trianglelefteq \hat{\beta}$ and
  $\tgt{\hat{\beta}} - \tgt{\beta} = \tgt{\hat{\alpha}} - \tgt{\alpha}$.
  This claim can formally be proved by induction on $\beta$.
  Intuitively,
  $\hat{\beta}$ is obtained from $\beta$ by replacing $\alpha$ by $\hat{\alpha}$ and
  translating the target configuration of each ``ancestor'' of $\alpha$ in $\beta$ by
  the vector $\tgt{\hat{\alpha}} - \tgt{\alpha}$.
  The resulting irun $\hat{\beta}$ satisfies $\beta \trianglelefteq \hat{\beta}$ by construction.
  It is routinely checked that
  $\vec{z} + \vec{x} = \tgt{\beta} - \tgt{\rho_v} + \tgt{\sigma} - \tgt{\rho_u} = \tgt{\hat{\beta}} - \tgt{\rho_v}$,
  hence,
  $\vec{z} + \vec{x}$ is in $\vec{P}_{\rho_v}$ since $\rho_v \trianglelefteq \beta \trianglelefteq \hat{\beta}$.
  We have shown that $\vec{z} + \vec{P}_{\rho_u} \subseteq \vec{P}_{\rho_v}$.

  Now observe that $\vec{z} = \tgt{\beta} - \tgt{\rho_v}$ is in $\vec{C}_v$
  since $\beta \in \IRuns{v}$.
  Recall also that $\vec{C}_u \subseteq \vec{C}_v$ by assumption.
  It follows that
  $\vec{z} + \vec{P}_u
  =
  \vec{z} + (\vec{P}_{\rho_u} \cap \vec{C}_u)
  \subseteq
  \vec{z} + (\vec{P}_{\rho_u} \cap \vec{C}_v)
  \subseteq
  (\vec{z} + \vec{P}_{\rho_u}) \cap \vec{C}_v
  \subseteq
  \vec{P}_{\rho_v} \cap \vec{C}_v
  =
  \vec{P}_v$.
  From $\vec{z} + \vec{P}_u \subseteq \vec{P}_v$,
  we derive firstly that
  $(\vec{P}_u - \vec{P}_u) \subseteq (\vec{P}_v - \vec{P}_v)$,
  and secondly that
  $\vec{z} + k\vec{p} \in \vec{P}_v$
  for every $\vec{p} \in \vec{P}_u$ and $k \in \setN$,
  hence,
  $\vec{P}_u \subseteq \overline{\setQ_{\geq 0} \vec{P}_v}$.
  The latter inclusion entails that
  $\overline{\setQ_{\geq 0} \vec{P}_u} \subseteq \overline{\setQ_{\geq 0} \vec{P}_v}$.
  We have shown that
  $\clin{\vec{P}_u} \subseteq \clin{\vec{P}_v}$,
  which concludes the proof of the lemma
  as $\vec{Q}_u = \clin{\vec{P}_u}$ and $\vec{Q}_v = \clin{\vec{P}_v}$.
\end{proof}

Given an IBVAS $\IB$ and a semilinear set $\vec{\Phi} \subseteq \setN^d$ such that $\reach{\IB} \subseteq \vec{\Phi}$,
a \emph{safety witness} for $\IB$ and $\vec{\Phi}$ is
a pair $(\vec{A}, W)$ where $\vec{A} \subseteq \vec{\Phi}$ is a semilinear attractor for $\IB$ and
$W$ is an homogeneous finite set of directed iruns of $\IB$ such that
for every irun $\rho \in \IRuns{\IB}$, we have $\tgt{\rho} \in \vec{A}$ or $\rho \in \IRuns{w}$ for some $w \in W$.
One may wonder about the existence of safety witnesses.
Consider the set $W_{min}$ of directed iruns $(\rho, \setQ_{\geq 0}^d)$ such that $\rho \in \Min{\IRuns{\IB}}{\trianglelefteq}$.
The set $W_{min}$ is finite by \cref{lem:ibvas-wpo-and-amalgamation},
hence,
the pair $(\emptyset, W_{min})$ is a safety witness for $\IB$ and $\vec{\Phi}$.

\smallskip

To prove \cref{thm:main},
we iteratively transform safety witnesses until we get a safety witness of the form $(\vec{A}, \emptyset)$,
in which case the semilinear attractor $\vec{A} \subseteq \vec{\Phi}$ is an inductive invariant for $\IB$.
More precisely,
we show the key property that
for every safety witness $(\vec{A}, W)$ with $W \neq \emptyset$,
there exists a safety witness $(\vec{A}', W')$ such that $W'$ is smaller than $W$
for a natural well-founded relation on finite sets of directed iruns.
This entails that a safety witness of the form $(\vec{A}, \emptyset)$ exists,
and concludes the proof of \cref{thm:main}.
The remainder of the paper is devoted to the proof of this key property.
To give a taste of the proof,
we conclude this section with an overview of the construction of $(\vec{A}', W')$ from $(\vec{A}, W)$.
This overview will be reformulated with full details in \cref{sec:wrapup}.

\smallskip

Consider a safety witness $(\vec{A}, W)$ with $W \neq \emptyset$.
Pick an arbitrary bottom SCC $\Gamma$ of the graph $(W, \edge[W])$.
We introduce the finitary $\vec{Q}_\Gamma$-cylindric set
$\vec{Y}_\Gamma \coloneqq \{\tgt{\rho_\gamma} \mid \gamma \in \Gamma\}  + \vec{Q}_\Gamma$.
Our approach is decomposed into two main steps:
\begin{enumerate}
\item
  Using \cref{lem:incremental-attractor},
  we show how to construct a finitary $\vec{Q}_\Gamma$-cylindric
  set $\vec{X} \subseteq \vec{Y}_{\Gamma}$ such that
  firstly $\vec{A}' \coloneqq \vec{A} \cup \vec{X}$ is an attractor contained in $\vec{\Phi}$,
  secondly $\vec{Y}_\Gamma \subseteq \vec{X} - \vec{Q}_\Gamma$, and
  thirdly $\bvasrel{\IB}$ has an empty intersection with $\vec{X} \times (\vec{Y}_\Gamma \setminus \vec{X})$.
\item
  We show that there exist
  some finitely-generated cones $\vec{C}_1, \ldots, \vec{C}_n$ contained in $\vec{C}_\Gamma$ and
  a disjoint decomposition\footnote{%
    \label{footnote:disjoint-decomposition}
    A disjoint decomposition of a set $X$ is a partition of $X$ that allows the emptyset, i.e., a finite sequence $X_1,\ldots,X_n$ of subsets of $X$ such that $X=X_1\cup\cdots\cup X_n$ and $X_i\cap X_j=\emptyset$ for every $i\neq j$.
  } $\vec{S}_1, \ldots, \vec{S}_n$ of $\vec{Y}_\Gamma \setminus \vec{X}$
  such that
  $\dim{\vec{C}_i-\vec{C}_i}<\dim{\vec{C}_\Gamma-\vec{C}_\Gamma}$ and
  $\vec{S}_i$ is finitary $(\vec{Q}\cap\vec{C}_i)$-cylindric
  for every $1\leq i\leq n$.
  The homogeneous set $W'$ of directed iruns is then defined as $W' \coloneqq (W\setminus \Gamma)\cup (W_1\cup\ldots W_n)$
  where each $W_i$ is a finite set of directed iruns obtained from~$\vec{S}_i$.
\end{enumerate}
By construction,
the pair $(\vec{A}', W')$ defined in this way is a safety witness and $W'$ is smaller than $W$.

\smallskip

The rest of the paper presents the details of our proof.
\cref{sec:wsvas} is devoted to the first step (construction of $\vec{A}'$) and
\cref{sec:stripping} is devoted to the second step (construction of $W'$).
Before that,
we provide in the next section some geometrical properties of periodic sets and their linearizations,
and we prove \cref{lem:Qfinite} along the way.

\section{Geometry of Branching VAS Reachability Sets}\label{sec:geo}
In \cite{DBLP:conf/popl/Leroux11}, reachability sets of initialized VAS (a VAS is a BVAS $(\vec{\Delta}_1,\ldots,\vec{\Delta}_r$) such that $r=1$) was proved to be \emph{almost semilinear}, a class of sets extending the class of semilinear sets. This result was recently extended to the class of IBVAS in~\cite{DBLP:conf/fossacs/BiziereLS26} by observing that IBVAS reachability sets are \emph{sections} of initialized VAS reachability sets. 

In \cref{sub:geolin}, we recall the class of almost semilinear sets, based on the class of definable cones and asymptotically-definable periodic sets. We explain why the two linearization operators introduced in \cref{sec:prelim} provide  ways to transfer properties from cones to periodic sets. In~\cref{subsec:pavas} we prove by structural induction on a irun $\rho$ that the periodic set $\vec{P}_\rho \coloneqq \{\tgt{\sigma} - \tgt{\rho} \mid \sigma \in \IRuns{\IB}\wedge \rho\trianglelefteq \sigma\}$ is asymptotically-definable. From this property, we derive a proof of \cref{lem:Qfinite}.

\subsection{Some Geometric Properties of Periodic Sets and Their Linearizations}\label{sub:geolin}
In \cref{sec:prelim}, we introduced the linearization and the closed linearization of a periodic set $\vec{P}\subseteq \setZ^d$ respectively by $\lin{\vec{P}}\coloneqq(\vec{P}-\vec{P})\cap\cone{\vec{P}}$ and $\clin{\vec{P}}\coloneqq(\vec{P}-\vec{P})\cap \overline{\cone{\vec{P}}}$. As $\vec{P}-\vec{P}$ is a group, it is finitely-generated. It follows that properties on cones $\cone{\vec{P}}$ and $\overline{\cone{\vec{P}}}$ can be transfered to properties on the full periodic sets $\lin{\vec{P}}$ and $\clin{\vec{P}}$.

\medskip

A cone $\vec{C}\subseteq\setQ^d$ is said to be \emph{definable}~\cite{DBLP:conf/popl/Leroux11,Turing-100:Vector_Addition_Systems_Reachability} if it can be denoted by a formula in $\fo{\setQ,+,\leq}$. It follows that finitely-generated cones are definable, and the intersection as well as the sum of two definable cones are definable. Recall~\cite[Theorem 3.5]{DBLP:conf/popl/Leroux11} that a cone $\vec{C}\subseteq \setQ^d$ is definable if, and only if, $\overline{\vec{C}\cap\vec{V}}$ is a finitely-generated cone for every vector space $\vec{V}\subseteq\setQ^d$. A periodic set $\vec{P}\subseteq \setZ^d$ is said to be \emph{asymptotically-definable}~\cite{DBLP:conf/popl/Leroux11} if the cone $\cone{\vec{P}}$ is definable. It follows that finitely-generated periodic sets are asymptotically-definable. 

\begin{remark}
  Since $\fo{\setQ,+,\leq}$ admits quantification elimination, a subset of $\setQ^d$ is definable if, and only if, it is a boolean combination of \emph{open and closed half-spaces}, i.e. sets of the form $\{\vec{x}\in\setQ^d \mid \scalar{\vec{a}}{\vec{x}}\sim 0\}$ where $\vec{a}\in\setZ^d$ and $\sim$ is either $>$ or $\geq$. One may wrongly think that a cone $\vec{C}\subseteq \setQ^d$ is definable if, and only if, $\vec{C}\setminus\{\vec{0}\}$ is a conjunction of open and closed half-spaces (as claimed in~\cite[second paragraph after Lemma 2.1]{DBLP:conf/concur/GuttenbergRE23}). For a counter-example, just pick the definable cone $\vec{C}\coloneqq\{(x,y)\in\setQ_{\geq 0}\times\setQ \mid x=0\Rightarrow y\geq 0\}$.
\end{remark}

\begin{example}
  The periodic set $\vec{P}\coloneqq \{(0,0)\}\cup \setN_{>0}^2$ is asymptotically-definable. More generally, semilinear periodic sets are asymptotically-definable (see \cref{lem:almostperiodic} below).
\end{example}

The following lemma shows that asymptotically-definable periodic sets are stable by intersection, sum, and projections.
\begin{lemma}\label{lem:coneperstable}
  For every periodic sets $\vec{P},\vec{Q}\subseteq \setZ^d$ and for every periodic set $\vec{R}\subseteq \setZ^k\times\setZ^d$ we have:
  \begin{itemize}
  \item $\cone{\vec{P}\cap\vec{Q}}=\cone{\vec{P}}\cap\cone{\vec{Q}}$,
  \item $\cone{\vec{P}+\vec{Q}}=\cone{\vec{P}}+\cone{\vec{Q}}$, and
  \item $\cone{\{\vec{y}\in\setZ^d \mid \exists \vec{x}\in\setZ^k,\, (\vec{x},\vec{y})\in\vec{R}\}}= \{\vec{y}\in\setQ^d \mid \exists \vec{x}\in\setQ^k,\, (\vec{x},\vec{y})\in\cone{\vec{R}}\}$.
  \end{itemize}
\end{lemma}
\begin{proof}
  As $\vec{P}\cap\vec{Q}$ is included in the cone $\cone{\vec{P}}\cap\cone{\vec{Q}}$, by minimality of the spanning cone, we deduce that $\cone{\vec{P}\cap\vec{Q}}\subseteq \cone{\vec{P}}\cap\cone{\vec{Q}}$. For the converse inclusions, let $\vec{c}\in \cone{\vec{P}}\cap\cone{\vec{Q}}$. There exists $n,m\in\setN_{>0}$ such that $n\vec{c}\in\vec{P}$ and $m\vec{c}\in \vec{Q}$. As $\vec{P}$ and $\vec{Q}$ are periodic, we deduce that $nm\vec{c}\in\vec{P}\cap\vec{Q}$. Hence $\vec{c}\in \cone{\vec{P}\cap\vec{Q}}$. We have prove the first equality.

  As $\vec{P}+\vec{Q}$ is included in the cone $\cone{\vec{P}}+\cone{\vec{Q}}$, by minimality of the spanning cone, we deduce that $\cone{\vec{P}+\vec{Q}}\subseteq \cone{\vec{P}}+\cone{\vec{Q}}$. For the converse inclusions, let $\vec{c}\in \cone{\vec{P}}+\cone{\vec{Q}}$. There exists $\vec{x}\in \cone{\vec{P}}$ and $\vec{y}\in\cone{\vec{Q}}$ such that $\vec{c}=\vec{x}+\vec{y}$. It follows that there exists $n,m\in\setN_{>0}$ such that $n\vec{x}\in\vec{P}$ and $m\vec{y}\in \vec{Q}$. As $\vec{P}$ and $\vec{Q}$ are periodic, we deduce that $nm\vec{x}\in\vec{P}$ and $nm\vec{y}\in\vec{Q}$. From $nm\vec{x}+nm\vec{y}\in \vec{P}+\vec{Q}$, we derive $\vec{c}\in\cone{\vec{P}+\vec{Q}}$. We have proved the second equality.

  Finally, let us introduce the projection $\pi:\setQ^k\times\setQ^d\rightarrow\setQ^d$ defined by $\pi(\vec{x},\vec{y})=\vec{y}$. As $\pi(\vec{R})$ is included in the cone $\pi(\cone{\vec{R}})$, by minimality of the spanning cone we deduce that $\cone{\pi(\vec{R})}\subseteq \pi(\cone{\vec{R}})$. For the converse inclusion, let $\vec{y}\in \pi(\cone{\vec{R}})$. There exist $\vec{x}\in\setQ^k$ such that $(\vec{x},\vec{y})\in\cone{\vec{R}}$. It follows that there exists $n\in\setN_{>0}$ such that $n(\vec{x},\vec{y})\in \vec{R}$. Hence $n\vec{y}\in \pi(\vec{R})$ and we deduce that $\vec{y}\in\cone{\pi(\vec{R})}$. We have proved\footnote{The same proof shows that $\cone{\pi(\vec{R})}=\pi(\cone{\vec{R}})$ for any linear mapping $\pi$.} the equality $\cone{\pi(\vec{R})}=\pi(\cone{\vec{R}})$. That is the third equality.
\end{proof}

As the intersection of a group (which is necessarily finitely-generated) with a definable cone is Presburger-definable, hence, semilinear, we easily derive that if $\vec{P}\subseteq\setZ^d$ is a periodic set, then $\lin{\vec{P}}$ is semilinear, if and only if, $\vec{P}$ is asymptotically-definable. Moreover, as the intersection of a group with a finitely-generated cone is a finitely-generated periodic set, we also deduce that a periodic set $\vec{P}\subseteq \setZ^d$ is asymptotically-definable if, and only if, $\clin{\vec{P}\cap\vec{V}}$ is a finitely-generated periodic set for every vector space $\vec{V}\subseteq\setQ^d$. In particular, the following result holds.
\begin{lemma}
  \label{lem:clin-of-asymp-def-is-fg}
  For every asymptotically-definable periodic set $\vec{P}\subseteq \setN^d$,
  the periodic set $\clin{\vec{P}}$ is finitely-generated.
\end{lemma}

\medskip

An \emph{almost linear set} is a set $\vec{L}\subseteq \setZ^d$ of the form $\vec{b}+\vec{P}$ where $\vec{b}\in\setZ^d$ and $\vec{P}\subseteq \setZ^d$ is an asymptotically-definable periodic set. An \emph{almost semilinear set} is a finite union of almost linear sets.
\begin{lemma}\label{lem:almostperiodic}
  Every almost semilinear periodic set is asymptotically-definable. 
\end{lemma}
\begin{proof}
  Assume that $\vec{P}\coloneqq (\vec{b}_1+\vec{P}_1)\cup (\vec{b}_n+\vec{P}_n)$ is a periodic set such that $\vec{b}_1,\ldots,\vec{b}_n\in\setZ^d$, and $\vec{P}_1,\ldots,\vec{P}_n\subseteq \setZ^d$ are asymptotically-definable periodic set. Since $\vec{P}$ is periodic, we deduce that $\vec{P}=\per{\{\vec{b}_1,\ldots,\vec{b}_n\}}+\vec{P}_1+\cdots+\vec{P}_n$. It follows that $\vec{P}$ is an asymptotically-definable periodic set from \cref{lem:coneperstable}.
\end{proof}

Let us recall from~\cite{DBLP:conf/fossacs/BiziereLS26} that the intersection of IBVAS reachability sets with semilinear sets are almost semilinear. This result can be easily extended to the relation $\bvasrel{\IB}$ of an IBVAS by using the doubling counter technique, well-known for VAS, and formally defined in the proof of the following lemma.
Recall from \cref{sec:big-picture} that ${\bvasrel{\IB}} \subseteq \setN^d \times \setN^d$ is defined by
$\vec{x} \bvasrel{\IB} \vec{y}$ if $\vec{y} \in \reach{\vec{\Delta}_0^* \, \vec{x} \, \vec{\Delta}_0^*}$.
\begin{lemma}\label{lem:almostrel}
  For every IBVAS $\IB$, the intersection of $\bvasrel{\IB}$ with a semilinear relation included in $\setN^d\times\setN^d$ is an almost semilinear relation.
\end{lemma}
\begin{proof}
  Let $\IB\coloneqq(\vec{\Delta}_0,(\vec{\Delta}_1,\ldots,\vec{\Delta}_r))$ be an IBVAS and let us prove that there exists an IBVAS $\IB'$ such that a section of $\reach{\IB'}$ is exactly $\bvasrel{\IB}$. The number of counters of $\IB'$ is $2+2d$ and vectors in $\setZ^{2+2d}$ are denoted by tuples in $\setZ^2\times \setZ^d\times\setZ^d$.

  We introduce the set $\vec{E}$ of tuples $((-1,1),\vec{e}_i,\vec{e}_i)$ and $((1,-1),\vec{e}_i,\vec{e}_i)$ where $i\in\{1,\ldots,n\}$ and $\vec{e}_i$ is the $i$th unit vector of $\setN^d$ defined by $\vec{e}_i(i)=1$ and $\vec{e}_i(j)=0$ if $j\not=i$.
  We introduce the set $\vec{T}_n\coloneqq \{((0,0),\vec{0})\}\times\vec{\Delta}_n$ for each $n\in\{0,\ldots,r\}$ and the IBVAS $\IB'\coloneqq (\vec{\Delta}_0',(\vec{\Delta}_1',\ldots,\vec{\Delta}_r'))$ where $\vec{\Delta}'_0\coloneqq \vec{T}_0\cup \{((1,0),\vec{0},\vec{0}), ((0,1),\vec{0},\vec{0})\}$, $\vec{\Delta}'_1\coloneqq  \vec{T}_1\cup \vec{E}$, and $\vec{\Delta}_n'\coloneqq\vec{T}_n$ if $n\geq 2$.
  
  Observe by structural induction on the runs of $\IB$ and $\IB'$ that for every $(a,b)\in\setN^2$ and $(\vec{x},\vec{y})\in \setN^d\times\setN^d$ we have:
  $$((a,b),\vec{x},\vec{y})\in \reach{\IB'}\Longleftrightarrow
    \begin{cases}
    \vec{x}=\vec{0}\wedge \vec{y}\in \reach{\IB} & \text{ if }(a,b)=(0,0)\\
    \vec{x} \bvasrel{\IB} \vec{y}& \text{ if }a+b=1\\
\end{cases}$$
Since the intersections of IBVAS reachability sets with semilinear sets are almost semilinear~\cite{DBLP:conf/fossacs/BiziereLS26}, we deduce that $\reach{\IB'}\cap (\{(1,0),(0,1)\}\times\setN^d\times\setN^d)$ is almost semilinear. This set is exactly $\{(1,0),(0,1)\}\times (\bvasrel{\IB})$. Finally, let us introduce the projection function $\pi:\setZ^2\times\setZ^d\times\setZ^d \rightarrow\setZ^d\times\setZ^d$ defined by $\pi((a,b),\vec{x},\vec{y})\coloneqq(\vec{x},\vec{y})$. \cref{lem:coneperstable} shows that the image by $\pi$ of an asymptotically-definable periodic set is an asymptotically-definable periodic set. We deduce that the image by $\pi$ of an almost semilinear set is an almost semilinear set. In particular, since $\pi(\{(1,0),(0,1)\}\times (\bvasrel{\IB}))$ is equal to $\bvasrel{\IB}$, we are done.
\end{proof}

Almost semilinear sets $\vec{S}\coloneqq\vec{b}_1+\vec{P}_1\cup\ldots\cup\vec{b}_n+\vec{P}_n$ can be over-approximated by semilinear sets $\vec{b}_1+\vec{Q}_1\cup\ldots\cup\vec{b}_n+\vec{Q}_n$ obtained by replacing each asymptotically-definable periodic set $\vec{P}_i$ by the finitely-generated periodic set $\vec{Q}_i\coloneqq\clin{\vec{P}_i}$. In order to recover from vectors in this over-approximation, vectors in $\vec{S}$, the following lemma will be central.
\begin{lemma}
  \label{lem:pushin}
  Let $\vec{P} \subseteq \setN^d$ be a periodic set and
  let $\vec{Q} \coloneqq \clin{\vec{P}}$.
  There exists a vector $\vec{p} \in \vec{P}$ such that
  for every $\vec{g} \in \vec{P}-\vec{P}$ and every $\vec{x} \in \vec{p} + \vec{Q}$,
  there exists $r \in \setN$ verifying
  $\vec{g} + (r + \setN)\vec{x} \subseteq \vec{P}$.
\end{lemma}

\newcommand{\cent}[1]{\operatorname{Cent}(#1)}
The previous lemma is a corollary of some geometrical results about the \emph{center} of a periodic set $\vec{P}\subseteq \setZ^d$, written $\cent{\vec{P}}$ and defined as the set of vectors $\vec{v}\in \setZ^d$ such that for every $\vec{g}\in\vec{P}-\vec{P}$ there exists $r\in\setN$ satisfying $\vec{g}+(r+\setN)\vec{v}\subseteq \vec{P}$.
\begin{lemma}\label{lem:center}
  The center of a periodic set $\vec{P}\subseteq \setZ^d$ coincides with the set of vectors in $\vec{P}-\vec{P}$ of the form $\setQ_{>0}\vec{p}_1+\cdots+\setQ_{>0}\vec{p}_n$ where $\vec{p}_1,\ldots,\vec{p}_n$ is a sequence of vectors in $\vec{P}$ spanning the vector space $\vechull{\vec{P}}$.
\end{lemma}
\begin{proof}
  Assume first that $\vec{v}$ is a vector in the center of $\vec{P}$ and let $\vec{p}_1,\ldots,\vec{p}_n$ be a sequence of vectors in $\vec{P}$ spanning the vector space $\vechull{\vec{P}}$, and let $\vec{p}\coloneqq\vec{p}_1+\cdots+\vec{p}_n$. There exists $r\in\setN$ such that $-\vec{p}+(r+\setN)\vec{v}\subseteq \vec{P}$. It follows that $(r+1)\vec{v}$ and $r\vec{v}$ are
 both in $\vec{P}$. From $\vec{v}=(r+1)\vec{v}-r\vec{v}$, we derive  $\vec{v}\in \vec{P}-\vec{P}$.
  Let $\vec{p}_{n+1}\coloneqq (r+1)\vec{v}-\vec{p}$ and observe that $\vec{v}=\frac{1}{r+1}(\vec{p}_1+\cdots+\vec{p}_{n+1})$. Since $\vec{p}_1,\ldots,\vec{p}_{n+1}$ is a sequence of vectors in $\vec{P}$ spanning $\vechull{\vec{P}}$, we are done. Now, let us assume that $\vec{v}\in\vec{P}-\vec{P}$ is such that there exists a sequence $\vec{p}_1,\ldots,\vec{p}_n$ of vectors in $\vec{P}$ spanning $\vechull{\vec{P}}$ and such that $\vec{v}\in \setQ_{>0}\vec{p}_1+\cdots+\setQ_{>0}\vec{p}_n$. Let $\vec{g}\in\vec{P}-\vec{P}$. There exists $\vec{p}\in\vec{P}$ such that $\vec{g}+\vec{p}\in\vec{P}$. Since $\vec{v}\in\vec{P}-\vec{P}$, there exists $\vec{q}\in\vec{P}$ such that $\vec{v}+\vec{q}\in \vec{P}$. As $-(\vec{p}+\vec{q})\in \vechull{\vec{P}}=\setQ\vec{p}_1+\cdots+\setQ\vec{p}_n$, there exists $m\in \setN_{>0}$ such that $-m(\vec{p}+\vec{q})\in \setZ\vec{p}_1+\cdots+\setZ\vec{p}_n$. There also exists $r\in\setN_{>0}$ such that $r\vec{v}\in \setN_{>0}\vec{p}_1+\cdots+\setN_{>0}\vec{p}_n$. By replacing $r$ by a multiple of $r$, we can also assume that $-m(\vec{p}+\vec{q})+r\vec{v}\in\setN\vec{p}_1+\cdots+\setN\vec{p}_n$. It follows that $-m(\vec{p}+\vec{q})+r\vec{v}\in\vec{P}$. We deduce that $r\vec{v}\in \vec{p}+\vec{P}$, $r\vec{v}\in \vec{q}+\vec{P}$, and $r\vec{v}\in\vec{P}$. Now, let $n\geq r^2+r$. The Euclidean divisor of $n-(r^2+r)$ by $r$ shows that there exists $z\in \setN$ and $i\in\{0,\ldots,r-1\}$ such that $n=r^2+r+zr+i$. It follows that $n\vec{v}=r(r\vec{v})+r\vec{v}+zr\vec{v}+i\vec{v}\in r\vec{q}+\vec{P}+\vec{p}+\vec{P}+\vec{P}+i\vec{v}\subseteq (r-i)\vec{q}+\vec{p}+i(\vec{v}+\vec{q})+\vec{P}\subseteq \vec{p}+\vec{P}$. From $\vec{g}+\vec{p}\in \vec{P}$ we deduce that $\vec{g}+n\vec{v}\subseteq \vec{P}$. We have proved that $\vec{g}+(r^2+r+\setN)\vec{v}\subseteq \vec{P}$. Hence $\vec{v}\in\cent{\vec{P}}$. 
\end{proof}

\begin{lemma}\label{lem:centercylindric}
  For every periodic set $\vec{P}\subseteq \setZ^d$, it holds that $\cent{\vec{P}}+\clin{\vec{P}}=\cent{\vec{P}}$.
\end{lemma}
\begin{proof}
  Since $\vec{0}\in \clin{\vec{P}}$, the inclusion $\cent{\vec{P}}+\clin{\vec{P}}\supseteq\cent{\vec{P}}$ is trivial. Let us prove the converse inclusion. We introduce the vector space $\vec{V}\coloneqq\vechull{\vec{P}}$. Let $\vec{v}\in \cent{\vec{P}}$ and $\vec{q}\in \clin{\vec{P}}$. From \cref{lem:center}, we deduce that $\vec{v}\in\vec{P}-\vec{P}$ and there exists a sequence $\vec{p}_1,\ldots,\vec{p}_n$ of vectors in $\vec{P}$ spanning $\vec{V}$ and such that $\vec{v}=\lambda_1\vec{p}_1+\cdots+\lambda_n\vec{p}_n$ for some rational numbers $\lambda_1,\ldots,\lambda_n\in\setQ_{>0}$. Since $\vec{q}\in\clin{\vec{P}}$, we deduce that $\vec{q}\in\vec{P}-\vec{P}$. It follows that $\vec{v}+\vec{q}\in\vec{P}-\vec{P}$. Since $\vec{p}_1,\ldots,\vec{p}_n$ is spanning the vector space $\vec{V}$, from Cramer's rules, there exists a sequence $\vec{a}_1,\ldots,\vec{a}_n\in\setQ^d$ such that for every $\vec{w}\in\vec{V}$, we have:
  $$\vec{w}=\sum_{i=1}^n(\scalar{\vec{a}_i}{\vec{w}})\vec{p}_i$$
  Observe that there exists $\varepsilon\in\setQ_{>0}$ such that $d\norm{\vec{a}_i}\varepsilon<\lambda_i$ for every $1\leq i\leq n$. Since $|\scalar{\vec{a}_i}{\vec{w}}|\leq d\norm{\vec{a}_i}\epsilon<\lambda_i$ for every $\vec{w}\in \ball{\vec{0},\varepsilon}$, we deduce the following inclusion:
  \begin{equation}\label{eq:ballinclusion}
  \ball{\vec{0},\varepsilon}\cap\vec{V}\subseteq (-\lambda_1,\lambda_1)\vec{p}_1+\cdots+(-\lambda_n,\lambda_n)\vec{p}_n
  \end{equation}
  As $\vec{q}\in \overline{\cone{\vec{P}}}$, there exists
   $\vec{x}\in \ball{\vec{q},\varepsilon}\cap\cone{\vec{P}}$.
    Let $\vec{w}\coloneqq\vec{q}-\vec{x}$. Since $\vec{w}\in \ball{\vec{0},\varepsilon}\cap\vec{V}$ and $\vec{v}=\lambda_1\vec{p}_1+\cdots+\lambda_n\vec{p}_n$, we deduce from \cref{eq:ballinclusion} that $\vec{v}+\vec{w}\in \setQ_{>0}\vec{p}_1+\cdots+\setQ_{>0}\vec{p}_n$. Since $\vec{x}\in\cone{\vec{P}}$, there exists $\vec{p}_{n+1}\in \vec{P}$ such that $\vec{x}\in\setQ_{>0}\vec{p}_{n+1}$. From $\vec{v}+\vec{q}=\vec{v}+\vec{w}+\vec{x}$, we derive $\vec{v}+\vec{q}\in \setQ_{>0}\vec{p}_1+\cdots+\setQ_{>0}\vec{p}_{n+1}$. \cref{lem:center} shows that $\vec{v}+\vec{q}\in \cent{\vec{P}}$.
\end{proof}

Finally, let us prove \cref{lem:pushin}. Let $\vec{P} \subseteq \setN^d$ be a periodic set and let $\vec{Q} \coloneqq \clin{\vec{P}}$. Let $\vec{p}$ be a vector obtained as a sum of vectors in $\vec{P}$ spanning the vector space $\vechull{\vec{P}}$. Lemma~\ref{lem:center} shows that $\vec{p}\in \cent{\vec{P}}$. Let $\vec{g}\in\vec{P}-\vec{P}$ and $\vec{x}\in \vec{p}+\vec{Q}$. Lemma~\ref{lem:centercylindric} shows that $\vec{x}\in\cent{\vec{P}}$. It follows that there exists $r\in\setN$ such that $\vec{g}+(r+\setN)\vec{x}\subseteq \vec{P}$. We have proved the lemma.

\subsection{Asymptotic Definability of the Periodic Sets Associated with Initialized Runs}
\label{subsec:pavas}

It was shown in \cite{DBLP:conf/fossacs/BiziereLS26} that the reachability set of any IBVAS is almost semilinear.

In this section,
we refine this result by exhibiting an almost semilinear decomposition.
More precisely,
we show that for every irun $\rho$ of an IBVAS $\IB$,
the periodic set\footnote{%
  Recall that by definition $\vec{P}_\rho = \{\tgt{\sigma} - \tgt{\rho} \mid \sigma \in \IRuns{\IB}\wedge \rho\trianglelefteq \sigma\}$.
} $\vec{P}_\rho$ is asymptotically-definable.
This property entails that the reachability set of $\IB$ is the almost semilinear set
$\bigcup_{\rho \in M} \tgt{\rho} + \vec{P}_\rho$ where $M$ is the finite set $M \coloneqq \Min{\IRuns{\IB}}{\trianglelefteq}$.
More importantly for us,
the asymptotic definability of the periodic sets $\vec{P}_\rho$ will allow us to prove \cref{lem:Qfinite}.

We introduce for a configurations $\vec{c}\in\setN^d$ the so-called \emph{transformer relation}~\cite{Turing-100:Vector_Addition_Systems_Reachability} $\transformer{\IB}{\vec{c}}$ defined over $\setN^d$ by $\vec{x}\transformer{\IB}{\vec{c}}\vec{y}$ if $\vec{c}+\vec{x}\bvasrel{\IB} \vec{c}+\vec{y}$.
\begin{lemma}\label{lem:transformerper}
  The transformer relation $\transformer{\IB}{\vec{c}}$ is an asymptotically-definable periodic relation. 
\end{lemma}
\begin{proof}
  Since $\transformer{\IB}{\vec{c}}$ is equal to $(\bvasrel{\IB}\cap ((\vec{c},\vec{c})+\setN^d\times\setN^d))-(\vec{c},\vec{c})$, we deduce from \cref{lem:almostrel} that $\transformer{\IB}{\vec{c}}$ is an almost semilinear relation.
  Clearly $\transformer{\IB}{\vec{c}}$ is transitive and reflexive. Moreover, it is periodic since $\vec{x}_1 \transformer{\IB}{\vec{c}}\vec{y}_1$ and $\vec{x}_2 \transformer{\IB}{\vec{c}}\vec{y}_2$ implies $\vec{x}_1+\vec{x}_2 \transformer{\IB}{\vec{c}}\vec{y}_1+\vec{x}_2\transformer{\IB}{\vec{c}}\vec{y}_1+\vec{y}_2$.
 It follows from \cref{lem:almostperiodic} that $\transformer{\IB}{\vec{c}}$ is asymptotically-definable.
\end{proof}

We are now ready to prove the following lemma.
\begin{lemma}\label{lem:Passymptotic}
  The periodic set $\vec{P}_\rho$ is asymptotically-definable for every irun $\rho\in\IRuns{\IB}$.
\end{lemma}
\begin{proof}
  Notice that for every irun $\rho\coloneqq(\vec{c},\rho_1,\ldots,\rho_n)$ the following equality holds:
  $$\vec{P}_\rho=\{\vec{y}\in \setN^d \mid \exists \vec{x}\in \vec{P}_{\rho_1}+\cdots+\vec{P}_{\rho_n} \mid \vec{x} \transformer{\IB}{\vec{c}}\vec{y}\}$$
  From \cref{lem:coneperstable}, we derive the following equality:
  $$\cone{\vec{P}_\rho}=\{\vec{y}\in \setQ_{\geq 0}^d \mid \exists \vec{x}\in \cone{\vec{P}_{\rho_1}}+\cdots+\cone{\vec{P}_{\rho_n}} \mid (\vec{x},\vec{y}) \in\cone{\transformer{\IB}{\vec{c}}}\}$$
  Let us recall from \cref{lem:transformerper} that $\cone{\transformer{\IB}{\vec{c}}}$ is a definable cone. So, by induction on the structure of $\rho$, the previous equality shows that $\cone{\vec{P}_\rho}$ is a definable cone. 
\end{proof}

We conclude this subsection with the proof of \cref{lem:Qfinite}.
Recall that this lemma claims that the periodic set $\vec{Q}_w$ is finitely-generated for every directed irun $w$ of an IBVAS $\IB$.
Let us also recall that,
by definition,
$\vec{P}_w = \vec{P}_{\rho_w} \cap \vec{C}_w$ and
$\vec{Q}_w = \clin{\vec{P}_w}$.
We derive that
$\cone{\vec{P}_w} = \cone{\vec{P}_{\rho_w}} \cap \vec{C}_w$ is definable,
since $\cone{\vec{P}_{\rho_w}}$ is definable and $\vec{C}_w$ is a finitely-generated cone.
It follows from \cref{lem:clin-of-asymp-def-is-fg} that $\vec{Q}_w$ is finitely-generated.
We also observe that $\tgt{\IRuns{w}} = \tgt{\rho_w} + \vec{P}_w$ is almost linear.

\section{Safe Linearization for Branching VAS}\label{sec:wsvas}
In this section,
we present the main ingredient for the first step of our proof,
namely the step where we expand the attractor
(see the overview at the end of \cref{sec:big-picture}).
This main ingredient is provided by the following technical lemma,
that we call \emph{safe linearization} lemma for BVAS.

\begin{lemma}
  \label{lem:bvas-safelin}
  Let $\Inv$ be an inductive invariant for an IBVAS $\IB \coloneqq (\vec{\Delta}_0, \B)$
  of the form $\Inv = \vec{A} \cup \reach{\IB}$ for some semilinear set $\vec{A} \subseteq \setN^d$, and
  let $\vec{b} + \vec{P} \subseteq \setN^d$ be an almost linear set.
  Define $\vec{Q} \coloneqq \clin{\vec{P}}$.
  For every semilinear set $\vec{T} \subseteq \setN^d$ such that
  $\reach{(\vec{b} + \vec{P}) \,{\shuffle}\, \Inv^*} \subseteq \vec{T} \subseteq \Inv + \vec{Q}$,
  there exists $p \in \vec{P}$ such that
  $\reach{(\vec{b} + \vec{p} + \vec{Q})^+ \,{\shuffle}\, \Inv^*} \subseteq \vec{T}$.
\end{lemma}

The remainder of this section is devoted to the proof of \cref{lem:bvas-safelin}.
Apart from this lemma,
the material presented in this section is not used in the sequel.
Firstly,
we introduce a natural extension of VAS,
shortly called WSVAS,
and prove a safe linearization lemma for this extension.
Secondly,
we show how to simulate BVAS runs
with source in $\vec{X} \,{\shuffle}\, \Inv^*$,
where $\vec{X} \subseteq \setN^d$,
by WSVAS runs with source in $\vec{X}$.
Thirdly,
we prove that $\reach{\vec{X} \,{\shuffle}\, \Inv^*} = \reach{\vec{X}^+ \,{\shuffle}\, \Inv^*}$
under a suitable condition on $\vec{X} \subseteq \setN^d$.
Lastly,
we combine these three results to prove \cref{lem:bvas-safelin}.

\subsection{Well-Structured VAS}
\label{subsec:wsvas}

In this subsection,
we prove a safe linearization lemma for an extension of VAS that permits infinitely many actions.
The definition of this extension relies on the notion of
diamond maps.
We call a map $f$ from a qoset $(X, \preceq)$ to $(\setZ^d, \leq)$
\emph{diamond}\footnote{%
  The name \emph{diamond} comes from the observation that
  if $x \sqsubseteq a, b \sqsubseteq y$ and
  $f(x) + f(y) = f(a) + f(b)$ then
  the differences
  $d_{xa} = f(a) - f(x)$, $d_{xb} = f(b) - f(x)$,
  $d_{ay} = f(y) - f(a)$ and $d_{by} = f(y) - f(b)$ satisfy
  $d_{xa} = d_{by}$ and $d_{xb} = d_{ay}$.
} when for every $x, a, b \in X$ such that $x \sqsubseteq a, b$,
there exists $y \in X$ such that $a, b \sqsubseteq y$ and
$f(x) + f(y) = f(a) + f(b)$.

\smallskip

A \emph{well-structured VAS} (or \emph{WSVAS} for short)
is a triple $\V \equaldef (T, \sqsubseteq, \delta)$ where
$(T, \sqsubseteq)$ is a wqo and
$\delta$ is an order-preserving and diamond map
from $(T, \sqsubseteq)$ to $(\setZ^d, \leq)$.
Elements of $T$ are called \emph{transitions}.
Given a transition $t \in T$,
the vector $\delta(t) \in \setZ^d$ is called the \emph{action} of $t$.
Notice that every plain VAS $\vec{A} \subseteq \setZ^d$ can be viewed as a WSVAS $(T, \sqsubseteq, \delta)$
by letting
$T$ be the finite set $\vec{A}$,
$\sqsubseteq$ be the equality relation on $\vec{A}$, and
$\delta$ be the identity function on $\vec{A}$.

\smallskip

We extend to WSVAS the usual semantic notions associated with VAS.
This extension is natural,
but we provide the details for the sake of completeness.
Let us consider a WSVAS $\V \equaldef (T, \sqsubseteq, \delta)$ for the rest of this subsection.
A \emph{configuration} of $\V$ is (as usual) a vector in $\setN^d$.
A \emph{run} of $\V$ is a finite, non-empty, alternating sequence
$\rho \equaldef (\vec{c}_0, t_1, \vec{c}_1, \ldots, t_k, \vec{c}_k)$
of configurations $\vec{c}_i$ and transitions $t_i$,
satisfying $\vec{c}_i = \vec{c}_{i-1} + \delta(t_i)$ for all $i \in \{1, \ldots, k\}$.
We call $k$ the \emph{length} of $\rho$,
$\vec{c}_0$ the \emph{source} of $\rho$, written $\src{\rho}$, and
$\vec{c}_k$ the \emph{target} of $\rho$, written $\tgt{\rho}$.
The \emph{direction} of $\rho$ is the pair $\dir{\rho} \equaldef (\src{\rho}, \tgt{\rho})$.
The set of runs of $\V$ is written $\Runs{\V}$.
Given a set $\vec{X} \subseteq \setN^d$,
we let $\reach{\vec{X}}$ denote the set of targets of runs $\rho$ such that $\source{\rho} \in \vec{X}$.

\smallskip

It is well-known that VAS runs can be equipped with a well-quasi-order~\cite{DBLP:journals/tcs/Jancar90,DBLP:conf/popl/Leroux11}
that satisfies the amalgamation property~\cite{LS15}.
This well-quasi-order can be lifted to WSVAS as follows.
We introduce the binary relation $\trianglelefteq$ on runs of $\V$ defined by
$\rho \trianglelefteq \sigma$ if
$\rho = (\vec{c}_0, t_1, \vec{c}_1, \ldots, t_k, \vec{c}_k)$ and
$\sigma = \sigma_0 u_1 \sigma_1 \cdots u_k \sigma_k$
for
some configurations $\vec{c}_0, \ldots, \vec{c}_k \in \setN^d$,
some runs $\sigma_0, \ldots, \sigma_k \in \Runs{\V}$ and
some transitions $t_1, \ldots, t_k, u_1, \ldots, u_k \in T$
such that
$\vec{c}_i \leq \src{\sigma_i}, \tgt{\sigma_i}$ for all $i \in \{0, \ldots, k\}$ and
$t_i \sqsubseteq u_i$ for all $i \in \{1, \ldots, k\}$.
By lifting to WSVAS the proof arguments for VAS~\cite{DBLP:journals/tcs/Jancar90,DBLP:conf/popl/Leroux11,LS15},
we obtain that $\trianglelefteq$ is a well-quasi-order that satisfies the amalgamation property,
see \cref{lem:wsvas-wqo-and-amalgamation} for the precise statement.
The detailed proof of this lemma is deferred to \cref{app:wsvas} to avoid disrupting the flow of the paper.

\begin{restatable}{lemma}{lemWsvasWqoAndAmalgamation}
  \label{lem:wsvas-wqo-and-amalgamation}
  The pair $(\Runs{\V}, \trianglelefteq)$ is a wqo and satisfies the amalgamation property,
  i.e.,
  for every runs $\rho \trianglelefteq \alpha, \beta$,
  there exists a run $\sigma$ such that
  $\alpha, \beta \trianglelefteq \sigma$ and $\dir{\rho} + \dir{\sigma} = \dir{\alpha} + \dir{\beta}$.
\end{restatable}

\begin{corollary}
  \label{cor:wsvas-amalgamation}
  For every runs $\rho \trianglelefteq \sigma$
  and for every $r \in \setN_{>0}$,
  there exists a run $\sigma \trianglelefteq \tau$ such that
  $\dir{\tau} = \dir{\rho} + r (\dir{\sigma} - \dir{\rho})$.
\end{corollary}
\begin{proof}
  The corollary follows from \cref{lem:wsvas-wqo-and-amalgamation} by induction on $r$.
\end{proof}

Notice that $\dir{\rho} \leq \dir{\sigma}$ for every runs $\rho \trianglelefteq \sigma$.
So an equivalent formulation of \cref{lem:wsvas-wqo-and-amalgamation} is that
the triple $(\Runs{\V}, \trianglelefteq, \dir{\cdot})$ is itself a WSVAS
(but of dimension $2d$).
We conclude this subsection with a safe linearization lemma for WSVAS.
This lemma will be used in \cref{subsec:proof:lem:bvas-safelin} to prove our safe linearization lemma for BVAS.

\begin{lemma}
  \label{lem:wsvas-safelin}
  Let $\vec{b} + \vec{P} \subseteq \setN^d$ be an almost linear set and
  define $\vec{Q} \coloneqq \clin{\vec{P}}$.
  For every semilinear set $\vec{T} \subseteq \setN^d$ such that
  $\reach{\vec{b} + \vec{P}} \subseteq \vec{T}$,
  there exists $\vec{p} \in \vec{P}$ such that
  $\reach{\vec{b} + \vec{p} + \vec{Q}} \subseteq \vec{T}$.
\end{lemma}
\begin{proof}
  We introduce,
  for every semilinear set $\vec{S} \subseteq \setN^d$,
  the binary relation $\preceq_{\vec{S}}$ on $\vec{S}$ defined by
  $\vec{x} \preceq_{\vec{S}} \vec{y}$ if $\vec{x} + \setN(\vec{y} - \vec{x}) \subseteq \vec{S}$.
  This relation need not be transitive,
  but it is \emph{almost-full}.\footnote{%
    A binary relation $R$ on a set $X$ is \emph{almost-full}~\cite{VeldmanBezem93} if
    for every infinite sequence $(x_n)_{n\in\setN}$ of elements in $X$,
    there exists $m<n$ such that $(x_m, x_n) \in R$.
    Note that a qoset $(X, \preceq)$ is a wqo if, and only if, $\preceq$ is almost-full.
  }
  Indeed,
  for every infinite sequence $\vec{x}_0, \vec{x}_1, \vec{x}_2, \ldots$ of elements in $\vec{S}$,
  there exist $m, n \in \setN$ with $m < n$ such that $\vec{x}_m \preceq_{\vec{S}} \vec{x}_n$.
  The reason is twofold.
  First,
  there is a linear set
  $\vec{L} \coloneqq \vec{b}_{\vec{L}} + \vec{Q}_{\vec{L}} \subseteq \vec{S}$
  that contains infinitely many $\vec{x}_i$.
  Second,
  $(\vec{L}, \leq_{\vec{Q}_{\vec{L}}})$ is a wqo (see \cref{par:cylindric})
  and $\vec{x} \preceq_{\vec{L}} \vec{y}$ for every
  $\vec{x}, \vec{y} \in \vec{L}$ with $\vec{x} \leq_{\vec{Q}_{\vec{L}}} \vec{y}$.

  \smallskip

  Let us now prove the lemma.
  Consider a semilinear set $\vec{T} \subseteq \setN^d$ such that
  $\reach{\vec{b} + \vec{P}} \subseteq \vec{T}$.
  According to \cref{lem:pushin},
  there exists a vector $\vec{p} \in \vec{P}$ such that
  for every $\vec{q} \in \vec{Q}$ and every $\vec{x} \in \vec{p} + \vec{Q}$,
  there exists $r \in \setN_{> 0}$ verifying
  $\vec{q} + r \vec{x} \in \vec{P}$.
  We introduce the set $\vec{S}_n \equaldef \vec{b} + n \vec{p} + \vec{Q}$
  for each $n \in \setN$.
  Note that $n \vec{p} \in \vec{P}$ for every $n \in \setN$,
  since $\vec{P}$ is periodic.
  It follows that to prove the lemma,
  we only need to show that $\reach{\vec{S}_n} \subseteq \vec{T}$ for some $n \in \setN$.

  \smallskip

  Suppose by contradiction that
  $\reach{\vec{S}_n} \not\subseteq \vec{T}$ for every $n \in \setN$.
  This means that for every $n \in \setN$,
  there exists a run $\rho_n$ from some $\vec{s}_n \in \vec{S}_n$ to some $\vec{t}_n \in (\setN^d \setminus \vec{T})$.
  We can write each $\vec{s}_n$ under the form
  $\vec{s}_n = \vec{b} + n \vec{p} + \vec{q}_n$
  for
  some $\vec{q}_n \in \vec{Q}$.

  \smallskip

  Recall that $(\Runs{\V}, \trianglelefteq)$ is a wqo by \cref{lem:wsvas-wqo-and-amalgamation}.
  Recall also that $(\vec{Q}, \leq_{\vec{Q}})$ is a wqo
  since $\vec{Q}$ is a finitely-generated periodic set
  (see \cref{par:cylindric}).
  As $\setN^d \setminus \vec{T}$ is semilinear,
  the relation $\preceq_{\setN^d \setminus \vec{T}}$ is almost-full (see above).
  So we can find an increasing pair $m, n \in \setN$ with $m < n$ such that
  $\rho_m \trianglelefteq \rho_n$,
  $\vec{q}_m \leq_{\vec{Q}} \vec{q}_n$ and
  $\vec{t}_m \preceq_{\setN^d \setminus \vec{T}} \vec{t}_n$.
  Let us define
  $\vec{y} \equaldef m \vec{p} + \vec{q}_m$ and
  $\vec{x} \equaldef \vec{q}_n - \vec{q}_m + (n-m) \vec{p}$.
  Observe that $\vec{y} \in \vec{Q}$ and
  that $\vec{x} \in \vec{p} + \vec{Q}$
  since $\vec{p} \in \vec{P} \subseteq \vec{Q}$.
  We derive that
  there exists $r \in \setN_{> 0}$ such that
  $\vec{y} + r \vec{x} \in \vec{P}$.

  \smallskip

  According to \cref{cor:wsvas-amalgamation} applied to $\rho_m \trianglelefteq \rho_n$,
  there exists a run $\tau$ such that
  $\src{\tau} = \vec{s}_m + r(\vec{s}_n - \vec{s}_m)$ and
  $\tgt{\tau} = \vec{t}_m + r(\vec{t}_n - \vec{t}_m)$.
  Observe that $\tgt{\tau} \in (\setN^d \setminus \vec{T})$ since
  $\vec{t}_m, \vec{t}_n \in (\setN^d \setminus \vec{T})$ and
  $\vec t_m \preceq_{\setN^d \setminus \vec{T}} \vec t_n$.
  Moreover,
  it is readily seen that $\src{\tau} = \vec{b} + \vec{y} + r \vec{x}$,
  hence,
  $\src{\tau} \in \vec{b} + \vec{P}$.
  This contradicts the assumption that
  $\reach{\vec{b} + \vec{P}} \subseteq \vec{T}$.
\end{proof}

\subsection{From Branching VAS to Well-Structured VAS}

This subsection shows how to simulate BVAS runs by WSVAS runs
(see \cref{lem:bvas-to-wsvas} below for a precise statement).
We first introduce additional notations to clarify what we mean by simulation.
The \emph{step} relation of a WSVAS $\V \equaldef (T, \sqsubseteq, \delta)$ is
the binary relation $\steprel{\V}$ on $\setN^d$ defined by
$\vec{x} \steprel{\V} \vec{y}$ if $\vec{y} = \vec{x} + \delta(t)$ for some $t \in T$.
Similarly,
given a BVAS $\B \equaldef (\vec{\Delta}_1, \ldots, \vec{\Delta}_r)$ and a set $\Inv \subseteq \setN^d$,
the \emph{step} relation of $\B$ and $\Inv$ is
the binary relation $\steprel{\B, \Inv}$ on $\setN^d$ defined by
$\vec{x} \steprel{\B, \Inv} \vec{y}$ if
there exist $n \in \{1, \ldots, r\}$ and $\vec{a} \in \vec{\Delta}_n$ such that
$\vec{y} = \vec{a} + \vec{x} + \sum_{j=2}^n \vec{c}_j$ for some
$\vec{c}_2, \ldots, \vec{c}_n \in \Inv$.
We let $\steprel[*]{\V}$ and $\steprel[*]{\B, \Inv}$ denote
the reflexive and transitive closure of $\steprel{\V}$ and $\steprel{\B, \Inv}$,
respectively.
Previously defined notations such as $\reach{\cdot}$ and $\trianglelefteq$ left
the underlying system implicit.
When the system under consideration is not clear from the context,
we add it as a subscript to these notations to prevent confusion.

\begin{lemma}
  \label{lem:steprel-reach}
  For every set $\vec{X} \subseteq \setN^d$,
  it holds that
  $\reach[\V]{\vec{X}} = \{\vec{y} \in \setN^d \mid \exists \vec{x} \in \vec{X} : \vec{x} \steprel[*]{\V} \vec{y}\}$
  and that
  $\reach[\B]{\vec{X} \,{\shuffle}\, \Inv^*} = \{\vec{y} \in \setN^d \mid \exists \vec{x} \in \vec{X} : \vec{x} \steprel[*]{\B, \Inv} \vec{y}\}$
  when $\reach{\Inv^+} \subseteq \Inv$.
\end{lemma}
\begin{proof}
  The first assertion follows from the easy observation that
  $\vec{x} \steprel[*]{\V} \vec{y}$
  if, and only if,
  there is a run $\rho$ of $\V$ with $\src{\rho} = \vec{x}$ and $\tgt{\rho} = \vec{y}$.
  Similarly,
  assuming that $\reach{\Inv^+} \subseteq \Inv$,
  the second assertion follows from the observation that
  $\vec{x} \steprel[*]{\B, \Inv} \vec{y}$
  if, and only if,
  there is a run $\rho$ of $\B$ with $\src{\rho} \in (\{\vec{x}\} \,{\shuffle}\, \Inv^*)$ and $\tgt{\rho} = \vec{y}$.
  Both directions of this second observation are routinely proved by induction.
\end{proof}

\begin{lemma}
  \label{lem:bvas-to-wsvas}
  For every IBVAS $\IB \coloneqq (\vec{\Delta}_0, \B)$ and semilinear set $\vec{A} \subseteq \setN^d$,
  there exists a WSVAS $\V$ such that $\steprel{\B, \Inv}$ and $\steprel{\V}$ coincide,
  where $\Inv \coloneqq \vec{A} \cup \reach{\IB}$.
\end{lemma}
\begin{proof}
  Consider an IBVAS $\IB \coloneqq (\vec{\Delta}_0, \B)$
  with $\B \coloneqq (\vec{\Delta}_1, \ldots, \vec{\Delta}_r)$, and
  a semilinear set $\vec{A} \subseteq \setN^d$.
  We have $\vec{A} = \bigcup_{i=1}^k \vec{b}_i + \vec{Q}_i$
  for some vectors $\vec{b}_1, \ldots, \vec{b}_k \in \setN^d$ and
  some finitely-generated periodic sets $\vec{Q}_1, \ldots, \vec{Q}_k \subseteq \setN^d$.
  Define $\Inv \coloneqq \vec{A} \cup \reach{\IB}$.
  We construct the desired WSVAS $\V \equaldef (T, \sqsubseteq, \delta)$ as follows.
  Let us first introduce the set
  $\Theta \equaldef \{0\} \times \IRuns{\IB} \cup \bigcup_{i=1}^k \{i\} \times (\vec{b}_i + \vec{Q}_i)$
  and the binary relation $\preceq$ on $\Theta$ defined by
  $\eta \preceq \theta$ if
  \begin{itemize}
  \item
    there exist $\rho, \sigma \in \IRuns{\IB}$ such that
    $\eta = (0, \rho)$, $\theta = (0, \sigma)$ and $\rho \trianglelefteq_{\IB} \sigma$, or
  \item
    there exists $i \in \{1, \ldots, k\}$ and $\vec{x}, \vec{y} \in \vec{b}_i + \vec{Q}_i$ such that
    $\eta = (i, \vec{x})$, $\theta = (i, \vec{y})$ and $\vec{x} \leq_{\vec{Q}_i} \vec{y}$.
  \end{itemize}
  Recall that $(\IRuns{\IB}, \trianglelefteq)$ is a wqo by \cref{lem:ibvas-wpo-and-amalgamation}.
  Recall also that each pair $(\vec{b}_i + \vec{Q}_i, \leq_{\vec{Q}_i})$ is a wqo
  since $\vec{Q}_i$ is a finitely-generated periodic set
  (see \cref{par:cylindric}).
  It follows that $(\Theta, \preceq)$ is a wqo.

  \smallskip

  We also introduce the map $\zeta : \Theta \rightarrow \setZ^d$ defined by
  $\zeta((0, \rho)) \equaldef \tgt{\rho}$ and
  $\zeta((i, \vec{x})) \equaldef \vec{x}$
  for all $\rho \in \IRuns{\IB}$, $i \in \{1, \ldots, k\}$ and
  $\vec{x} \in (\vec{b}_i + \vec{Q}_i)$.
  It is routinely checked that $\zeta$ is an order-preserving and diamond map
  from $(\Theta, \preceq)$ to $(\setZ^d, \leq)$.
  This follows, firstly, from
  the amalgamation property of $(\IRuns{\IB}, \trianglelefteq)$
  obtained from \cref{lem:ibvas-wpo-and-amalgamation}, and, secondly,
  from the following ``amalgamation property'' of $(\vec{b}_i + \vec{Q}_i, \leq_{\vec{Q}_i})$.
  For every $\vec{x}, \vec{u}, \vec{v} \in \vec{b}_i + \vec{Q}_i$,
  if $\vec{x} \leq_{\vec{Q}_i} \vec{u}, \vec{v}$ then the vector
  $\vec{y} \equaldef \vec{u} + \vec{v} - \vec{x}$ is in $\vec{b}_i + \vec{Q}_i$ and verifies
  $\vec{u}, \vec{v} \leq_{\vec{Q}_i} \vec{y}$.

  \smallskip

  The set $T$ is defined as the set of triples $(n, \vec{a}, \vec{h})$ with
  $n \in \{1, \ldots, r\}$,
  $\vec{a} \in \vec{\Delta}_n$ and
  $\vec{h} \in \Theta^{n-1}$.
  It is understood that $\Theta^0$ is the singleton set $\{()\}$.
  The binary relation $\sqsubseteq$ on $T$ is defined by
  $(n, \vec{a}, \vec{h}) \sqsubseteq (n', \vec{a}', \vec{h}')$ if
  $n = n'$, $\vec{a} = \vec{a}'$ and
  $\vec{h}(j) \preceq \vec{h}'(j)$ for all $j \in \{1, \ldots, n-1\}$.
  Since $\vec{\Delta}_1, \ldots, \vec{\Delta}_r$ are finite,
  we derive from Dickson's Lemma that $(T, \sqsubseteq)$ is a wqo.
  The map $\delta : T \rightarrow \setZ^d$ is defined by
  $\delta(n, \vec{a}, \vec{h}) \equaldef \vec{a} + \sum_{j=1}^{n-1} \zeta(\vec{h}(j))$.
  Since $\zeta$ is order-preserving and diamond, so is $\delta$.

  \smallskip

  We have shown that $\V = (T, \sqsubseteq, \delta)$ is a WSVAS.
  It remains to show that
  $\steprel{\B, \Inv}$ and $\steprel{\V}$ coincide.
  We first observe that,
  by construction,
  $\zeta(\Theta) = \vec{A} \cup \reach{\IB} = \Inv$.
  Assume that $\vec{x} \steprel{\B, \Inv} \vec{y}$.
  There exist $n \in \{1, \ldots, r\}$ and $\vec{a} \in \vec{\Delta}_n$ such that
  $\vec{y} = \vec{a} + \vec{x} + \sum_{j=2}^{n} \vec{c}_j$ for some
  $\vec{c}_2, \ldots, \vec{c}_n \in \Inv$.
  So for each $j \in \{2, \ldots, n\}$,
  there exists $\theta_j \in \Theta$ such that $\zeta(\theta_j) = \vec{c}_j$.
  Define $t \equaldef (n, \vec{a}, \vec{h})$
  where $\vec{h} \in \Theta^{n-1}$ is defined by
  $\vec{h}(j) \equaldef \theta_{j+1}$ for all $j \in \{1, \ldots, n-1\}$.
  We observe that $t \in T$ and $\vec{y} = \vec{x} + \delta(t)$,
  hence,
  $\vec{x} \steprel{\V} \vec{y}$.
  Conversely,
  assume that $\vec{x} \steprel{\V} \vec{y}$.
  There exists a transition $t \equaldef (n, \vec{a}, \vec{h})$ in $T$ such that
  $\vec{y} = \vec{x} + \delta(t) = \vec{x} + \vec{a} + \sum_{j=1}^{n-1} \zeta(\vec{h}(j))$.
  Since $\zeta(\vec{h}(j)) \in \Inv$ for each $j \in \{1, \ldots, n-1\}$,
  we derive that $\vec{x} \steprel{\B, \Inv} \vec{y}$.
  We have shown that $\steprel{\B, \Inv}$ and $\steprel{\V}$ coincide.
\end{proof}

\subsection{One Source Outside the Inductive Invariant Is Enough}

In this subsection,
we show that,
under a technical condition,
BVAS runs with multiple sources outside an inductive invariant are
equivalent, reachability-wise, to runs with only one such source.
This property is formally stated in \cref{lem:bvas-branch-like} below.
We consider a BVAS $\B \coloneqq (\vec{\Delta}_1, \ldots, \vec{\Delta}_r)$ for the rest of this subsection.

\begin{lemma}
  \label{lem:bvas-branch-like}
  For every sets $\Inv, \vec{P}, \vec{X} \subseteq \setN^d$ such that
  $\reach{\Inv^+} \subseteq \Inv$,
  $\vec{P}$ is periodic and
  $\vec{X} + \vec{P} \subseteq \vec{X}$,
  if $\reach{\vec{X} \,{\shuffle}\, \Inv^*} \subseteq \Inv + \vec{P}$ then
  it holds that
  $\reach{\vec{X} \,{\shuffle}\, \Inv^*} = \reach{\vec{X}^+ \,{\shuffle}\, \Inv^*}$.
\end{lemma}
\begin{proof}
  Let $\Inv, \vec{P}$ and $\vec{X}$ be subsets of $\setN^d$ such that
  $\reach{\Inv^+} \subseteq \Inv$,
  $\vec{P}$ is periodic and
  $\vec{X} + \vec{P} \subseteq \vec{X}$.
  We put $\vec{R} \equaldef \reach{\vec{X} \,{\shuffle}\, \Inv^*}$ and
  assume that $\vec{R} \subseteq \Inv + \vec{P}$.
  We only need to show that
  $\reach{\vec{X}^+ \,{\shuffle}\, \Inv^*} \subseteq \vec{R}$
  as the converse inclusion trivially holds.

  \smallskip

  Let us show by induction that for every run $\rho \in \Runs{\B}$,
  if $\src{\rho} \in (\vec{X}^+ \,{\shuffle}\, \Inv^*)$ then $\tgt{\rho} \in \vec{R}$.
  Consider a run $\rho \coloneqq (\vec{c}, (\rho_1, \ldots, \rho_n))$
  and assume by induction that
  $\src{\rho_i} \in (\vec{X}^+ \,{\shuffle}\, \Inv^*)$ implies $\tgt{\rho_i} \in \vec{R}$,
  for all $i \in \{1, \ldots, n\}$.
  Suppose that $\src{\rho} \in (\vec{X}^+ \,{\shuffle}\, \Inv^*)$ and
  let us show that $\tgt{\rho} \in \vec{R}$.
  Notice
  that $\src{\rho_i} \in (\vec{X}^* \,{\shuffle}\, \Inv^*)$ for all $i$, and
  that $\src{\rho_j} \in (\vec{X}^+ \,{\shuffle}\, \Inv^*)$ for some $j$.
  To shorten notation,
  let us write $\vec{c}_i \equaldef \tgt{\rho_i}$ for each $i \in \{1, \ldots, n\}$.
  We make the following easy observations:
  \begin{itemize}
  \item
    if $\src{\rho_i} \in (\vec{X}^+ \,{\shuffle}\, \Inv^*)$ then
    $\vec{c}_i \in \vec{R}$ by induction hypothesis.
    It follows that $\vec{c}_i \in \Inv + \vec{P}$
    since $\vec{R} \subseteq \Inv + \vec{P}$ by assumption.
  \item
    otherwise,
    $\src{\rho_i} \in \Inv^+$ and it follows that $\vec{c}_i \in \Inv$,
    since $\reach{\Inv^+} \subseteq \Inv$.
  \end{itemize}
  These observations entail
  that $\vec{c}_i \in \Inv + \vec{P}$ for all $i$ and
  that $\vec{c}_j \in \vec{R}$ for some $j$.
  By definition of $\vec{R}$,
  we get that $\vec{c}_j = \tgt{\sigma_j}$ for some run $\sigma_j$ of $\B$
  such that $\src{\sigma_j} \in (\vec{X} \,{\shuffle}\, \Inv^*)$.
  Let us define $\rho'_j \equaldef \sigma_j$ and
  $\rho'_i \equaldef (\vec{c}_i, ())$ for all $i \neq j$.
  Clearly,
  $\rho'_i$ is a run of $\B$ with same target as $\rho_i$,
  for every $i \in \{1, \ldots, n\}$.
  We derive that
  $\rho' \coloneqq (\vec{c}, (\rho'_1, \ldots, \rho'_n))$ is also a run of $\B$.
  Moreover,
  the two following properties hold by construction.
  Firstly,
  $\src{\rho'_j} = u \vec{x} v$ for some $\vec{x} \in \vec{X}$ and $u, v \in \Inv^*$.
  Secondly,
  for every $i \neq j$,
  $\tgt{\rho'_i} = \vec{h}_i + \vec{p}_i$
  for some $\vec{h}_i \in \Inv$ and $\vec{p}_i \in \vec{P}$.
  Note that $\vec{h}_i, \vec{p}_i \in \setN^d$.

  \smallskip

  We now transform
  $\rho'$ into a run $\tau$ with $\tgt{\tau} = \tgt{\rho'}$ by
  removing $\vec{p}_i$ from $\rho'_i$ for each $i \neq j$ and
  adding $\vec{p} \equaldef \sum_{i \neq j} \vec{p}_i$ to $\rho'_j$.
  Formally,
  we define $\tau_i \equaldef (\vec{h}_i, ())$ for all $i \neq j$, and
  we construct $\tau_j$ from $\rho'_j$ as follows.
  According to \cref{fact:monotony},
  there exists a run $\tau_j$ of $\B$ with
  $\src{\tau_j} = u (\vec{x} + \vec{p}) v$ and
  $\tgt{\tau_j} = \tgt{\rho'_j} + \vec{p}$.
  The assumption that $\vec{X} + \vec{P} \subseteq \vec{X}$
  guarantees that $\src{\tau_j} \in (\vec{X} \,{\shuffle}\, \Inv^*)$.
  Observe that each $\tau_i$ is a run of $\B$ and that
  $\sum_{i=1}^n \tgt{\tau_i} = \sum_{i=1}^n \tgt{\rho'_i}$.
  We derive that
  $\tau \coloneqq (\vec{c}, (\tau_1, \ldots, \tau_n))$ is also a run of $\B$.
  Moreover,
  it holds by construction that
  $\src{\tau} \in (\vec{X} \,{\shuffle}\, \Inv^*)$.
  Therefore,
  $\tgt{\rho} = \tgt{\tau} \in \vec{R}$.
  This concludes the induction and the proof of the lemma.
\end{proof}

\subsection{Proof of \cref{lem:bvas-safelin}}
\label{subsec:proof:lem:bvas-safelin}

We now have the required ingredients to prove \cref{lem:bvas-safelin}.
Consider an IBVAS $\IB \coloneqq (\vec{\Delta}_0, \B)$,
two semilinear sets $\vec{A}, \vec{T} \subseteq \setN^d$ and
an almost linear set $\vec{b} + \vec{P} \subseteq \setN^d$.
Define $\Inv \coloneqq \vec{A} \cup \reach{\IB}$ and
$\vec{Q} \coloneqq \clin{\vec{P}}$.
Assume
that $\reach[\B]{\Inv^+} \subseteq \Inv$ and
that $\reach[\B]{(\vec{b} + \vec{P}) \,{\shuffle}\, \Inv^*} \subseteq \vec{T} \subseteq \Inv + \vec{Q}$.
To prove \cref{lem:bvas-safelin},
we need to show that
there exists $p \in \vec{P}$ such that
$\reach[\B]{(\vec{b} + \vec{p} + \vec{Q})^+ \,{\shuffle}\, \Inv^*} \subseteq \vec{T}$.

\smallskip

According to \cref{lem:bvas-to-wsvas},
there exists a WSVAS $\V$ such that $\steprel{\B, \Inv}$ and $\steprel{\V}$ coincide.
By \cref{lem:steprel-reach},
it holds that
$\reach[\V]{\vec{b} + \vec{P}} = \reach[\B]{(\vec{b} + \vec{P}) \,{\shuffle}\, \Inv^*} \subseteq \vec{T}$.
We derive from \cref{lem:wsvas-safelin} that
there exists $\vec{p} \in \vec{P}$ satisfying
$\reach[\V]{\vec{b} + \vec{p} + \vec{Q}} \subseteq \vec{T}$.
Define $\vec{X} \equaldef \vec{b} + \vec{p} + \vec{Q}$.
Using \cref{lem:steprel-reach} once more,
we get that
$\reach[\B]{\vec{X} \,{\shuffle}\, \Inv^*} = \reach[\V]{\vec{X}} \subseteq \vec{T}$.
Note that $\vec{X} + \vec{Q} \subseteq \vec{X}$ and recall that $\vec{T} \subseteq \Inv + \vec{Q}$.
It follows from \cref{lem:bvas-branch-like} that
$\reach[\B]{\vec{X} \,{\shuffle}\, \Inv^*} = \reach[\B]{\vec{X}^+ \,{\shuffle}\, \Inv^*}$,
hence,
$\reach[\B]{\vec{X}^+ \,{\shuffle}\, \Inv^*} = \reach[\V]{\vec{X}} \subseteq \vec{T}$.
This concludes the proof of \cref{lem:bvas-safelin}.

\section{Face-Stripping Theorem}\label{sec:stripping}
Let $\vec{Q}\subseteq\setZ^d$ be a full\footnote{See \cref{rem:full} for details.} finitely-generated periodic set. In this section, we state and prove a way to decompose the difference of two finitary $\vec{Q}$-cylindric sets $\vec{X}\subseteq\vec{Y}\subseteq \setZ^d$. This decomposition, so-called the \emph{face-stripping theorem}, is obtained by stripping $\vec{Y}\setminus \vec{X}$ following the \emph{faces} (a classical notion recalled below) of the cone spanned by $\vec{Q}$. We first provide some definitions and then provide the statement of the main theorem and a central application of that result.

\medskip

 The \emph{orthogonal} $\vec{X}^\perp$ of a set $\vec{X}\subseteq\setQ^d$ is the set of vectors $\vec{e}\in\setQ^d$ such that $\scalar{\vec{e}}{\vec{x}}=0$ for every $\vec{x}\in\vec{X}$ where $\scalar{\vec{e}}{\vec{x}}$ is the \emph{dot-product} (see~\cref{sec:prelim}). Notice that $\vec{X}^\perp$ is a vector space. We say that a set $\vec{F}\subseteq \setQ^d$ is a face of a finitely-generated cone $\vec{C}\subseteq\setQ^d$ if $\vec{F}=\vec{C}\cap\vec{L}^\perp$ for a finite set $\vec{L}\subseteq \setQ^d$ of vectors $\vec{e}\in\setQ^d$ satisfying $\scalar{\vec{e}}{\vec{c}}\geq 0$ for every $\vec{c}\in\vec{C}$. Since the intersection of a finitely-generated cone with a vector space is a finitely-generated cone, we deduce that a face is a finitely-generated cone. We denote by $\mathcal{F}(\vec{C})$ the set of \emph{faces} of $\vec{C}$. Recall that $\mathcal{F}(\vec{C})$ is a finite set~\cite[Chapter 8]{DBLP:books/daglib/0090562}, $\vec{C}\in\mathcal{F}(\vec{C})$ and for every face $\vec{F}\in \mathcal{F}(\vec{C})\setminus\{\vec{C}\}$, we have $\dim{\vec{F}-\vec{F}}<\dim{\vec{C}-\vec{C}}$\footnote{Since $\vec{F}=\vec{C}\cap\vec{L}^\perp$ and $\vec{F}\not=\vec{C}$ we deduce that the vector space $\vec{V}\coloneqq\vec{C}-\vec{C}$ cannot be included in $\vec{L}^\perp$. In particular the vector space $\vec{W}\coloneqq(\vec{F}-\vec{F})\cap\vec{L}^\perp$ is strictly included in $\vec{V}$. As $\vec{W}\subsetneq \vec{V}$ we deduce that $\dim{\vec{W}}<\dim{\vec{V}}$.}.

\medskip

Let $\vec{Q}$ be a full finitely-generated periodic set, let $\vec{C}\coloneqq\setQ_{\geq 0}\vec{Q}$ be the finitely-generated cone spanned by $\vec{Q}$, and let $\vec{F}_1,\ldots,\vec{F}_n$ be a linearization of the set of faces of $\vec{C}$ with respect to $\supseteq$. It follows that $n=|\mathcal{F}(\vec{C})|$, $\mathcal{F}(\vec{C})=\{\vec{F}_1,\ldots,\vec{F}_n\}$, $\vec{F}_1=\vec{C}$, and $\vec{F}_i\supseteq \vec{F}_j$ implies $i\leq j$.

\begin{theorem}[Face-Stripping Theorem]\label{thm:leaf}
  Let $\vec{X}\subseteq \vec{Y}$ be two finitary $\vec{Q}$-cylindric subsets of $\setZ^d$, and let $\vec{S}_1,\ldots,\vec{S}_n$ be subsets of $\setZ^d$ defined inductively for every $i\in \{1,\ldots,n\}$ as follows where $\vec{S}_{<i}\coloneqq\vec{S}_1\cup\ldots\cup\vec{S}_{i-1}$:
  $$\vec{S}_i\coloneqq\{\vec{s}\in\setZ^d \mid \vec{s}+(\vec{Q}\cap\vec{F}_i)\subseteq (\vec{Y}\setminus\vec{S}_{<i})\setminus \vec{X}\}$$
  Then $\vec{S}_1,\ldots,\vec{S}_n$ is a disjoint decomposition\footnotemark[\getrefnumber{footnote:disjoint-decomposition}] of $\vec{Y}\setminus\vec{X}$ such that for every $i\in\{1,\ldots,n\}$ we have:
  \begin{itemize}
  \item $\vec{S}_i$ is a finitary $(\vec{Q}\cap\vec{F}_i)$-cylindric set, and
  \item $\vec{Y}\setminus \vec{S}_{<i}$ is a finitary $\vec{Q}$-cylindric set that contains $\vec{X}$.
  \end{itemize}
  Moreover, if $\vec{Y}\subseteq \vec{X}-\vec{Q}$ then $\vec{S}_1$ is empty.
\end{theorem}

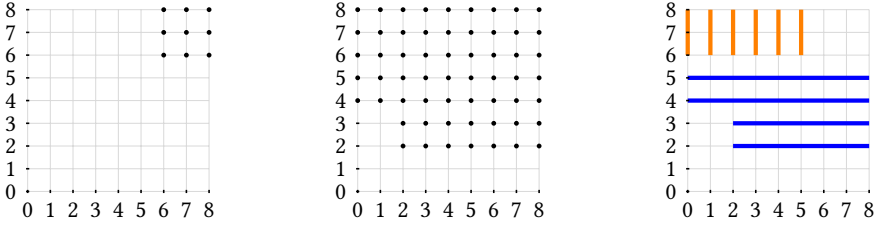
\begin{figure}[t]
  \centering
\hspace{1cm}
  \begin{tikzpicture}[scale=0.3]
    \draw[step=1cm, gray!30, very thin] (0,0) grid (8,8);

    \foreach \x in {6,7,8} \foreach \y in {6,7,8}\draw[fill=black] (\x,\y) circle (2pt);

    \foreach \x in {0,1,...,8}
        \draw (\x, 2pt) -- (\x, -2pt) node[below] {\small $\x$};
        
    \foreach \y in {0,1,...,8}
        \draw (2pt, \y) -- (-2pt, \y) node[left] {\small $\y$};

\end{tikzpicture}
\hfill
\begin{tikzpicture}[scale=0.3]
    \draw[step=1cm, gray!30, very thin] (0,0) grid (8,8);

    \foreach \x in {2,3,4,5,6,7,8} \draw[fill=black] (\x,2) circle (2pt);
    \foreach \x in {2,3,4,5,6,7,8} \draw[fill=black] (\x,3) circle (2pt);
    \foreach \x in {0,1,2,3,4,5,6,7,8} \foreach \y in {4,5,6,7,8}\draw[fill=black] (\x,\y) circle (2pt);

    \foreach \x in {0,1,...,8}
        \draw (\x, 2pt) -- (\x, -2pt) node[below] {\small $\x$};
        
    \foreach \y in {0,1,...,8}
        \draw (2pt, \y) -- (-2pt, \y) node[left] {\small $\y$};

\end{tikzpicture}
\hfill
\begin{tikzpicture}[scale=0.3]
    \draw[step=1cm, gray!30, very thin] (0,0) grid (8,8);

    \draw[blue, ultra thick] (2,2) -- (8,2);
    \draw[blue, ultra thick] (2,3) -- (8,3);
    \draw[blue, ultra thick] (0,4) -- (8,4);
    \draw[blue, ultra thick] (0,5) -- (8,5);

    \foreach \x in {0,1,2,3,4,5} {
        \draw[orange, ultra thick] (\x,6) -- (\x,8);
    }

    \foreach \x in {0,1,...,8}
        \draw (\x, 2pt) -- (\x, -2pt) node[below] {\small $\x$};
        
    \foreach \y in {0,1,...,8}
        \draw (2pt, \y) -- (-2pt, \y) node[left] {\small $\y$};

\end{tikzpicture}
\hspace{1cm}
\caption{From left to right : (1) $\vec{X}\coloneqq (6,6)+\setN^2$, (2) $\vec{Y}\coloneqq \{(2,2),(0,4)\}+\setN^2$, and (3) $S_2$ in blue and $S_3$ in orange.
}\label{fig:striping}
\end{figure}

\begin{example}\label{ex:striping1}
  Let $\vec{Q}$ be the full finitely-generated periodic set $\setN^2$. Notice that $\vec{C}\coloneqq\cone{\vec{P}}$ is the cone $\setQ_{\geq 0}^2$. It follows that $\mathcal{F}(\vec{C})=\{\vec{F}_1,\vec{F}_2,\vec{F}_3,\vec{F}_4\}$ where $\vec{F}_1\coloneqq \setQ_{\geq 0}^2$, $\vec{F}_2\coloneqq\setQ_{\geq 0}\times\{0\}$, $\vec{F}_3\coloneqq\{0\}\times\setQ_{\geq 0}$, and $\vec{F}_4=\{(0,0)\}$. Notice that $\vec{F}_i\subseteq \vec{F}_j$ implies $i\leq j$. Let us introduce the finitary $\vec{Q}$-cylindric sets $\vec{X}\coloneqq (6,6)+\setN^2$ and $\vec{Y}\coloneqq \{(2,2),(0,4)\}+\setN^2$. Those sets as well as the sets $\vec{S}_2,\vec{S}_3$ introduced by \cref{thm:leaf} are depicted in \cref{fig:striping}. Notice that $\vec{S}_1=\emptyset$, $\vec{S}_2=\{(2,2),(2,3),(0,4),(0,5)\}+(\setN\times\{0\})$, $\vec{S}_3=\{(0,6), (1,6),(2,6),(3,6),(4,6),(5,6)\}+(\{0\}\times\setN)$, and $\vec{S}_4=\emptyset$.
\end{example}

We now prove the following central result satisfied by the decomposition given by the face-stripping theorem. 
\begin{lemma}\label{lem:diagonalstripping} 
  Let $R$ be a $\vec{Q}$-diagonal\footnote{See \cref{sec:prelim} for the definition.} binary relation over $\setZ^d$ such that $R\cap (\vec{X}\times(\vec{Y}\setminus \vec{X}))$ is empty. Then $R\cap (\vec{S}_j \times\vec{S}_i)$ is empty for every $0\leq i<j\leq n$.
\end{lemma}
\begin{proof}
  Assume by contradiction that there exists $(\vec{s}_j,\vec{s}_i)\in R\cap (\vec{S}_j\times \vec{S}_i)$ for some $0\leq i<j\leq n$. Let $\vec{x}\in \vec{s}_j+(\vec{Q}\cap\vec{F}_i)$ and let us prove that $\vec{x}\not\in\vec{X}$. There exists $\vec{q}\in\vec{Q}\cap\vec{F}_i$ such that $\vec{x}=\vec{s}_j+\vec{q}$. As $\vec{S}_i$ is $(\vec{Q}\cap\vec{F}_i)$-cylindric, we deduce that $\vec{y}\coloneqq\vec{s}_i+\vec{q}$ is in $\vec{S}_i$, and in particular in $\vec{Y}\setminus\vec{X}$. As $(\vec{s}_j,\vec{s}_i)\in R$ and $R$ is $\vec{Q}$-diagonal, we deduce that $(\vec{x},\vec{y})\in R$. As $R$ has an empty intersection with $\vec{X}\times(\vec{Y}\setminus \vec{X})$ and $\vec{y}\in \vec{Y}\setminus\vec{X}$, we deduce that $\vec{x}\not\in\vec{X}$. We have proved that $\vec{s}_j+(\vec{Q}\cap\vec{F}_i)$ has an empty intersection with $\vec{X}$. Since $\vec{s}_j\in \vec{S}_j$ and $\vec{S}_1,\ldots,\vec{S}_n$ is a disjoint decomposition of $\vec{Y}\setminus\vec{X}$, we deduce that $\vec{s}_j\in\vec{Y}\setminus \vec{S}_{<i}$. As $\vec{Y}\setminus \vec{S}_{< i}$ is $\vec{Q}$-cylindric, we deduce that $\vec{s}_j+(\vec{Q}\cap\vec{F}_i)\subseteq \vec{Y}\setminus \vec{S}_{< i}$. Since $\vec{s}_j+(\vec{Q}\cap\vec{F}_i)$ has an empty intersection with $\vec{X}$, it is included in $(\vec{Y}\setminus \vec{S}_{< i})\setminus \vec{X}$. Hence $\vec{s}_j\in \vec{S}_i$ and we get a contradiction since $\vec{S}_i$ and $\vec{S}_j$ are disjoint. Therefore $R\cap (\vec{S}_j\times \vec{S}_i)$ is empty.
\end{proof}

\begin{example}\label{ex:striping2}
  Let us come back to \cref{ex:striping1}. Assume that $R$ is a $\setN^2$-diagonal relation such that there exists $(\vec{x},\vec{y})\in R\cap (\vec{S}_3\times\vec{S}_2)$. From $\vec{x}\in \vec{S}_3$ and $\vec{y}\in\vec{S}_2$, we deduce that $\vec{x}'\coloneqq\vec{x}+(6,0)\in \vec{X}$ and $\vec{y}'\coloneqq\vec{y}+(6,0)\in \vec{Y}\setminus\vec{X}$. As $R$ is $\setN^2$-diagonal, we deduce that $(\vec{x}',\vec{y}')$ is in $R\cap (\vec{X}\times(\vec{Y}\setminus\vec{X}))$.
\end{example}

The remainder of this section is devoted to the proof of \cref{thm:leaf}. Apart from this theorem and \cref{lem:diagonalstripping}, the material presented in this section is not used in the sequel.

\medskip

Let us fix a full finitely-generated periodic set $\vec{Q}$, and let $\vec{C}\coloneqq\setQ_{\geq 0}\vec{Q}$ be the cone spanned by $\vec{Q}$. Since $\vec{C}$ is a finitely-generated cone, the Farkas-Minkowski-Weyl theorem~\cite[Corollary 7.1a]{DBLP:books/daglib/0090562} shows that there exists a finite set $\vec{E}\subseteq \setZ^d$ satisfying the following equality.
\begin{equation}\label{eq:dual}\vec{C}=\{\vec{c}\in\setQ^d \mid \bigwedge_{\vec{e}\in \vec{E}}\scalar{\vec{e}}{\vec{c}}\geq 0\}
\end{equation}
Recall from the Farkas lemma~\cite[Corollary 7.1d]{DBLP:books/daglib/0090562} that a set $\vec{F}\subseteq \setQ^d$ is a face of $\vec{C}$ if, and only if, $\vec{F}=\vec{C}\cap\vec{L}^\perp$ where $\vec{L}\coloneqq\vec{E}\cap\vec{F}^\perp$. It follows that $\mathcal{F}(\vec{C})$ is finite since it contains at most $2^{|\vec{E}|}$ faces, $\vec{C}$ is the $\subseteq$-maximal face of $\vec{C}$, and $\vec{E}^\perp$ is the $\subseteq$-minimal face of $\vec{C}$. Notice that the minimal face is a vector space.

\medskip

We first introduce some folklore results.
\begin{lemma}\label{lem:centerF}
    For every face $\vec{F}$, there exists a vector $\vec{q}\in \vec{Q}\cap\vec{F}$ such that $\vec{E}\cap\vec{F}^\perp=\vec{E}\cap\vec{q}^\perp$.
\end{lemma}
\begin{proof}
    For every $\vec{e}\in\vec{E}\setminus\vec{F}^\perp$ there exists $\vec{f}_{\vec{e}}\in \vec{F}$ such that $\scalar{\vec{e}}{\vec{f}_{\vec{e}}}\not=0$. Since $\vec{F}\subseteq \vec{C}$, we have $\scalar{\vec{e}'}{\vec{f}_{\vec{e}}}\geq 0$ for every $\vec{e},\vec{e}'\in\vec{E}$. Now, just observe that $\vec{f}\coloneqq \sum_{\vec{e}\in \vec{E}\setminus\vec{F}^\perp}\vec{f}_{\vec{e}}$ is a vector in $\vec{F}$ such that $\vec{E}\cap\vec{F}^\perp=\vec{E}\cap\vec{f}^\perp$. Since $\vec{F}\subseteq\vec{C}=\setQ_{\geq 0}\vec{Q}$, there exists $n\in\setN_{>0}$ such that $q\coloneqq n\vec{f}$ is in $\vec{Q}$. Notice that $\vec{q}$ satisfies the lemma.
\end{proof}

\begin{lemma}[{\cite[Corollary 7.1b]{DBLP:books/daglib/0090562}}]\label{lem:system}
  Let $(\vec{a}_j,c_j)_{1\leq j\leq k}$ be a sequence of pairs in $\setZ^d\times\setZ$. There exists a finite set $\vec{B}\subseteq\setZ^d$ such that the following equality holds.
  $$
  \{\vec{s}\in\setZ^d \mid \bigwedge_{j=1}^k\scalar{\vec{a}_j}{\vec{s}}\geq c_j\}
  =
  \vec{B}+ 
  \{\vec{s}\in\setZ^d \mid \bigwedge_{j=1}^k\scalar{\vec{a}_j}{\vec{s}}\geq 0\}
  $$
\end{lemma}

\begin{corollary}\label{lem:solsystem}
  Let $\vec{K}\subseteq\vec{E}$, let $\vec{H}$ be the face $\vec{C}\cap \vec{K}^\perp$, and let $\vec{P}$ be the finitely-generated periodic set $\vec{Q}\cap\vec{H}$. For every mapping $\lambda\in\setQ^{\vec{E}}$, the following set $\vec{S}$ is a finitary $\vec{P}$-cylindric set.
  $$
    \vec{S}\coloneqq\{\vec{s}\in\vec{Q} \mid \bigwedge_{\vec{e}\in\vec{E}}\scalar{\vec{e}}{\vec{s}}\geq \lambda(\vec{e})\wedge\bigwedge_{\vec{e}\in\vec{K}}\scalar{\vec{e}}{\vec{s}}=\lambda(\vec{e})\}
  $$
\end{corollary}
\begin{proof}
  By observing that an equality of the form $\scalar{\vec{e}}{\vec{s}}=\lambda(\vec{e})$ is equivalent to the conjunction of $\scalar{\vec{e}}{\vec{s}}\geq \lambda(\vec{e})$ and $\scalar{-\vec{e}}{\vec{s}}\geq -\lambda(\vec{e})$, the proof follows from \cref{lem:system}.
\end{proof}

Let $\vec{X}\subseteq \vec{Y}\subseteq \setZ^d$ and $\vec{F}$ be a face of $\vec{C}$. We introduce the following $(\vec{Q}\cap\vec{F})$-cylindric set:
$$\extract{\vec{F}}{\vec{X}}{\vec{Y}}\coloneqq\{\vec{s}\in\setZ^d \mid \vec{s}+(\vec{Q}\cap\vec{F})\subseteq \vec{Y}\setminus\vec{X}\}$$
\begin{lemma}\label{lem:cyclicpreserve}
  The set $\vec{Y}\setminus \extract{\vec{F}}{\vec{X}}{\vec{Y}}$ is $\vec{Q}$-cylindric for every $\vec{Q}$-cylindric sets $\vec{X}\subseteq \vec{Y}$ and for every face $\vec{F}$ of $\vec{C}$.
\end{lemma}
\begin{proof}
  Let us denote by $\vec{S}$ the set $\extract{\vec{F}}{\vec{X}}{\vec{Y}}$. Let $\vec{y}\in \vec{Y}\setminus \vec{S}$ and $\vec{q}\in \vec{Q}$, and let us prove that $\vec{y}+\vec{q}\in \vec{Y}\setminus \vec{S}$. Since $\vec{y}\not\in\vec{S}$, there exists $\vec{p}\in \vec{Q}\cap\vec{F}$ such that $\vec{y}+\vec{p}\not\in \vec{Y}\setminus\vec{X}$. As $\vec{Y}$ is $\vec{Q}$-cylindric and $\vec{p}\in \vec{Q}$, we deduce that $\vec{y}+\vec{p}\in \vec{Y}$. It follows that $\vec{y}+\vec{p}\in\vec{X}$. Hence $\vec{y}+\vec{q}+\vec{p}\in\vec{X}$ since $\vec{X}$ is $\vec{Q}$-cylindric. We derive $\vec{y}+\vec{q}\not\in \vec{S}$. Since $\vec{y}\in \vec{Y}$ and $\vec{Y}$ is $\vec{Q}$-cylindric, we also get $\vec{y}+\vec{q}\in \vec{Y}$. Therefore $\vec{y}+\vec{q}\in \vec{Y}\setminus\vec{S}$.
\end{proof}

\begin{lemma}\label{lem:emptyset}
  For every $\vec{Q}$-cylindric sets $\vec{X}\subseteq \vec{Y}$, the $\subseteq$-minimal face $\vec{F}_{min}$ satisfies $\extract{\vec{F}_{min}}{\vec{X}}{\vec{Y}}=\vec{Y}\setminus \vec{X}$.
\end{lemma}
\begin{proof}
  Trivially, $\extract{\vec{F}_{min}}{\vec{X}}{\vec{Y}}$ is a subset of $\vec{Y}\setminus \vec{X}$. For the converse inclusion, let $\vec{y}\in \vec{Y}\setminus\vec{X}$ and let us prove that for every $\vec{q}\in\vec{Q}\cap\vec{F}_{min}$ we have $\vec{y}+\vec{q}\in \vec{Y}\setminus\vec{X}$. Since $\vec{Y}$ is $\vec{Q}$-cylindric we deduce that $\vec{y}+\vec{q}\in\vec{Y}$. Assume by contradiction that $\vec{y}+\vec{q}\in \vec{X}$. Since $\vec{q}\in \vec{F}_{min}$ we deduce that $-\vec{q}\in \vec{F}_{min}$ since the minimal face $\vec{F}_{min}$ is a vector space. As $\vec{Q}$ is full, we deduce that $-\vec{q}\in \vec{Q}$. As $\vec{X}$ is $\vec{Q}$-cylindric, it follows that $(\vec{y}+\vec{q})-\vec{q}\in \vec{X}$. Hence $\vec{y}\in\vec{X}$ and we get a contradiction. Hence $\vec{y}+\vec{q}\not\in\vec{X}$. We have proved that $\vec{y}+(\vec{Q}\cap\vec{F}_{min})\subseteq \vec{Y}\setminus\vec{X}$. Hence $\vec{y}\in \extract{\vec{F}_{min}}{\vec{X}}{\vec{Y}}$ and we have proved the converse inclusion.
\end{proof}

\begin{lemma}\label{lem:extractB}
  For every finitary $\vec{Q}$-cylindric sets $\vec{X}\subseteq \vec{Y}$ and for every $\subseteq$-maximal face $\vec{F}$ such that $\extract{\vec{F}}{\vec{X}}{\vec{Y}}$ is non-empty, the set $\extract{\vec{F}}{\vec{X}}{\vec{Y}}$ is finitary $(\vec{Q}\cap\vec{F})$-cylindric.
\end{lemma}
\begin{proof}
  It is sufficient to prove that $\leq_{\vec{Q}\cap\vec{F}}$ is a wqo on $\extract{\vec{F}}{\vec{X}}{\vec{Y}}$. So, let us consider a sequence $(\vec{s}_n)_{n\in\setN}$ of vectors in that set. Since this sequence is in the finitary $\vec{Q}$-cylindric set $\vec{Y}$, by extracting a subsequence we can assume w.l.o.g. that $(\vec{s}_n)_{n\in\setN}$ is non-decreasing for $\leq_{\vec{Q}}$. It follows that for every $\vec{e}\in \vec{E}$, the sequence $(\scalar{\vec{e}}{\vec{s}_n})_{n\in\setN}$ is a non-decreasing sequence of integers. By extracting a subsequence, we can assume w.l.o.g. that for every $\vec{e}\in \vec{E}$, the sequence $(\scalar{\vec{e}}{\vec{s}_n})_{n\in\setN}$ is either constant or unbounded. We denote by $\vec{U}$ the set of $\vec{e}\in\vec{E}$ such that this sequence is unbounded.
  
  Since $\vec{F}$ is a face of $\vec{C}$, we have $\vec{F}=\vec{C}\cap\vec{L}^\perp$ where $\vec{L}\coloneqq \vec{E}\cap\vec{F}^\perp$. We introduce the set $\vec{K}\coloneqq\vec{L}\setminus \vec{U}$, the face $\vec{H}\coloneqq \vec{C}\cap \vec{K}^\perp$, the finitely-generated periodic set $\vec{P}\coloneqq\vec{Q}\cap\vec{H}$, the mapping $\lambda\in\setZ^{\vec{E}}$ defined by $\lambda(\vec{e})\coloneqq\scalar{\vec{e}}{\vec{s}_0}$. The set $\vec{S}$ introduced by \cref{lem:solsystem} is thus a finitary $\vec{P}$-cylindric set. In particular $\leq_{\vec{P}}$ is a wqo on $\vec{S}$. As $(\vec{s}_n)_{n\in\setN}$ is a sequence of vector in  $\vec{S}$, by extracting a subsequence, we can assume w.l.o.g. that the sequence is non-decreasing for $\leq_{\vec{P}}$.

  Assume by contradiction that $\vec{s}_0+\vec{P}\not\subseteq \vec{Y}\setminus\vec{X}$. There exists $\vec{p}\in\vec{P}$ such that $\vec{s}_0+\vec{p}\not\in \vec{Y}\setminus\vec{X}$. Since $\vec{Y}$ is $\vec{Q}$-cylindric, $\vec{s}_0\in\vec{Y}$, and $\vec{p}\in\vec{Q}$, we deduce that $\vec{s}_0+\vec{p}\in\vec{Y}$. It follows that $\vec{s}_0+\vec{p}\in \vec{X}$. There exists $n\in\setN$ such that for every $\vec{e}\in\vec{U}$ we have $\scalar{\vec{e}}{\vec{s}_n}\geq \scalar{\vec{e}}{(\vec{s}_0+\vec{p})}$. \cref{lem:centerF} shows that there exists a vector $\vec{q}\in \vec{Q}\cap\vec{F}$ such that $\vec{E}\cap\vec{F}^\perp=\vec{E}\cap\vec{q}^\perp$. By replacing $\vec{q}$ by a multiple of $\vec{q}$, we can assume w.l.o.g. that for every $\vec{e}\in \vec{E}\setminus\vec{F}^\perp$, we have $\scalar{\vec{e}}{\vec{q}_n}\geq 0$ where $\vec{q}_n\coloneqq \vec{s}_n-\vec{s}_0-\vec{p}+\vec{q}$. Moreover, for every $\vec{e}$ in $\vec{E}\cap\vec{F}^\perp=\vec{L}$ either $\vec{e}\in \vec{U}$ and in this case $\scalar{\vec{e}}{\vec{q}_n}\geq \scalar{\vec{e}}{\vec{q}}\geq 0$, or $\vec{e}\in \vec{K}$ and in that case from $\scalar{\vec{e}}{(\vec{s}_n-\vec{s}_0)}=0$, $\scalar{\vec{e}}{-\vec{p}}=0$ since $\vec{p}\in\vec{P}\subseteq \vec{H}\subseteq \vec{K}^\perp$, we deduce that $\scalar{\vec{e}}{\vec{q}_n}\geq \scalar{\vec{e}}{\vec{q}}\geq 0$. We have proved that $\scalar{\vec{e}}{\vec{q}_n}\geq 0$ for every $\vec{e}\in\vec{E}$. Hence $\vec{q}_n\in \vec{C}$. Since $\vec{q}_n\in \vec{Q}-\vec{Q}$ and $\vec{Q}$ is full, we deduce that $\vec{q}_n\in\vec{Q}$. As $\vec{s}_n+\vec{q}=\vec{s}_0+\vec{p}+\vec{q}_n$, $\vec{s}_0+\vec{p}\in\vec{X}$, $\vec{q}_n\in\vec{Q}$, and $\vec{X}$ is $\vec{Q}$-cylindric, we deduce that $\vec{s}_0+\vec{p}+\vec{q}_n\in\vec{X}$. We have proved that $\vec{s}_n+\vec{q}\in \vec{X}$. From $\vec{s}_n+(\vec{Q}\cap\vec{F})\subseteq \vec{Y}\setminus\vec{X}$ and $\vec{q}\in \vec{Q}\cap\vec{F}$ we get a contradiction. Therefore $\vec{s}_0+(\vec{Q}\cap\vec{H})\subseteq \vec{Y}\setminus\vec{X}$ since $\vec{P}=\vec{Q}\cap\vec{H}$. 
  
  It follows that $\extract{\vec{H}}{\vec{X}}{\vec{Y}}$ is non-empty since it contains $\vec{s}_0$. As $\vec{F}\subseteq \vec{H}$, by maximality of $\vec{F}$, we deduce that $\vec{F}=\vec{H}$. In particular $\vec{P}=\vec{Q}\cap\vec{F}$. Since $(\vec{s}_n)_{n\in\setN}$ is non-decreasing for $\leq_{\vec{P}}$, we deduce that this sequence is non-decreasing for $\leq_{\vec{Q}\cap\vec{F}}$. We have proved that $\leq_{\vec{Q}\cap\vec{F}}$ is a wqo on $\extract{\vec{F}}{\vec{X}}{\vec{Y}}$.
\end{proof}

Now, let $\vec{F}_1,\ldots,\vec{F}_n$ be a linearization of the set of faces of $\vec{C}$ with respect to $\supseteq$. Let $\vec{X}\subseteq \vec{Y}$ be two finitary $\vec{Q}$-cylindric subsets of $\setZ^d$, and let $\vec{S}_1,\ldots,\vec{S}_n$ be subsets of $\setZ^d$ defined inductively for every $i\in \{1,\ldots,n\}$ as follows where $\vec{S}_{<i}\coloneqq\vec{S}_1\cup\ldots\cup\vec{S}_{i-1}$:
$$\vec{S}_i\coloneqq\extract{\vec{F}_i}{\vec{Y}\setminus \vec{S}_{<i}}{\vec{X}}$$
  
\medskip

We introduce the set $\mathcal{M}_i$ of faces $\vec{F}$ such that $\extract{\vec{F}}{\vec{Y}\setminus \vec{S}_{< i}}{\vec{X}}$ is non-empty. 
\begin{lemma}
  For every $i\in\{1,\ldots,n+1\}$ the set $\vec{Y}\setminus \vec{S}_{< i}$ is $\vec{Q}$-cylindric and $\mathcal{M}_i\subseteq \{\vec{F}_i,\ldots,\vec{F}_n\}$.
\end{lemma}
\begin{proof}  
  Let us prove the lemma by induction on $i\in \{1,\ldots,n+1\}$. The rank $i=1$ is trivial. Assume the rank $i$ proved for some $i\in \{1,\ldots,n\}$. Since $\vec{Y}\setminus \vec{S}_{<i}$ is $\vec{Q}$-cylindric, \cref{lem:cyclicpreserve} shows that $(\vec{Y}\setminus \vec{S}_{< i})\setminus \vec{S}_i$, i.e. $\vec{Y}\setminus\vec{S}_{<i+1}$ is $\vec{Q}$-cylindric. By construction observe that $\mathcal{M}_{i+1}\subseteq \mathcal{M}_{i}\setminus\{\vec{F}_i\}\subseteq \{\vec{F}_{i+1},\ldots,\vec{F}_n\}$. We have proved the induction.
\end{proof}

The $\vec{Q}$-cylindric set $\vec{Y}\setminus \vec{S}_{< i}$ is finitary since it is included in the finitary $\vec{Q}$-cylindric set $\vec{Y}$. Observe that if $\vec{S}_i$ is non-empty, since $\mathcal{M}_{i}\subseteq \{\vec{F}_i,\ldots,\vec{F}_n\}$, we deduce that $\vec{F}_i$ is a $\subseteq$-maximal face such that $\extract{\vec{F}_i}{\vec{Y}\setminus \vec{S}_{<i}}{\vec{X}}$ is non-empty. \cref{lem:extractB} shows that $\vec{S}_i$ is a finitary $(\vec{Q}\cap\vec{F}_i)$-cylindric set. Notice that this last property is also trivially true when $\vec{S}_i$ is empty. Clearly the sets $\vec{S}_1,\ldots,\vec{S}_n$ are pairwise-disjoint subsets of $\vec{Y}\setminus \vec{X}$. Moreover, since $\mathcal{M}_{n+1}=\emptyset$, \cref{lem:emptyset} shows that $(\vec{Y}\setminus\vec{S}_{< n+1})\setminus \vec{X}$ is empty. It follows that $\vec{S}_1,\ldots,\vec{S}_n$ is a disjoint decomposition of $\vec{Y}\setminus\vec{X}$.

\medskip

Finally, assume that $\vec{Y}\subseteq \vec{X}-\vec{Q}$ and let us prove that $\vec{S}_1$ is empty. Recall that $\vec{F}_1=\vec{C}$. It follows that if $\vec{S}_1$ is non-empty, there exists $\vec{s}\in\setZ^d$ such that $\vec{s}+\vec{Q}\subseteq \vec{Y}\setminus\vec{X}$. As $\vec{s}\in\vec{Y}\subseteq \vec{X}-\vec{Q}$, there exists $\vec{q}\in\vec{Q}$ such that $\vec{s}+\vec{q}\in \vec{X}$. We get a contradiction. Hence $\vec{S}_1$ is empty.

\medskip

We have proved \cref{thm:leaf}.

\section{From Safety Witnesses to Safe Semilinear Inductive Invariants}\label{sec:wrapup}
We are now equipped with the necessary ingredients to prove our main result,
namely that
every semilinear set $\vec{\Phi} \subseteq \setN^d$ containing the reachability set of an IBVAS $\IB$
also contains a semilinear inductive invariant for $\IB$
(see \cref{thm:main}).
Let us consider,
for the rest of this section,
an IBVAS $\IB \coloneqq (\vec{\Delta}_0, \B)$ and
a semilinear set $\vec{\Phi} \subseteq \setN^d$ such that $\reach{\IB} \subseteq \vec{\Phi}$.

\smallskip

We first introduce a well-founded relation $\Yleft$ on finite sets of directed iruns of $\IB$.
The set $\setN^{d+1}$ is ordered by the strict lexicographic order $<_{lex}$ defined by
$(m_0, \ldots, m_d) <_{lex} (n_0, \ldots, n_d)$ if $(m_0, \ldots, m_d) \neq (n_0, \ldots, n_d)$ and
the maximal\footnote{%
  In the usual definition of $<_{lex}$,
  the condition $m_i < n_i$ must hold for the \emph{minimal} $i$ such that $m_i \neq n_i$.
  In our definition,
  we consider the \emph{maximal} such $i$,
  as it is more convenient for our purposes.
}
$i$ such that $m_i \neq n_i$ satisfies $m_i < n_i$.
We associate with a directed irun $w \coloneqq (\rho_w, \vec{C}_w)$ the vector
$\rank{w} \coloneqq (n_0, \ldots, n_d) \in \setN^{d+1}$ defined by
$n_j = 1$ if $j = \dim{\vec{C}_w - \vec{C}_w}$, and
$n_j = 0$ otherwise.
For every finite set $W$ of directed iruns,
we let $\rank{W} \coloneqq \sum_{w \in W} \rank{w}$.
The binary relation $\Yleft$ on finite sets of directed iruns is defined by
$V \Yleft W$ if $\rank{V} <_{lex} \rank{W}$.
It is readily seen that $\Yleft$ is a well-founded relation
(i.e., there is no infinite sequence $W_0, W_1, W_2, \ldots$ with $W_{i+1} \Yleft W_i$ for all $i$).

\smallskip

We observe that for every semilinear set $\vec{A} \subseteq \setN^d$,
the pair $(\vec{A}, \emptyset)$ is a safety witness for $\IB$ and $\vec{\Phi}$ if,
and only if,
the set $\vec{A}$ is an inductive invariant for $\IB$ such that $\vec{A} \subseteq \vec{\Phi}$.
We have already shown in \cref{sec:big-picture} that there exists a safety witness for $\IB$ and $\vec{\Phi}$.
The remainder of this section proves the key property stated in \cref{lem:safety-witness-transformation}
(and already mentioned in \cref{sec:big-picture}).
As $\Yleft$ is well-founded,
this entails that a safety witness of the form $(\vec{A}, \emptyset)$ exists,
and concludes the proof of \cref{thm:main}.

\begin{lemma}
  \label{lem:safety-witness-transformation}
  For every safety witness $(\vec{A}, W)$ for $\IB$ and $\vec{\Phi}$,
  if $W \neq \emptyset$ then
  there exists a safety witness $(\vec{A}', W')$ for $\IB$ and $\vec{\Phi}$
  such that $W' \Yleft W$.
\end{lemma}

As we hinted before,
our proof of \cref{lem:safety-witness-transformation} uses \cref{lem:bvas-safelin,thm:leaf}.
We also rely on the two following technical lemmas.
The first one, \cref{lem:safety-witness-succ},
will help us show the premises of \cref{lem:bvas-safelin}.
The second one, \cref{lem:attractors-intersect-SCC},
provides a sufficient condition for an attractor to ``touch'' a directed irun.

\begin{lemma}
  \label{lem:safety-witness-succ}
  Let $(\vec{A}, W)$ be a safety witness for $\IB$ and $\vec{\Phi}$.
  For every directed irun $u \in W$,
  the set
  $\reach{\tgt{\IRuns{u}} \,{\shuffle}\, \vec{\Delta}_0^*}$
  is contained in $\vec{A} \cup \bigcup_{v \in V} \tgt{\IRuns{v}}$,
  where $V \coloneqq \{v \in W \mid u \edge v\}$.
\end{lemma}
\begin{proof}
  Consider a run $\beta$ with $\src{\beta} \in (\tgt{\IRuns{u}} \,{\shuffle}\, \vec{\Delta}_0^*)$.
  There exists an irun $\alpha \in \IRuns{u}$ such that
  $\src{\beta} \in (\{\tgt{\alpha}\} \,{\shuffle}\, \vec{\Delta}_0^*)$.
  By replacing $\tgt{\alpha}$ by $\alpha$ in $\beta$,
  we obtain an irun $\hat{\beta}$ that satisfies
  $\alpha \sqsubseteq \hat{\beta}$ and
  $\tgt{\hat{\beta}} = \tgt{\beta}$.
  Since $(\vec{A}, W)$ is a safety witness,
  we have $\tgt{\hat{\beta}} \in \vec{A}$ or $\hat{\beta} \in \IRuns{v}$ for some $v \in W$.
  In the latter case,
  we have $u \edge v$ since $\alpha \sqsubseteq \hat{\beta}$,
  $\alpha \in \IRuns{u}$ and $\hat{\beta} \in \IRuns{v}$,
  hence,
  $v \in V$.
  We conclude that
  $\tgt{\beta} = \tgt{\hat{\beta}}$ is in
  $\vec{A} \cup \bigcup_{v \in V} \tgt{\IRuns{v}}$.
\end{proof}

\begin{lemma}
  \label{lem:attractors-intersect-SCC}
  Let $\vec{Q} \subseteq \setN^d$ be a finitely-generated periodic set and
  let $\vec{A}$ be a $\vec{Q}$-cylindric attractor for $\IB$.
  For every directed iruns $u, v$ of $\IB$ such that
  $u \edge v$ and $\vec{Q}_u = \vec{Q}_v = \vec{Q}$,
  if $\vec{A}$ intersects $\tgt{\rho_u} + \vec{Q}$
  then $\vec{A}$ intersects $\tgt{\rho_v} + \vec{Q}$.
\end{lemma}
\begin{proof}
  Let $u$ and $v$ be directed iruns of $\IB$ such that
  $u \edge v$ and $\vec{Q}_u = \vec{Q}_v = \vec{Q}$.
  Assume that $\vec{A} \cap (\tgt{\rho_u} + \vec{Q})$ is not empty.
  There exists $\vec{q} \in \vec{Q}$ such that $(\tgt{\rho_u} + \vec{q}) \in \vec{A}$.
  We get from $u \edge v$ that $\alpha \sqsubseteq \beta$
  for some $\alpha \in \IRuns{u}$ and $\beta \in \IRuns{v}$.
  Put $\vec{x} \coloneqq \tgt{\alpha}$ and $\vec{y} \coloneqq \tgt{\beta}$ to reduce clutter.
  By definition,
  we have $\vec{x} \in \tgt{\rho_u} + \vec{P}_u$ and $\vec{y} \in \tgt{\rho_v} + \vec{P}_v$.
  It follows from $\vec{Q}_u = \vec{Q}_v = \vec{Q}$ that
  $\vec{x} \in \tgt{\rho_u} + \vec{Q}$ and $\vec{y} \in \tgt{\rho_v} + \vec{Q}$.
  Note that $\vec{y} + \vec{q}$ is in $\tgt{\rho_v} + \vec{Q}$.
  We now make two observations.
  Firstly, it holds that
  $(\vec{x} + \vec{q}) \in (\tgt{\rho_u} + \vec{q} + \vec{Q}) \subseteq \vec{A}$
  since $\vec{A}$ is $\vec{Q}$-cylindric.
  Secondly, we have
  $\vec{y} \in \reach{\{\vec{x}\} \,{\shuffle}\, \vec{\Delta}_0^*}$
  since $\vec{x} = \tgt{\alpha}$, $\vec{y} = \tgt{\beta}$ and $\alpha \sqsubseteq \beta$.
  It follows by monotony of BVAS runs (see \cref{fact:monotony}) that
  $(\vec{y} + \vec{q}) \in \reach{\{\vec{x} + \vec{q}\} \,{\shuffle}\, \vec{\Delta}_0^*} \subseteq \vec{A}$
  since $(\vec{x} + \vec{q}) \in \vec{A}$ and $\vec{A}$ is an attractor.
  We have shown that $\vec{y} + \vec{q}$ is in $\vec{A} \cap (\tgt{\rho_v} + \vec{Q})$.
\end{proof}

Let us now prove \cref{lem:safety-witness-transformation}.
Consider a safety witness $(\vec{A}, W)$ for $\IB$ and $\vec{\Phi}$, and
assume that $W \neq \emptyset$.
As in the overview presented at the end of \cref{sec:big-picture},
we pick an arbitrary bottom SCC $\Gamma$ of the graph $(W, \edge[W])$ and
we introduce the finitary $\vec{Q}_\Gamma$-cylindric set
$\vec{Y}_\Gamma \coloneqq \{\tgt{\rho_\gamma} \mid \gamma \in \Gamma\} + \vec{Q}_\Gamma$.
Put
$\Inv \coloneqq \vec{A} \cup \reach{\IB}$ and
$\vec{T} \equaldef \vec{A} \cup (\vec{Y}_\Gamma \cap \vec{\Phi})$.
Observe
that $\Inv$ is an inductive invariant for $\IB$ since
$\vec{A}$ is an attractor for $\IB$, and
that $\vec{T}$ is semilinear since $\vec{A}$, $\vec{Y}_\Gamma$ and $\vec{\Phi}$ are semilinear.
Pick a directed irun $w \in \Gamma$.
We will apply \cref{lem:bvas-safelin} to the set
$\tgt{\IRuns{w}} = \tgt{\rho_w} + \vec{P}_w$,
which is almost linear (see \cref{subsec:pavas}).
To do so,
we need to show that
$\reach{\tgt{\IRuns{w}} \,{\shuffle}\, \Inv^*} \subseteq \vec{T}$ and $\vec{T} \subseteq \Inv + \vec{Q}_w$.
The second inclusion follows from the observations that
$\vec{T} \subseteq \vec{A} \cup \vec{Y}_\Gamma$,
$\vec{Y}_\Gamma \subseteq \reach{\IB} + \vec{Q}_\Gamma$ and
$\vec{Q}_w = \vec{Q}_\Gamma$ since $w \in \Gamma$.
The first inclusion follows from the following relationships:
\begin{align*}
  \reach{\tgt{\IRuns{w}} \,{\shuffle}\, \Inv^*}
  & =
    \reach{\vec{A}^* \,{\shuffle}\, \tgt{\IRuns{w}} \,{\shuffle}\, \reach{\IB}^*}
  & [\Inv = \vec{A} \cup \reach{\IB}]
  \\
  & \subseteq
    \vec{A} \cup \reach{\tgt{\IRuns{w}} \,{\shuffle}\, \reach{\IB}^*}
  & [\vec{A} \text{ is an attractor}]
  \\
  & \subseteq
    \vec{A} \cup \reach{\tgt{\IRuns{w}} \,{\shuffle}\, \vec{\Delta}_0^*}
  & [\text{\cref{fact:reach-shuffle}}]
  \\
  & \subseteq
    \vec{A} \cup \bigcup_{\gamma \in \Gamma} \tgt{\IRuns{\gamma}}
  & [\text{\cref{lem:safety-witness-succ}}]
  \\
  & \subseteq
    \vec{A} \cup (\vec{Y}_\Gamma \cap \vec{\Phi})
  & [\text{see below}]
  \\
  & = \vec{T}
\end{align*}
The last inclusion holds because for every $\gamma \in \Gamma$,
we have
$\tgt{\IRuns{\gamma}} \subseteq \reach{\IB} \subseteq \vec{\Phi}$ and
$\tgt{\IRuns{\gamma}} = \tgt{\rho_\gamma} + \vec{P}_\gamma \subseteq \tgt{\rho_\gamma} + \vec{Q}_\gamma \subseteq \vec{Y}_\Gamma$.
According to \cref{lem:bvas-safelin},
there exists $p \in \vec{P}_w$ such that
$\reach{(\tgt{\rho_w} + \vec{p} + \vec{Q}_w)^+ \,{\shuffle}\, \Inv^*} \subseteq \vec{T}$.
We define
$\vec{R} \coloneqq \reach{(\tgt{\rho_w} + \vec{p} + \vec{Q}_w)^+ \,{\shuffle}\, \Inv^*}$ and
$\vec{X} \coloneqq \vec{R} \cap \vec{Y}_\Gamma$.
Let us show that the set $\vec{X}$ enjoys the properties claimed in the overview
presented at the end of \cref{sec:big-picture}.

\begin{itemize}
\item
  By monotony of BVAS runs (see \cref{fact:monotony}),
  the set $\vec{R}$ is $\vec{Q}_w$-cylindric.
  Recall that $\vec{Q}_w = \vec{Q}_\Gamma$ since $w \in \Gamma$.
  It follows that $\vec{X}$ is finitary $\vec{Q}_\Gamma$-cylindric since
  it is the intersection of
  the $\vec{Q}_\Gamma$-cylindric set $\vec{R}$ and
  the finitary $\vec{Q}_\Gamma$-cylindric set $\vec{Y}_\Gamma$.
  This entails, in particular, that $\vec{X}$ is semilinear.
\item
  We have $\vec{X} \subseteq \vec{\Phi}$ since
  $\vec{X} \subseteq \vec{R} \subseteq \vec{T} \subseteq \vec{A} \cup \vec{\Phi}$ and
  $\vec{A} \subseteq \vec{\Phi}$ as $(\vec{A}, W)$ is a safety witness.
\item
  The set $\vec{A} \cup \vec{R}$ is an attractor for $\IB$ by \cref{lem:incremental-attractor}.
  Since $\vec{R} \subseteq \vec{T} \subseteq \vec{A} \cup \vec{Y}_\Gamma$,
  we derive that
  $\vec{A} \cup \vec{R} = \vec{A} \cup (\vec{R} \cap \vec{Y}_\Gamma) = \vec{A} \cup \vec{X}$.
  Therefore,
  the set $\vec{A} \cup \vec{X}$ is an attractor for $\IB$.
\item
  The set $\vec{R}$ is an attractor for $\IB$ by \cref{lem:incremental-attractor}
  (applied with $\vec{A} \coloneqq \emptyset$).
  As mentioned above, $\vec{R}$ is $\vec{Q}_\Gamma$-cylindric.
  Note that $\vec{R}$ intersects $\tgt{\rho_w} + \vec{Q}_\Gamma$ since
  $\vec{R}$ contains $\tgt{\rho_w} + \vec{p}$ and $\vec{p} \in \vec{P}_w \subseteq \vec{Q}_w = \vec{Q}_\Gamma$.
  It follows from \cref{lem:attractors-intersect-SCC} that
  $\tgt{\rho_\gamma}$ is in $\vec{R} - \vec{Q}_\Gamma$ for every $\gamma \in \Gamma$.
  We derive from $\vec{R} + \vec{Q}_\Gamma \subseteq \vec{R}$ and $\vec{X} = \vec{R} \cap \vec{Y}_\Gamma$ that
  $\vec{Y}_\Gamma \subseteq \vec{X} - \vec{Q}_\Gamma$.
\item
  For every $\vec{r} \in \vec{R}$ and $\vec{y} \in \setN^d$ such that
  $\vec{r} \bvasrel{\IB} \vec{y}$,
  we have $\vec{y} \in \reach{\vec{R} \,{\shuffle}\, \vec{\Delta}_0^*}$,
  hence,
  $\vec{y} \in \vec{R}$ as $\vec{R}$ is an attractor (see above).
  Since $\vec{X} = \vec{R} \cap \vec{Y}_\Gamma$,
  we derive that the binary relation $\bvasrel{\IB}$ has an empty intersection with
  $\vec{X} \times (\vec{Y}_\Gamma \setminus \vec{X})$.
\end{itemize}

To summarize what we have done so far,
we have constructed a finitary $\vec{Q}_\Gamma$-cylindric set $\vec{X} \subseteq \vec{Y}_{\Gamma}$ such that
firstly $\vec{A}' \coloneqq \vec{A} \cup \vec{X}$ is a semilinear attractor for $\IB$ contained in $\vec{\Phi}$,
secondly $\vec{Y}_\Gamma \subseteq \vec{X} - \vec{Q}_\Gamma$, and
thirdly $\bvasrel{\IB}$ has an empty intersection with $\vec{X} \times (\vec{Y}_\Gamma \setminus \vec{X})$.

Before introducing the set $W'$ of directed iruns, we first prove the following lemma.
\begin{lemma}\label{lem:insideGammaQ}
  For every directed iruns $w$ and for every iruns $\rho\in \IRuns{w}$, we have:
  \begin{itemize}
  \item $\IRuns{(\rho,\vec{C}_w)}\subseteq \IRuns{w}$, and
  \item $\vec{P}_\rho\cap\vec{C}_w\subseteq \vec{Q}_w$.
  \end{itemize}
\end{lemma}
\begin{proof}
  Let $u\coloneqq (\rho,\vec{C}_w)$ and $\sigma\in \IRuns{u}$. We have $\rho\trianglelefteq\sigma$ and $\tgt{\sigma}\in \tgt{\rho}+\vec{C}_w$. Moreover notice that $\rho_w \trianglelefteq \rho$ and $\tgt{\rho}\in \tgt{\rho_w}+\vec{C}_w$. We deduce that $\tgt{\sigma}\in \tgt{\rho_w}+\vec{C}_w$. By transitivity of $\trianglelefteq$ we derive $\rho_w\trianglelefteq\sigma$. Hence $\sigma\in \IRuns{w}$ and we have prove the first inclusion.
  
  For the second inclusion, notice that $\vec{C}_u=\vec{C}_w$, $\rho\sqsubseteq\rho$, $\rho\in \IRuns{u}$ and $\rho\in \IRuns{w}$. \cref{lem:Qinclusion} shows that $\vec{Q}_u\subseteq \vec{Q}_w$. Since $\vec{P}_u\subseteq\vec{Q}_u$ and $\vec{P}_u=\vec{P}_\rho\cap\vec{C}_w$ we deduce the inclusion $\vec{P}_\rho\cap\vec{C}_w\subseteq \vec{Q}_w$.
\end{proof}

Let us recall that $\vec{C}_\Gamma$ is the finitely-generated cone satisfying $\vec{C}_\gamma=\vec{C}_\Gamma$ for every $\gamma\in \Gamma$. We introduce the finitely-generated cone $\vec{C}\coloneqq\cone{\vec{Q}_\Gamma}$. From $\vec{Q}_\Gamma\subseteq \vec{C}_\Gamma$ we derive $\vec{C}\subseteq \vec{C}_\Gamma$. Let us fix a linearization $\vec{F}_1,\ldots,\vec{F}_n$ of the set of faces of $\vec{C}$ with respect to $\supseteq$. It follows that $n=|\mathcal{F}(\vec{C})|$, $\mathcal{F}(\vec{C})=\{\vec{F}_1,\ldots,\vec{F}_n\}$, $\vec{F}_1=\vec{C}$, and $\vec{F}_i\supseteq \vec{F}_j$ implies $i\leq j$.
We introduce the sequence $\vec{S}_1,\ldots,\vec{S}_n$ of subsets of $\setZ^d$ defined inductively for every $i\in \{1,\ldots,n\}$ as follows where $\vec{S}_{<i}\coloneqq\vec{S}_1\cup\ldots\cup\vec{S}_{i-1}$:
  $$\vec{S}_i\coloneqq\{\vec{s}\in\setZ^d \mid \vec{s}+(\vec{Q}\cap\vec{F}_i)\subseteq (\vec{Y}\setminus\vec{S}_{<i})\setminus \vec{X}\}$$
Since $\vec{X} \subseteq \vec{Y}_{\Gamma}$ are finitary $\vec{Q}_\Gamma$-cylindric, 
\cref{thm:leaf} shows that $\vec{S}_1,\ldots,\vec{S}_n$ is a disjoint decomposition of $\vec{Y}_\Gamma\setminus\vec{X}$ such that $\vec{S}_i$ is a finitary $(\vec{Q}\cap\vec{F}_i)$-cylindric set for every $1\leq i\leq n$. Moreover, since $\vec{Y}_\Gamma\subseteq\vec{X}-\vec{Q}_\Gamma$, we deduce that $\vec{S}_1$ is empty.

\smallskip

Let $i\in\{2,\ldots,n\}$ and let us introduce the finitely-generated periodic set $\vec{Q}_i\coloneqq\vec{Q}_\Gamma\cap \vec{F}_i$ and the set $R_i\coloneqq\{\rho\in \IRuns{\Gamma}\mid \tgt{\rho}\in \vec{S}_i\}$. We denote by $\leq_i$ the partial order $\leq_{\vec{Q}_i}$. Since $\vec{Q}_i$ is a finitely-generated periodic set and $\vec{S}_i$ is a finitary $\vec{Q}_i$-cylindric set, we deduce that $\leq_i$ is a wpo on $\vec{S}_i$. We also introduce the partial order $\trianglelefteq_i$ over $R_i$ defined by $\rho\trianglelefteq_i\sigma$ if $\rho\trianglelefteq \sigma$ and $\tgt{\rho}\leq_i\tgt{\sigma}$. Since $\trianglelefteq$ is a wpo on $\IRuns{\IB}$ and $\leq_i$ is a wpo on $\vec{S}_i$, we deduce that $\trianglelefteq_i$ is a wpo on $R_i$. We introduce the finite set $W_i$ of directed iruns $(\rho,F_i)$ where $\rho\in\Min{R_i}{\trianglelefteq_i}$. We are going to prove that the pair $(\vec{A}',W')$ is a safety witness satisfying \cref{lem:safety-witness-transformation} where $W'$ is the following set of directed iruns.
$$W'\coloneqq (W\setminus \Gamma)\cup \bigcup_{i=2}^n W_i$$

Let us first prove that $W' \Yleft W$. We introduce $i_\Gamma\coloneqq\dim{\vec{C}_\Gamma-\vec{C}_\Gamma}$, and the natural numbers $n_0,\ldots,n_d,n_0',\ldots,n_d'$ such that $\rank{W}=(n_0,\ldots,n_d)$ and $\rank{W'}=(n_0',\ldots,n_d')$. Recall from \cref{sec:stripping} that $\dim{\vec{F}_i-\vec{F}_i}<\dim{\vec{C}-\vec{C}}$ for every $i\in\{2,\ldots,n\}$. Moreover, since $\vec{C}\subseteq \vec{C}_\Gamma$, we deduce that $\dim{\vec{F}_i-\vec{F}_i}<i_\Gamma$. It follows that $n'_{i_\Gamma}=n_{i_\Gamma}-|\Gamma|$ where $|\Gamma|$ denotes the cardinal of $\Gamma$, and $n'_i=n_i$ for every $i\in \{i_\Gamma+1,\ldots,d\}$. We deduce that $\rank{W'} <_{lex} \rank{W}$. Therefore $W' \Yleft W$ as claimed.

Since $W$ is homogeneous, we deduce from the following lemma that $W'$ is homogeneous as well.
\begin{lemma}
  For every $i\in\{2,\ldots,n\}$, for every $u\in W_i$, and for every $v\in W'$ such that $u\edge v$, there exists $j\in\{i,\ldots,n\}$ such that $v\in W_j$.
\end{lemma}
\begin{proof}
  Clearly $\bvasrel{\IB}$ is a $\setN^d$-diagonal relation\footnote{See \cref{par:diagonal} for the definition of diagonal relations.}, and in particular a $\vec{Q}_\Gamma$-diagonal relation since $\vec{Q}_\Gamma\subseteq \setN^d$. Moreover, as $\bvasrel{\IB}$ has an empty intersection with $\vec{X}\times(\vec{Y}_\Gamma\setminus \vec{X})$, \cref{lem:diagonalstripping} shows that $\bvasrel{\IB}$ has an empty intersection with $\vec{S}_i\times\vec{S}_j$ for every $1\leq j<i\leq k$.

  Now, let $u\in W_i$ for some $i\in \{2,\ldots,n\}$ and $v\in W'$ such that $u\edge v$. It follows that $\alpha\sqsubseteq\beta$ for some iruns $\alpha\in \IRuns{u}$ and $\beta\in \IRuns{v}$.
  
  Assume by contradiction that $v\in W\setminus \Gamma$. Since $u\in W_i$, we deduce that $u=(\rho_u,\vec{F}_i)$ where $\rho_u\in \IRuns{\gamma}$ for some $\gamma\in\Gamma$. As $\vec{F}_i\subseteq\vec{C}$, we derive from \cref{lem:insideGammaQ} that $\IRuns{u}\subseteq \IRuns{\gamma}$. It follows that $\alpha\in \IRuns{\gamma}$. From $\alpha\in \IRuns{\gamma}$, $\beta\in \IRuns{v}$, and $\alpha\sqsubseteq\beta$, we deduce that $\gamma\edge[W] v$. As $\Gamma$ is a bottom SCC of $W$, $\gamma\in\Gamma$, and $v\in W$, we deduce that $v\in \Gamma$ and we get a contradiction. It follows that there exists $j\in \{2,\ldots,n\}$ such that $v\in W_j$.

  Since $\alpha\in \IRuns{u}$ we deduce that $\tgt{\rho_u}\trianglelefteq_i\tgt{\alpha}$. In particular $\tgt{\alpha}\in \tgt{\rho_u}+(\vec{P}_{\rho_u}\cap \vec{F}_i)$. Since $\vec{F}_i\subseteq \vec{C}\subseteq\vec{C}_\Gamma$, we derive from \cref{lem:insideGammaQ} that $\vec{P}_{\rho_u}\cap\vec{F}_i\subseteq\vec{Q}_\Gamma\cap\vec{F}_i$. As $\tgt{\rho_u}\in\vec{S}_i$ and $\vec{S}_i$ is $(\vec{Q}_\Gamma\cap\vec{F}_i)$-cylindric, we deduce that $\tgt{\alpha}\in\vec{S}_i$. Symmetrically, since $\beta\in\IRuns{v}$, we deduce that $\tgt{\beta}\in \vec{S}_j$. From $\alpha\sqsubseteq \beta$, we deduce that $\tgt{\alpha}\bvasrel{\IB}\tgt{\beta}$. As $\tgt{\alpha}\in\vec{S}_i$ and $\tgt{\beta}\in\vec{S}_j$, we deduce from the first paragraph that $i\leq j$.
\end{proof}

In order to prove that $(\vec{A}',W')$ is a safety witness for $\IB$ and $\vec{\Phi}$, it just remain to prove that for every irun $\rho \in \IRuns{\IB}$, we have $\tgt{\rho} \in \vec{A}'$ or $\rho \in \IRuns{W'}$. So, let $\rho\in \IRuns{\IB}$ such that $\tgt{\rho}\not\in \vec{A}'$. As $(\vec{A},W)$ is a safety witness and $\tgt{\rho}\not\in\vec{A}$ we deduce that there exists $w\in W$ such that $\rho\in \IRuns{w}$. Clearly, if $w\not\in \Gamma$ then $w\in W'$ and in particular $\rho\in \IRuns{W'}$. So, let us assume that $\rho\in \IRuns{\Gamma}$. It follows that $\tgt{\rho}\in \vec{Y}_\Gamma$. If $\tgt{\rho}\in \vec{X}$ then $\tgt{\rho}\in\vec{A}'$, so we can assume that $\tgt{\rho}\not\in \vec{X}$. Hence $\tgt{\rho}\in \vec{Y}_\Gamma\setminus\vec{X}$. As $\vec{S}_2,\ldots,\vec{S}_n$ is a disjoint decomposition of $\vec{Y}_\Gamma\setminus \vec{X}$, there exists $i\in\{2,\ldots,n\}$ such that $\tgt{\rho}\in \vec{S}_i$. Hence $\rho\in R_i$. We deduce that there exists $u\in W_i$ such that $\rho_u\trianglelefteq_i \rho$. In particular $\rho_u\trianglelefteq \rho$ and $\tgt{\rho}\in \tgt{\rho_u}+\vec{F}_i$. We deduce that $\rho\in\IRuns{u}\subseteq \IRuns{W'}$. We have proved that $(\vec{A}',W')$ is a safety witness for $\IB$ and $\vec{\Phi}$. We have proved \cref{lem:safety-witness-transformation}.

\section{Conclusion}\label{sec:conclusion}

We have established the decidability of reachability in branching vector addition systems, resolving a problem that has remained open for more than three decades.

Several questions remain open. First, our argument does not provide any complexity upper bound. For plain VAS, a complexity analysis of the invariant-based approach is known but it is based on the computational complexity of the KLM algorithm. Determining the complexity of BVAS reachability therefore appears to require new ideas, and perhaps fundamentally different algorithms. At present, even in dimension five, where decidability was previously established~\cite{DBLP:conf/fossacs/BiziereLS26}, no non-trivial complexity upper bound is known. Likewise, the best lower bound currently known is the Ackermann lower bound inherited from VAS reachability.

A second natural direction concerns extended BVAS (EBVAS), introduced in~\cite{jacquemard:hal-00769249} as a counter automaton equivalent to two-variable first-order logic over data trees, and later connected in~\cite{DBLP:conf/esop/Cotton-BarrattM17} to the observational equivalence problem for a finitary fragment of ML. The model of EBVAS enriches branching transitions with additional resource-transfer mechanisms, and is believed to be strictly more expressive than BVAS. The decidability of reachability for EBVAS remains open and, to the best of our knowledge, has not yet been systematically investigated. From the perspective of the present work, a first challenge would be to determine whether EBVAS have almost semilinear reachability set and whether this could be proved with the same approach as in~\cite{DBLP:conf/fossacs/BiziereLS26}.

\bibliographystyle{ACM-Reference-Format}
\bibliography{biblio}

\clearpage
\appendix

\section{Proof of \cref{lem:fgper-wqo-on-Zd}}
\label{app:lem:fgper-wqo-on-Zd}
\lemFgperWqoOnZd*

\begin{proof}
  If $\vec{P}$ is a finitely-generated periodic set, there exists a finite sequence $\vec{p}_1,\ldots,\vec{p}_k\in\setZ^d$ such that $\vec{P}=\setN\vec{p}_1+\cdots+\setN\vec{p}_k$. Since $\leq$ is a wqo on $\setN^k$ (Dickson's Lemma), we deduce that $\leq_{\vec{P}}$ is a wqo on~$\vec{P}$. Conversely, assume that $\leq_{\vec{P}}$ is a wqo on $\vec{P}$. We introduce the mapping $f_{\vec{s}}:\setZ^d\rightarrow \setZ^d$ where $\vec{s}\in\{-1,1\}^d$ defined by $f_{\vec{s}}(\vec{z})=(\vec{s}(1)\vec{z}(1),\ldots,\vec{s}(d)\vec{z}(d))$. The mapping is bijective since $f_{\vec{s}}\circ f_{\vec{s}}$ is the identity. Notice that $\vec{Q}_{\vec{s}}\coloneqq f_{\vec{s}}(\vec{P})\cap\setN^d$ is a periodic set included in $\setN^d$ such that $\leq_{\vec{Q}_s}$ is a wpo on $\vec{Q}_{\vec{s}}$. We introduce the finite set $\vec{M}_{\vec{s}}\coloneqq\Min{\vec{Q}_{\vec{s}}\setminus\{\vec{0}\}}{\leq_{\vec{Q}_{\vec{s}}}}$. By induction on the norm $\sum_{i=1}^d\vec{q}$ where $\vec{q}\in\vec{Q}_{\vec{s}}$, we derive that $\vec{Q}_{\vec{s}}=\per{\vec{M}_{\vec{s}}}$. Now, just observe that $\vec{M}\coloneqq\bigcup_{\vec{s}}f_{\vec{s}}(\vec{M}_{\vec{s}})$ satisfies $\vec{P}=\per{\vec{M}}$.
\end{proof}

\section{Well-Partial-Order on Iruns of an Initialized BVAS and Amalgamation Property}
\label{app:lem:ibvas-wpo-and-amalgamation}
This appendix provides a detailed proof of \cref{lem:ibvas-wpo-and-amalgamation}.
As mentioned in \cref{sec:big-picture},
this lemma was proved in~\cite{DBLP:conf/fossacs/BiziereLS26} with slightly different notations.
We reprove it because our approach crucially relies on it.

\lemIbvasWpoAndAmalgamation*

Let us consider an IBVAS $\IB \equaldef (\vec{\Delta}_0, \B)$ where
$\B \coloneqq (\vec{\Delta}_1, \ldots, \vec{\Delta}_r)$.
To prove \cref{lem:ibvas-wpo-and-amalgamation},
we show
that $(\IRuns{\IB}, \trianglelefteq)$ is a wpo (see \cref{lem:ibvas-run-order-wpo}) and
that it satisfies the amalgamation property (see \cref{lem:ibvas-amalgamation}).

\begin{lemma}
  \label{lem:ibvas-run-order-wpo}
  The pair $(\IRuns{\IB}, \trianglelefteq)$ is a wpo.
\end{lemma}

\begin{proof}
  It is straightforward to verify that $(\IRuns{\IB},\trianglelefteq)$ is a poset. We therefore only need to show that it is a wqo.
  To show that it is a wqo, we will apply Kruskal's lemma to ``decorated'' iruns.
  We first recall succinctly that lemma (see~\cite{polyWQO} for more details).

  A \emph{tree} $t$ on a set $A$ is defined inductively as a pair $(a,(t_1,\ldots, t_n))$ where $a\in A$, $n\in\setN$, and $t_1,\ldots,t_n$ is a sequence of trees on $A$. Such a pair is simply denoted by $a\langle t_1,\ldots,t_n\rangle$, and we denote by $T(A)$ the set of trees on $A$. Let $\preceq$ be a quasi-order on $A$. We introduce the quasi-order $\preceq_T$ on $T(A)$ inductively on the structure of a tree $t$ by $s\preceq_T t$ where $s\coloneqq a\langle s_1,\ldots,s_n\rangle$ and $t\coloneqq b\langle t_1,\ldots,t_m\rangle$ if (1) there exists $i\in\{1,\ldots,m\}$ such that $s\preceq_T t_i$, or (2) $a\preceq b$ and there exists a sequence $1\leq i_1<\ldots<i_n\leq m$ such that $s_1\preceq_T t_{i_1}\wedge \ldots s_n\preceq_T t_{i_n}$. By the Kruskal theorem, we deduce that if $\preceq$ is wqo on $A$ then $\preceq_T$ is wqo on $T(A)$.

  Let $\Lambda = \bigcup_{n=1}^r \setN^d \times \Delta_n \times (\setN^d)^n$. We define a wqo $\preceq$ on $\Lambda$ by $(\vec x, \delta, \vec y) \preceq (\vec x', \delta', \vec y')$ if $\vec x \le \vec x'$, $\delta = \delta'$ and $\vec y \le \vec y'$ belong to the same set $(\setN^d)^n$. Clearly $\preceq$ is a wqo on $\Lambda$ and we deduce that $\preceq_T$ is a wqo on $T(\Lambda)$.

  We introduce the (decoration) function $f:\IRuns{\IB} \to T(\Lambda)$ defined by induction on iruns $\rho \coloneqq (\vec c, (\rho_1,...,\rho_n))$ by $f(\rho)\coloneqq \lambda \langle f(\rho_1),...,f(\rho_n) \rangle$ with $\lambda \coloneqq (\vec c, \act{\rho}, (\tgt{\rho_1},...,\tgt{\rho_n}))$. By structural induction on $\sigma$, notice that $\rho\trianglelefteq \sigma$ if, and only if, $f(\rho) \preceq_{T(\Lambda)} f(\sigma)\wedge \tgt{\rho}\leq \tgt{\sigma}$. It follows that $\trianglelefteq$ is a wpo on $\IRuns{\IB}$.
\end{proof}

\begin{lemma}
  \label{lem:ibvas-amalgamation}
  For every iruns $\rho, \alpha, \beta$ such that $\rho \trianglelefteq \alpha, \beta$,
  there exists an irun $\sigma$ such that
  $\alpha, \beta \trianglelefteq \sigma$ and
  $\tgt{\rho} + \tgt{\sigma} = \tgt{\alpha} + \tgt{\beta}$.
\end{lemma}

\begin{proof}
  We now prove the amalgamation property by structural induction on $\rho$.
  Let $\rho \trianglelefteq \alpha, \beta$ be three iruns. Following the definition of $\trianglelefteq$, there are subruns $\alpha' \sqsubseteq \alpha$ and $\beta' \sqsubseteq \beta$ such that 
  $\tgt{\rho} \le \tgt{\alpha}, \tgt{\alpha'}, \tgt{\beta}, \tgt{\beta'}$,
  $\ra{\rho} = \ra{\alpha'} = \ra{\beta'}$,
  $\act{\rho} = \act{\alpha'} = \act{\beta'}$
  and $\rho[j] \trianglelefteq \alpha'[j], \beta'[j]$ for every $j$.
  Define $n \coloneqq \ra{\rho}$.
  By induction hypothesis, for each $j \in [1,n]$ there is an irun $\sigma_j$ such that $\alpha'[j], \beta'[j] \trianglelefteq \sigma_j$ and $\tgt{\rho[j]} + \tgt{\sigma_j} = \tgt{\alpha'[j]} + \tgt{\beta'[j]}$.
  Summing these equalities over $j=1,\ldots,n$
  and adding $\act{\rho}$ yields
  \begin{equation} \label{eq:sum-amalg}
  \tgt{\alpha'} + \tgt{\beta'} = \tgt{\rho} + \act{\rho} + (\sum_{j=1}^n \tgt{\sigma_j}) 
  \end{equation}

  Let $\alpha''$ be a run obtained from $\alpha$ by replacing a subrun $\alpha'$ by just a leaf labeled by $\tgt{\alpha'}$.
  Since $\tgt{\beta'} \ge \tgt{\rho}$, by \cref{fact:monotony}, there is a run $\hat{\alpha}$ with $\src{\hat{\alpha}} \in \Delta_0^* (\tgt{\alpha'} + \tgt{\beta'} - \tgt{\rho}) \Delta_0^*$ and $\tgt{\hat{\alpha}} = \tgt{\alpha} + \tgt{\beta'} - \tgt{\rho}$.
  By a similar argument, we construct a run $\hat{\beta}$ such that $\src{\hat{\beta}} \in \Delta_0^* (\tgt{\beta'} + \tgt{\alpha} - \tgt{\rho}) \Delta_0^*$ and $\tgt{\hat{\beta}} = \tgt{\beta} + \tgt{\alpha} - \tgt{\rho}$.

  Let $\sigma'$ be the irun $\left(\act{\rho} + \sum_{j=1}^n \tgt{\sigma_j}, (\sigma_1,...,\sigma_n) \right)$. 
  Using \eqref{eq:sum-amalg}, we can replace in $\hat{\alpha}$ a leaf labeled $\tgt{\alpha'}+\tgt{\beta'}-\tgt{\rho}$ by the irun $\sigma'$. The resulting irun is denoted by $\tau$.
  The irun $\sigma$ is obtained by similarly replacing in $\hat{\beta}$ a leaf labeled $\tgt{\beta'} + \tgt{\alpha} - \tgt{\rho}$ by $\tau$.

  There only remains to show that $\alpha, \beta \trianglelefteq \sigma$. It is readily seen that $\beta' \trianglelefteq \tau$ (by considering the subrun $\sigma' \sqsubseteq \tau$). We conclude that $\beta \trianglelefteq \sigma$ by induction, using the following observation.
  \begin{fact}
    Let $n \in \setN$ and $\eta \coloneqq (\vec c, (\eta_1,...,\eta_n)), \theta \coloneqq (\vec c, (\theta_1,...,\theta_n))$ be two runs. If $\eta_j \trianglelefteq \theta_j$ for all $j \in [1,n]$, then $\eta \trianglelefteq \theta$.
  \end{fact}
  The same fact allows us to show that $\alpha \trianglelefteq \tau$. We deduce that $\alpha \trianglelefteq \sigma$ thanks to the following fact.
  \begin{fact}
    Let $\eta, \theta$ be two runs. If $\tgt{\eta} \le \tgt{\theta}$ and $\eta \trianglelefteq \theta'$ for some $\theta' \sqsubseteq \theta$, then $\eta \trianglelefteq \theta$.
  \end{fact}
\end{proof}

\section{Well-Quasi-Order on Runs of a Well-Structured VAS and Amalgamation Property}
\label{app:wsvas}
This appendix provides the detailed proof of \cref{lem:wsvas-wqo-and-amalgamation}.
As mentioned in \cref{subsec:wsvas},
we essentially lift to WSVAS existing proof arguments for VAS~\cite{DBLP:journals/tcs/Jancar90,DBLP:conf/popl/Leroux11,LS15}.

\lemWsvasWqoAndAmalgamation*

Let us consider a WSVAS $\V \equaldef (T, \sqsubseteq, \delta)$.
To prove \cref{lem:wsvas-wqo-and-amalgamation},
we show
that $(\Runs{\V}, \trianglelefteq)$ is a wqo (see \cref{cor:wsvas-run-order-wqo}) and
that it satisfies the amalgamation property (see \cref{lem:wsvas-amalgamation}).

\medskip

The binary relation $\trianglelefteq$ can be characterized in terms of subsequence embedding.
Recall that for every qoset $(\Sigma, \preceq)$,
the \emph{subsequence embedding} is the quasi-order $\preceq^*$ on $\Sigma^*$ defined by
$u \preceq^* v$ if $u = a_1 \cdots a_k$ and $v = v_0 b_1 v_1 \cdots b_k v_k$ for some
$v_i \in \Sigma^*$ and $a_i, b_i \in \Sigma$ such that $a_i \preceq b_i$.
To characterize $\trianglelefteq$ in terms of $\preceq^*$,
we consider the alphabet $\Sigma = \setN^d \times T \times \setN^d$ equipped with the
quasi-order $\preceq$ on $\Sigma$ defined by $(\vec{x}, t, \vec{y}) \preceq (\vec{x}', t', \vec{y}')$ if
$\vec{x} \leq \vec{x}'$, $t \sqsubseteq t'$ and $\vec{y} \leq \vec{y}'$.
The \emph{encoding} of a run $\rho \equaldef (\vec{c}_0, t_1, \vec{c}_1, \ldots, t_k, \vec{c}_k)$,
written $\enc{\rho}$,
is the sequence in $\Sigma^*$ defined by
$\enc{\rho} \equaldef (\vec{c}_0, t_1, \vec{c}_1) \cdots (\vec{c}_{k-1}, t_k, \vec{c}_k)$.
Note that $\enc{\rho}$ is the empty sequence if $\rho$ has length zero.

\begin{lemma}
  \label{lem:wsvas-run-order-charac}
  For every runs $\rho$ and $\sigma$,
  it holds that $\rho \trianglelefteq \sigma$
  if, and only if,
  $\dir{\rho} \leq \dir{\sigma}$ and
  $\enc{\rho} \preceq^* \enc{\sigma}$.
\end{lemma}
\begin{proof}
  For the ``only if'' direction,
  assume that
  $\rho = (\vec{c}_0, t_1, \vec{c}_1, \ldots, t_k, \vec{c}_k)$ and
  $\sigma = \sigma_0 u_1 \sigma_1 \cdots u_k \sigma_k$
  for
  some configurations $\vec{c}_0, \ldots, \vec{c}_k$,
  some runs $\sigma_0, \ldots, \sigma_k$ and
  some transitions $t_1, \ldots, t_k, u_1, \ldots, u_k$
  such that
  $\vec{c}_i \leq \src{\sigma_i}, \tgt{\sigma_i}$ for all $i \in \{0, \ldots, k\}$ and
  $t_i \sqsubseteq u_i$ for all $i \in \{1, \ldots, k\}$.
  We have
  $\src{\rho} = \vec{c}_0 \leq \src{\sigma_0} = \src{\sigma}$ and
  $\tgt{\rho} = \vec{c}_k \leq \tgt{\sigma_k} = \tgt{\sigma}$,
  hence,
  $\dir{\rho} \leq \dir{\sigma}$.
  Observe firstly that
  $\enc{\sigma} = \enc{\sigma_0} \, (\tgt{\sigma_0}, u_1, \src{\sigma_1}) \, \enc{\sigma_1} \cdots (\tgt{\sigma_{k-1}}, u_k, \src{\sigma_k}) \, \enc{\sigma_k}$,
  and secondly that
  $(\vec{c}_{i-1}, t_i, \vec{c}_i) \preceq (\tgt{\sigma_{i-1}}, u_i, \src{\sigma_i})$
  for all $i \in \{1, \ldots, k\}$.
  It follows that the encodings of $\rho$ and $\sigma$ satisfy
  $\enc{\rho}
  =
  (\vec{c}_0, t_1, \vec{c}_1) \cdots (\vec{c}_{k-1}, t_k, \vec{c}_k)
  \preceq^*
  (\tgt{\sigma_0}, u_1, \src{\sigma_1}) \cdots (\tgt{\sigma_{k-1}}, u_k, \src{\sigma_k})
  \preceq^*
  \enc{\sigma}$.

  \smallskip

  Conversely,
  assume that $\dir{\rho} \leq \dir{\sigma}$ and $\enc{\rho} \preceq^* \enc{\sigma}$.
  Let us write $\rho$ as $\rho = (\vec{c}_0, t_1, \vec{c}_1, \ldots, t_k, \vec{c}_k)$.
  We derive from $\enc{\rho} \preceq^* \enc{\sigma}$ that $\enc{\sigma}$ may be written as
  $\enc{\sigma} = w_0 (\vec{x}_1, u_1, \vec{y}_1) w_1 \cdots (\vec{x}_k, u_k, \vec{y}_k) w_k$
  for some sequences $w_0, \ldots, w_k \in \Sigma^*$ and some triples $(\vec{x}_1, u_1, \vec{y}_1), \ldots, (\vec{x}_k, u_k, \vec{y}_k) \in \Sigma$ such that
  $(\vec{c}_{i-1}, t_i, \vec{c}_i) \preceq (\vec{x}_i, u_i, \vec{y}_i)$ for all $i \in \{1, \ldots, k\}$.
  It follows that $\sigma$ may be written as
  $\sigma = \sigma_0 u_1 \sigma_1 \cdots u_k \sigma_k$
  with $w_i = \enc{\sigma_i}$ for all $i \in \{0, \ldots, k\}$.
  Since
  $\enc{\sigma} = \enc{\sigma_0} \, (\vec{x}_1, u_1, \vec{y}_1) \, \enc{\sigma_1} \cdots (\vec{x}_k, u_k, \vec{y}_k) \, \enc{\sigma_k}$,
  we get that
  $\src{\sigma} = \src{\sigma_0}$,
  $\tgt{\sigma_{i-1}} = \vec{x}_i$ and $\vec{y}_i = \src{\sigma_i}$ for all $i \in \{1, \ldots, k\}$, and
  $\tgt{\sigma_k} = \tgt{\sigma}$.
  Recall that
  $\vec{c}_0 = \src{\rho} \leq \src{\sigma}$,
  $\vec{c}_k = \tgt{\rho} \leq \tgt{\sigma}$ and
  $(\vec{c}_{i-1}, t_i, \vec{c}_i) \preceq (\vec{x}_i, u_i, \vec{y}_i)$ for all $i \in \{1, \ldots, k\}$.
  It follows that
  $\vec{c}_i \leq \src{\sigma_i}, \tgt{\sigma_i}$ for all $i \in \{0, \ldots, k\}$ and
  $t_i \sqsubseteq u_i$ for all $i \in \{1, \ldots, k\}$.
  We have shown that $\rho \trianglelefteq \sigma$.
\end{proof}

\begin{corollary}
  \label{cor:wsvas-run-order-wqo}
  The pair $(\Runs{\V}, \trianglelefteq)$ is a wqo.
\end{corollary}

\begin{proof}
  The pair $(\setN^d, \leq)$ is a wqo by Dickson's Lemma.
  Since $(T, \sqsubseteq)$ is a wqo,
  we get that $(\Sigma, \preceq)$ is also a wqo.
  It follows from Higman's Lemma that $(\Sigma^*, \preceq^*)$ is a wqo.
  The corollary then follows from \cref{lem:wsvas-run-order-charac}.
\end{proof}

\begin{corollary}
  \label{cor:wsvas-run-order-decomp}
  Let $\rho$ and $\sigma$ be two runs such that
  $\rho = \rho_0 t_1 \rho_1 \cdots t_k \rho_k$ and
  $\sigma = \sigma_0 u_1 \sigma_1 \cdots u_k \sigma_k$
  for
  some runs $\rho_0, \ldots, \rho_k, \sigma_0, \ldots, \sigma_k$ and
  some transitions $t_1, \ldots, t_k, u_1, \ldots, u_k$.
  It holds that $\rho \trianglelefteq \sigma$
  if $\rho_i \trianglelefteq \sigma_i$ for all $i \in \{0, \ldots, k\}$ and
  $t_i \sqsubseteq u_i$ for all $i \in \{1, \ldots, k\}$.
\end{corollary}

\begin{proof}
  We have
  $\src{\rho} = \src{\rho_0} \leq \src{\sigma_0} = \src{\sigma}$ as $\rho_0 \trianglelefteq \sigma_0$.
  Similarly,
  $\tgt{\rho} = \tgt{\rho_k} \leq \tgt{\sigma_k} = \tgt{\sigma}$ as $\rho_k \trianglelefteq \sigma_k$.
  Let us show that $\enc{\rho} \preceq^* \enc{\sigma}$.
  Observe that
  \begin{align*}
    \enc{\rho} & = \enc{\rho_0} \, (\tgt{\rho_0}, t_1, \src{\rho_1}) \, \enc{\rho_1} \cdots (\tgt{\rho_{k-1}}, t_k, \src{\rho_k}) \, \enc{\rho_k}
    \\
    \enc{\sigma} & = \enc{\sigma_0} \, (\tgt{\sigma_0}, u_1, \src{\sigma_1}) \, \enc{\sigma_1} \cdots (\tgt{\sigma_{k-1}}, u_k, \src{\sigma_k}) \, \enc{\sigma_k}
  \end{align*}
  According to \cref{lem:wsvas-run-order-charac},
  we have
  $\src{\rho_i} \leq \src{\sigma_i}$,
  $\tgt{\rho_i} \leq \tgt{\sigma_i}$ and
  $\enc{\rho_i} \preceq^* \enc{\sigma_i}$
  for all $i \in \{0, \ldots, k\}$.
  This entails in particular that
  $(\tgt{\rho_{i-1}}, t_i, \src{\rho_i}) \preceq (\tgt{\sigma_{i-1}}, u_i, \src{\sigma_i})$
  for all $i \in \{1, \ldots, k\}$.
  We derive that
  $\enc{\rho} \preceq^* \enc{\sigma}$.
  The corollary then follows from \cref{lem:wsvas-run-order-charac}.
\end{proof}

\begin{lemma}
  \label{lem:wsvas-amalgamation}
  For every runs $\rho, \alpha, \beta$ such that $\rho \trianglelefteq \alpha, \beta$,
  there exists a run $\sigma$ such that
  $\alpha, \beta \trianglelefteq \sigma$ and
  $\dir{\rho} + \dir{\sigma} = \dir{\alpha} + \dir{\beta}$.
\end{lemma}

\begin{proof}
  We use some additional notations in this proof.
  Given a run $\rho = (\vec{c}_0, t_1, \vec{c}_1, \ldots, t_k, \vec{c}_k)$ and
  a vector $\vec{x} \in \setN^d$,
  we let $\uprun{\rho}{\vec{x}}$ denote the alternating sequence
  $(\vec{c}_0 + \vec{x}, t_1, \vec{c}_1 + \vec{x}, \ldots, t_k, \vec{c}_k + \vec{x})$.
  It is readily seen that $\uprun{\rho}{\vec{x}}$ is also a run and
  that $\rho \trianglelefteq \uprun{\rho}{\vec{x}}$.
  Given two runs $\rho, \rho'$ such that
  $\tgt{\rho} = \src{\rho'}$,
  we let $\rho \runcat \rho'$ denote the run obtained
  by concatenating $\rho$ and $\rho'$ in the obvious way
  (the target of $\rho$ and the source of $\rho'$ are merged together).
  Note that $\enc{\rho}, \enc{\rho'} \preceq^* \enc{\rho \runcat \rho'}$.
  Let us now prove the lemma.

  \smallskip

  Consider a run $\rho = (\vec{c}_0, t_1, \vec{c}_1, \ldots, t_k, \vec{c}_k)$ and
  let $\alpha, \beta$ be runs such that $\rho \trianglelefteq \alpha, \beta$.
  By definition of $\trianglelefteq$, it holds that
  $\alpha = \alpha_0 u_1 \alpha_1 \cdots u_k \alpha_k$
  and
  $\beta = \beta_0 v_1 \beta_1 \cdots v_k \beta_k$ for
  some runs $\alpha_0, \ldots, \alpha_k, \beta_0, \ldots, \beta_k$ and
  some transitions $u_1, \ldots, u_k, v_1, \ldots, v_k$
  such that
  $\vec{c}_i \leq \src{\alpha_i}, \src{\beta_i}, \tgt{\alpha_i}, \tgt{\beta_i}$ for all $i \in \{0, \ldots, k\}$ and
  $t_i \sqsubseteq u_i, v_i$ for all $i \in \{1, \ldots, k\}$.
  Since $\delta$ is a diamond map from $(T, \sqsubseteq)$ to $(\setZ^d, \leq)$,
  there exist $w_1, \ldots, w_k \in T$ such that
  $u_i, v_i \sqsubseteq w_i$ and $\delta(t_i) + \delta(w_i) = \delta(u_i) + \delta(v_i)$
  for all $i \in \{1, \ldots, k\}$.
  To construct the desired run $\sigma$,
  we introduce for each $i \in \{0, \ldots, k\}$
  the vectors $\vec{f}_i, \vec{g}_i, \vec{h}_i \in \setN^d$ and
  the runs $\sigma_i, \sigma'_i$ defined by:
  \begin{align*}
    \vec{f_i} & = \src{\alpha_i} + \src{\beta_i} - \vec{c}_i
    &
    \sigma_i & = \uprun{\alpha_i}{(\src{\beta_i} - \vec{c}_i)}
    \\
    \vec{g_i} & = \tgt{\alpha_i} + \src{\beta_i} - \vec{c}_i
    &
    \sigma'_i & = \uprun{\beta_i}{(\tgt{\alpha_i} - \vec{c}_i)}
    \\
    \vec{h_i} & = \tgt{\alpha_i} + \tgt{\beta_i} - \vec{c}_i
  \end{align*}
  Observe that
  $\src{\sigma_i} = \vec{f_i}$,
  $\tgt{\sigma_i} = \src{\sigma'_i} = \vec{g_i}$ and
  $\tgt{\sigma'_i} = \vec{h_i}$.
  Moreover,
  it is routinely checked that
  $\vec{f}_i = \vec{h}_{i-1} + \delta(w_i)$
  for all $i \in \{1, \ldots, k\}$.
  We derive that the alternating sequence
  $$
  \sigma \ = \ 
  (\sigma_0 \runcat \sigma'_0) \, w_1 \, (\sigma_1 \runcat \sigma'_1) \cdots \, w_k \, (\sigma_k \runcat \sigma'_k)
  $$
  is a run.
  Moreover,
  we have
  $\src{\sigma} = \vec{f}_0 = \src{\alpha_0} + \src{\beta_0} - \vec{c}_0 = \src{\alpha} + \src{\beta} - \src{\rho}$ and
  $\tgt{\sigma} = \vec{h}_k = \tgt{\alpha_k} + \tgt{\beta_k} - \vec{c}_k = \tgt{\alpha} + \tgt{\beta} - \tgt{\rho}$.
  This entails that
  $\dir{\rho} + \dir{\sigma} = \dir{\alpha} + \dir{\beta}$.
  It remains to show that $\alpha, \beta \trianglelefteq \sigma$.
  We will use \cref{cor:wsvas-run-order-decomp} to do so.

  \smallskip

  Let $i \in \{0, \ldots, k\}$ and let us show that $\alpha_i, \beta_i \trianglelefteq (\sigma_i \runcat \sigma'_i)$.
  The definitions of $\sigma_i$ and $\sigma'_i$ immediately entail that
  $\alpha_i \trianglelefteq \sigma_i$ and $\beta_i \trianglelefteq \sigma'_i$,
  hence,
  $\enc{\alpha_i} \preceq^* \enc{\sigma_i}$ and $\enc{\beta_i} \preceq^* \enc{\sigma'_i}$
  by \cref{lem:wsvas-run-order-charac}.
  It follows that $\enc{\alpha_i}, \enc{\beta_i} \preceq^* (\sigma_i \runcat \sigma'_i)$
  since $\enc{\sigma_i}, \enc{\sigma'_i} \preceq^* (\sigma_i \runcat \sigma'_i)$.
  Notice that
  $\src{\sigma_i \runcat \sigma'_i} = \vec{f_i} \geq \src{\alpha_i}, \src{\beta_i}$ and
  $\tgt{\sigma_i \runcat \sigma'_i} = \vec{h_i} \geq \tgt{\alpha_i}, \tgt{\beta_i}$.
  We derive,
  again by \cref{lem:wsvas-run-order-charac},
  that $\alpha_i, \beta_i \trianglelefteq (\sigma_i \runcat \sigma'_i)$.

  \smallskip

  We have shown in the previous paragraph that
  $\alpha_i, \beta_i \trianglelefteq (\sigma_i \runcat \sigma'_i)$
  for all $i \in \{0, \ldots, k\}$.
  Recall that
  $u_i, v_i \sqsubseteq w_i$ for all $i \in \{1, \ldots, k\}$.
  It follows from \cref{cor:wsvas-run-order-decomp} that
  $\alpha, \beta \trianglelefteq \sigma$.
\end{proof}

\end{document}